%% file: main.tex
\documentclass[12pt]{article}

\input{macros}

\newenvironment{sciabstract}{%
\begin{quote} \bf}
{\end{quote}}

\title{Uncovering the Higher-Order Structure of Social Systems}
\title{A Higher-Order Lens for Social Systems}

\author{
Giulia Preti,$^{1\ast}$ 
Adriano Fazzone,$^{1}$\\ 
Giovanni Petri,$^{2,1}$
Gianmarco De Francisci Morales$^{1}$\\
\\
\normalsize{$^{1}$CENTAI Institute}\\
\normalsize{Corso Inghilterra 3, 10138, Turin, Italy}\\
\normalsize{$^{2}$Network Science Institute, Northeastern University London}\\
\normalsize{Devon House, 58 St Katharine's Way, London, E1W 1LP, UK}
\\
\normalsize{$^\ast$Corresponding author; E-mail: giulia.preti@centai.eu}
}

\date{}

\begin{document}


\baselineskip24pt


\maketitle


\begin{sciabstract}
Despite the widespread adoption of higher-order mathematical structures su\-ch as hypergraphs, methodological tools for their analysis lag behind those for traditional graphs.
This work addresses a critical gap in this context by proposing two micro-canonical random null models for directed hypergraphs: the \nameAlong (\nameA) and the \nameClong (\nameC).
These models preserve essential structural properties of directed hypergraphs such as node in- and out-degree sequences and hyperedge head and tail size sequences, or their joint tensor.
We also describe two efficient MCMC algorithms, \algoA and \algoC, to sample random hypergraphs from these ensembles.

To showcase the interdisciplinary applicability of the proposed null models, we present three distinct use cases in sociology, epidemiology, and economics.
First, we reveal the oscillatory behavior of increased homophily in opposition parties in the US Congress over a 40-year span, emphasizing the role of higher-order structures in quantifying political group homophily.
Second, we investigate non-linear contagion in contact hyper-networks, demonstrating that disparities between simulations and theoretical predictions can be explained by considering higher-order joint degree distributions.
Last, we examine the economic complexity of countries in the global trade network, showing that local network properties preserved by \algo explain the main structural economic complexity indexes.

This work pioneers the development of null models for directed hypergraphs, addressing the intricate challenges posed by their complex entity relations, and providing a versatile suite of tools for researchers across various domains.
\end{sciabstract}

\input{sections/intro}
\input{sections/methods}
\input{sections/exp}
\input{sections/discussion}
\input{sections/algos}

\bibliographystyle{Science}
\bibliography{refs}

\clearpage

\startcontents[appendix]
\printcontents[appendix]{}{1}{\setcounter{tocdepth}{2}}

\appendix
\addcontentsline{toc}{section}{Supplementary Material}

\input{sections/proofs}
\input{sections/algodetails}
\input{sections/data}
\input{sections/convergence}
\input{sections/appendix}

\end{document}

%% file: macros.tex
\usepackage{amsmath}
\usepackage{amsthm}
\usepackage{amsfonts}
\usepackage{mathtools}
\usepackage{hyperref}
\usepackage[capitalise]{cleveref}
\usepackage{algorithm}%
\usepackage[noend]{algpseudocode}%
\usepackage{xfrac} 
\usepackage{xspace} 
\usepackage{booktabs}
\usepackage{multirow} 
\usepackage{url}
\usepackage{caption}
\usepackage{subcaption}
\usepackage{MnSymbol}
\usepackage{graphicx}
\usepackage{siunitx}
\usepackage{scicite}
\usepackage{times}
\usepackage{tikz} 
\usetikzlibrary{shapes.geometric}
\usepackage{xcolor}
\usepackage{xfrac} 
\usepackage{titletoc} 

\graphicspath{{figures}}

\newtheorem{theorem}{Theorem}
\newtheorem{definition}{Definition}
\newtheorem{lemma}{Lemma}

\newtheorem{corollary}{Corollary}

\newcommand{\spara}[1]{\smallskip\noindent{\bf #1}}
\newcommand{\mpara}[1]{\medskip\noindent{\bf #1}}

\newcommand{\indeg}[2]{\ensuremath{\mathsf{ideg}_{#1}(#2)}\xspace} 
\newcommand{\outdeg}[2]{\ensuremath{\mathsf{odeg}_{#1}(#2)}\xspace} 
\newcommand{\indegG}[1]{\ensuremath{\mathsf{ideg}(#1)}\xspace} 
\newcommand{\outdegG}[1]{\ensuremath{\mathsf{odeg}(#1)}\xspace} 
\newcommand{\sqindeg}[1]{\ensuremath{\mathsf{ideg}^2(#1)}\xspace} 
\newcommand{\sqoutdeg}[1]{\ensuremath{\mathsf{odeg}^2(#1)}\xspace} 
\newcommand{\neighsplus}[2]{\ensuremath{\overrightarrow{\Gamma_{#1}}(#2)}\xspace} 
\newcommand{\neighsminus}[2]{\ensuremath{\overleftarrow{\Gamma_{#1}}(#2)}\xspace} 
\newcommand{\outneighs}[1]{\ensuremath{\overrightarrow{\Gamma}(#1)}\xspace} 
\newcommand{\inneighs}[1]{\ensuremath{\overleftarrow{\Gamma}(#1)}\xspace} 
\newcommand{\degG}{\ensuremath{\mathsf{d}(G)}\xspace}
\newcommand{\degGR}{\ensuremath{\mathsf{d}^\mathrm{R}(G)}\xspace}

\newcommand{\nameA}{\textsc{DHCM}\xspace} 
\newcommand{\nameAlong}{Directed Hypergraph Configuration Model\xspace} 
\newcommand{\nameC}{\textsc{DHJM}\xspace} 
\newcommand{\nameClong}{Directed Hypergraph JOINT Model\xspace} 
\newcommand{\nullmodel}{\ensuremath{\Pi}\xspace} 
\newcommand{\nullmodelA}{\ensuremath{\Pi^{\nameA}}\xspace} 
\newcommand{\nullmodelC}{\ensuremath{\Pi^{\nameC}}\xspace} 
\newcommand{\nullprob}{\ensuremath{\pi}\xspace} 
\newcommand{\nullset}{\ensuremath{\mathcal{Z}}\xspace} 
\newcommand{\neighdistr}[1]{\ensuremath{\xi_{#1}}\xspace} 
\newcommand{\states}{\ensuremath{\mathcal{G}}\xspace} 
\newcommand{\observed}{\ensuremath{\mathring{H}}\xspace} 
\newcommand{\observedbip}{\ensuremath{\mathring{G}}\xspace} 
\newcommand{\biot}[1]{\ensuremath{\mathsf{T}_{#1}}\xspace} 
\newcommand{\biotG}{\ensuremath{\mathsf{T}}\xspace} 
\newcommand{\maxi}{\ensuremath{\mathsf{IN}}\xspace} 
\newcommand{\maxo}{\ensuremath{\mathsf{OUT}}\xspace} 
\newcommand{\suchthat}{\ensuremath{\mathrel{:}}\xspace} 
\newcommand{\typeB}[1]{\ensuremath{\mathsf{z}^B(#1)}\xspace} 
\newcommand{\butterfly}[1]{\ensuremath{\mathsf{b}(#1)}\xspace} 
\newcommand{\pso}{\ensuremath{\xrightarrow{\mathrm{PSO}}}\xspace} 
\newcommand{\rpso}{\ensuremath{\xrightarrow{\mathrm{RPSO}}}\xspace} 
\newcommand{\fromL}{\ensuremath{\overrightarrow{D}}\xspace} 
\newcommand{\fromR}{\ensuremath{\overleftarrow{D}}\xspace} 
\newcommand{\posL}{\ensuremath{\overrightarrow{L}}\xspace} 
\newcommand{\posR}{\ensuremath{\overleftarrow{R}}\xspace} 
\newcommand{\setdiff}[2]{\ensuremath{\mathsf{X}^+_{#1, #2}}\xspace}
\newcommand{\setidiff}[2]{\ensuremath{\mathsf{X}^-_{#1, #2}}\xspace}
\newcommand{\sizeof}[1]{\ensuremath{\left\lvert{#1}\right\rvert}\xspace}
\newcommand{\slfrac}[2]{\raisebox{5pt}{\ensuremath{\left.\raisebox{5pt}{\ensuremath{\displaystyle#1}}\middle/\raisebox{-5pt}{\ensuremath{\displaystyle#2}}\right.}}}
\newcommand{\dataset}{\ensuremath{\datasetsym}\xspace} 
\newcommand{\datasetsym}{\ensuremath{\mathcal{D}}\xspace} 

\newcommand{\supp}[2]{\ensuremath{\sigma_{#1}(#2)}\xspace} 
\newcommand{\algoA}{\textsc{NuDHy-Degs}\xspace}
\newcommand{\algoB}{\textsc{NuDHy-Degs-MH}\xspace}
\newcommand{\algoC}{\textsc{NuDHy-JOINT}\xspace}
\newcommand{\algoD}{\textsc{NuDHy-JOINT-MH}\xspace}
\newcommand{\algo}{\textsc{NuDHy}\xspace}
\newcommand{\redi}{\textsc{ReDi}\xspace}
\newcommand{\base}{\textsc{Base}\xspace}
\newcommand{\based}{\textsc{BaseD}\xspace}
\newcommand{\nullm}{\textsc{Null}\xspace}

\newenvironment{squishlist}
{\begin{list}{$\bullet$}
 {\setlength{\itemsep}{0pt}
     \setlength{\parsep}{3pt}
     \setlength{\topsep}{3pt}
     \setlength{\partopsep}{0pt}
     \setlength{\leftmargin}{1.5em}
     \setlength{\labelwidth}{1em}
     \setlength{\labelsep}{0.5em} } }
{\end{list}}

\hypersetup{
    colorlinks=true,
    linkcolor={red!50!black},
    filecolor={green!50!black},
    citecolor={green!50!black}, 
    urlcolor={blue!80!black},
}

\algdef{SE}[REPEATN]{RepeatN}{End}[1]{\algorithmicrepeat\ #1 \textbf{times}}{\algorithmicend}

\makeatletter
\ifthenelse{\equal{\ALG@noend}{t}}%
  {\algtext*{End}}
  {}%
\makeatother

\definecolor{hedgegc}{HTML}{d4fb79}
\definecolor{hedgepc}{HTML}{ff8ad8}
\definecolor{hedgeoc}{HTML}{ff9300}
\definecolor{black}{HTML}{2a2b2a}

\newcommand{\hedgeg}{%
  \tikz[baseline=-0.5ex]\node[regular polygon, regular polygon sides=8, draw=black, fill=hedgegc, inner sep=2pt, minimum size=4mm]{};}

\newcommand{\hedgeo}{%
  \tikz[baseline=-0.5ex]\node[regular polygon, regular polygon sides=8, draw=black, fill=hedgeoc, inner sep=2pt, minimum size=4mm]{};}

%% file: sections/intro.tex
\section{Introduction}\label{sec:intro}

Higher-order mathematical structures such as hypergraphs and simplicial complexes have emer\-ged as powerful modeling tools that overcome the limitations of traditional graph models, which by construction are restricted to binary relations between entities \cite{battiston2020networks,battiston2021physics,bick2023higher}.
Indeed, their adoption is motivated by the observation that real-world scenarios often entail interactions among multiple entities simultaneously.
Examples span systems across multiple spatial and temporal scales, including cellular processes~\cite{ritz2015pathway}, protein interaction networks~\cite{feng2021hypergraph}, neural processing \cite{schneidman2003network,schneidman2006weak}, whole-brain activity \cite{giusti2015clique,petri2014homological}, co-authorship networks~\cite{patania2017shape,luo2022toward}, and contact networks~\cite{billings2019simplex2vec}.
Hypergraphs, in particular, are a natural and flexible generalization of graphs that model arbitrary $q$-ary relations among entities.
Directed hypergraphs further extend this concept by representing a link from a set of nodes (the \emph{head} of the hyperedge) to another set of nodes (its \emph{tail}).
Consider, for instance, the case of citations among scientific publications. 
In this case, each citation in a publication can be modeled as a directed hyperedge from the set of authors of the publication to the set of authors of the cited work.
The application of hypergraphs already spans diverse domains, from forecasting urban traffic~\cite{luo2022directed} and modeling Bitcoin transactions~\cite{ranshous2017exchange} to representing web structures for accurate page reputation scoring~\cite{berlt2010modeling}.
However, the current methodological tools for hypergraphs lag behind their counterparts in the graph world.

Understanding complex networks often involves comparing observed structures against mo\-dels that mimic random scenarios. 
Originating from Fisher's groundwork in hypothesis testing~\cite{fisher1936design}, this methodology has expanded into graph theory with the study of random graph null models~\cite{manly1995note}. 
These models define graph ensembles that retain only selected features of the observed graph while being random in any other respect~\cite{squartini2015unbiased}.
They are key tools in graph theory because they allow us to assess the significance of the observed properties of real-world networks, by comparing them to those obtained from randomly generated graphs \cite{schlauch2015influence}.
This comparative analysis unveils the influence of local node features versus additional factors on network properties, and aids in identifying structural irregularities within the networks\cite{fischer2015sampling}.
Furthermore, it enables assessing the role of specific properties in the presence of specific empirically-observed topological and structural features.

Akin to any hypothesis test, the selection of topological features to preserve in these ensembles significantly influences the conclusions drawn from the analyses.
Common approaches preserve the degree sequence~\cite{verhelst2008efficient,strona2014fast} and the joint degree sequence~\cite{stanton2012constructing,gjoka2015construction}.
Random graph ensembles can be categorized into two fundamental families: micro-canonical and canonical~\cite{cimini2019statistical}.
Micro-canonical ensembles preserve the properties in a `hard' fashion, i.e., each of the graphs in the ensemble satisfies the imposed constraints.
Conversely, canonical ensembles preserve the properties in a `soft' fashion: they maintain the constraints in expectation across the graphs in the ensemble.
The choice between these approaches should be based on principled criteria, considering factors such as the characteristics of the observed data.
Canonical ensembles, for instance, are better suited for scenarios where data may contain measurement errors or noise since they maintain constraints on an average basis.

Despite a vast literature on canonical and micro-canonical graph ensembles~\cite{kannan1999simple,tabourier2011generating,strona2014fast,saracco2015randomizing,squartini2015unbiased,boroojeni2017generating,aksoy2017measuring,del2010efficient}, little attention has been devoted to defining null models for directed hypergraphs and developing efficient sampling algorithms for their corresponding ensembles.
Existing work in the realm of hypergraphs predominantly focuses on configuration models for undirected hypergraphs~\cite{saracco2022entropy,do2020structural,barthelemy2022class,guo2016non,wang2010evolving,chodrow2020configuration,zeng2023hyper}, introduces max entropy models~\cite{sun2021higher}, or generalizes the concept of dK-series to undirected hypergraphs~\cite{nakajima2021randomizing,miyashita2023randomizing}.

Transitioning to developing null models for directed hypergraphs brings unique challenges due to their intricate entity relations, characterized by a broader set of properties---and thus constraints. 
Parameters such as the number of nodes, number of hyperedges, head and tail size sequences, and the frequency of nodes within hyperedge heads or tails should be taken into consideration when defining these models.
Recently, Kim et al.~\cite{kim2022reciprocity} proposed two samplers for generating directed hypergraphs in the canonical ensemble with prescribed head and tail size sequences. 
However, due to certain design choices aimed at improving efficiency, the generated hypergraphs often exhibit structural dissimilarities from the real-world ones.

This work proposes two micro-canonical null models for directed hypergraphs.
The first model, \nameAlong (\nameA), preserves the in- and out-degree sequences of the nodes, as well as the head-size and tail-size sequences of the hyperedges.
The second model, called \nameClong (\nameC), preserves the \emph{joint out-in degree tensor}, which encodes information about the in- and out-degree of the nodes involved in hyperedges of specific head and tail sizes.
We also describe two samplers, \algoA and \algoC, to efficiently draw random hypergraphs from the corresponding ensembles.
Both samplers are Markov Chain Monte Carlo algorithms based on Metropolis-Hastings and employ targeted shuffling operations for traversal within the Markov graph.

We demonstrate the wide interdisciplinary applicability of the proposed suite of null models by showcasing three distinct use cases in sociology, epidemiology, and economics, respectively.
The first one shows the role of higher-order structures in quantifying genuine political group homophily by uncovering an oscillatory behavior of increased homophily in opposition parties in the US Congress across a 40-years span. 
The second one focuses on nonlinear contagion in contact hyper-networks, demonstrating that the disparities observed between simulations in the hyper-networks and theoretical predictions can be explained when considering higher-order joint degree distribution, thus shedding some light on the underlying mechanisms governing these phenomena.
The third and final one studies the economic complexity of countries in the global trade network, and shows that the main structural economic complexity indexes~\cite{hidalgo2007product,tacchella2012metrics,sciarra2020reconciling} can be almost entirely explained by local properties of the network preserved by \algo.
A more comprehensive evaluation of \algo with respect to other existing null models and related samplers is provided in \Cref{ax:additional}. 

%% file: sections/methods.tex
\section{Null Models for Weighted Directed Hypergraphs}\label{sec:nulls}

We consider weighted directed hypergraphs of the form $H \doteq (V,E)$, where $V = \{v_1,\ldots, v_n\}$ is a set of nodes and $E = \{e_1,\ldots,e_m\}$ is a multi-set of \emph{directed hyperedges} where the multiplicity of each hyperedge represents its weight.
Each hyperedge $e \doteq (h, t) \in E$ consists of a \emph{head} $h$ and a \emph{tail} $t$ such that $h,t \subseteq V$. 
The \emph{size} of $e$ is the sum of the sizes of its head and tail, $\sizeof{e} = \sizeof{h} + \sizeof{t}$.
The \emph{in-degree} of a node $v$ in $H$, denoted as $\indeg{H}{v}$, is the number of tails that contain $v$;
the \emph{out-degree} of $v$ in $H$, denoted as $\outdeg{H}{v}$, is the number of heads that contain $v$.

\begin{figure}[th!]
	\centering
	\includegraphics[width=\columnwidth]{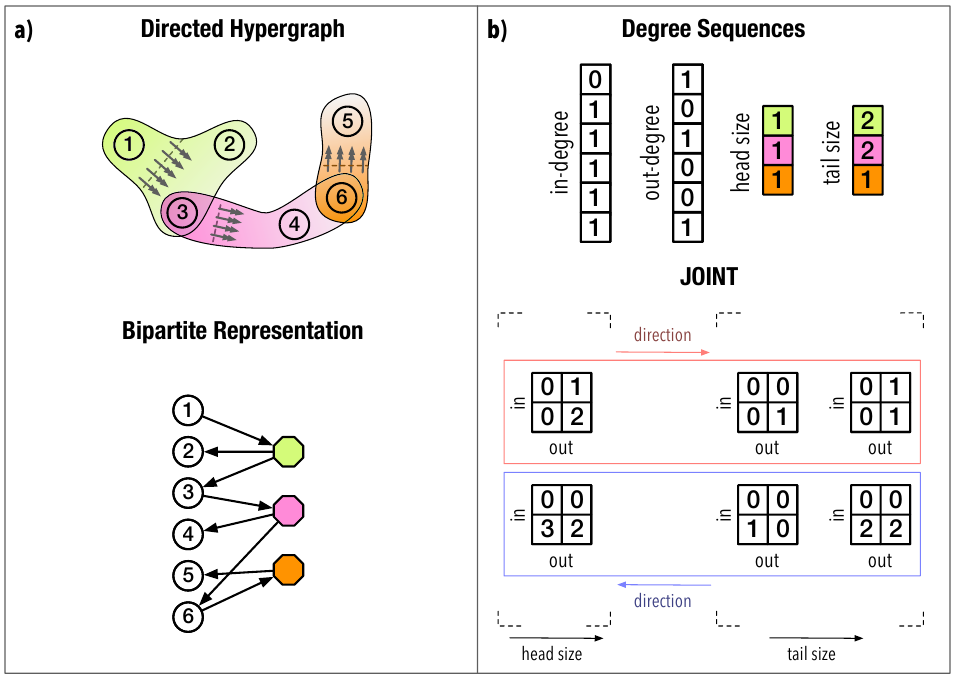}
	\caption{
	\textbf{Construction of directed hypergraph configuration models.}
	\textbf{a)} A directed hypergraph (top) and its representation as a bipartite graph (bottom). 
	The left vertices (circles) correspond to hypergraph nodes, while the right vertices (hexagons) correspond to hyperedges. 
	Dotted lines in the directed hypergraph separate the head and tail of each hyperedge, with arrows pointing towards the tail.
	\textbf{b)} The characteristics of the observed hypergraph preserved by \nameA and \nameC: left and right in- and out-degree sequences (top), and JOINT (bottom). The right in-degree sequence corresponds to the head-size sequence, while the right out-degree sequence corresponds to the tail-size sequence.}
	\label{fig:toy}
\end{figure}

A weighted directed hypergraph $H$ can be represented, without loss of information, as a directed bipartite graph $G \doteq (L, R, D)$, where $L = V$ (\emph{left} vertices), $R = E$ (\emph{right} vertices),\footnote{For clarity, we refer to nodes when talking about the elements of the hypergraph and to vertices when talking about the elements of the bipartite graph.}
and $D$ is a set of triplets defined as follows:
\begin{flalign*}
	&\forall e \doteq (h,t) \in E, \forall v \in h \implies (v,e,+1) \in D;& \\
	&\forall e \doteq (h,t) \in E, \forall v \in t \implies (v,e,-1) \in D.&
\end{flalign*}
Each triplet $(v,e,d)$ is a \emph{directed edge} involving a node $v$ and a hyperedge $e$, where $d$ denotes the \emph{direction} of the edge: $+1$ indicates that the edge goes from a left vertex to a right vertex whereas $-1$ indicates the opposite direction.
We denote with $\fromL$ the set of pairs of vertices connected by an edge with direction $d = +1$, i.e., $(v, \alpha) \in \fromL \iff (v, \alpha, +1) \in D$.
Similarly, we denote with $\fromR$ the set of pairs of vertices connected by an edge with opposite direction $d = -1$. 
For any vertex $v \in L$, we denote with $\neighsplus{G}{v}$ the set of vertices $\alpha \in R$ such that $(v,\alpha) \in \fromL$, and with $\neighsminus{G}{v}$ the set of vertices $\alpha \in R$ such that $(v,\alpha) \in \fromR$. 
The size of $\neighsplus{G}{v}$ is called the \emph{out-degree} of $v$, while the size of $\neighsminus{G}{v}$ is the \emph{in-degree} of $v$.
Similarly, we can define the in-degree (resp. out-degree) of a vertex $\alpha \in R$ as the size of the set of vertices $v \in L$ such that $(v,\alpha) \in \fromR$ (resp. $(v,\alpha) \in \fromL$).
\Cref{fig:toy}\textbf{a} shows an example of a directed hypergraph and the corresponding bipartite graph.

To encode the information of both the in- and out-degree of the vertices connected by the edges in $G$, we define the bipartite \emph{Joint Out-In degree Tensor (JOINT) $\biot{G}$}.
\begin{definition}[JOINT]\label{def:biot}
Let $G \doteq (L, R, D)$ be a directed bipartite graph, and $\maxi_L = \max\limits_{v \in L}{\sizeof{\neighsminus{G}{v}}}$ and $\maxo_L = \max\limits_{v \in L}{\sizeof{\neighsplus{G}{v}}}$ be the largest in-/out-degree of a vertex in $L$, respectively.
$\maxi_R$ and $\maxo_R$ are similarly defined for $R$.
The bipartite \emph{Joint Out-In degree Tensor (JOINT) $\biot{G}$ of $G$} is a 5-dimensional tensor with size $\maxi_L+1 \times \maxo_L+1 \times \maxi_R+1 \times \maxo_R+1 \times 2$, and whose $(i,j,k,l,d)$-th entry $\biot{G}[i,j,k,l,d]$ for $i \in \left[0,\maxi_L\right]$, $j \in \left[0,\maxo_L\right]$, $k \in \left[0,\maxi_R\right]$, $l \in \left[0,\maxo_R\right]$, and $d \in \{+1,-1\}$, is the number of edges with direction $d$ connecting a left vertex with in-degree $i$ and out-degree $j$ and a right vertex with in-degree $k$ and out-degree $l$, i.e., 
\begin{align*}
	& \biot{G}[i,j,k,l,d] \doteq \sizeof{ \left\{(v,\alpha,d) \in D \suchthat \sizeof{\neighsminus{G}{v}} = i \wedge \sizeof{\neighsplus{G}{v}} = j \wedge \sizeof{\neighsplus{G}{\alpha}} = k \wedge \sizeof{\neighsminus{G}{\alpha}} = l \right\} }.
\end{align*}
\end{definition}

\spara{Null model.}
Let $\mathcal{P}$ be a set of properties of an observed hypergraph \observed.
A \emph{null model} $\nullmodel \doteq (\nullset, \nullprob)$ is a tuple where
\nullset is the set of all the hypergraphs where each $P$ in $\mathcal{P}$ holds (i.e., the ensemble of hypergraphs that preserve these properties), and \nullprob is a probability distribution over \nullset. 

The first null model proposed, called \nameAlong (\nameA) and denoted as $\nullmodelA \doteq (\nullset^{\nameA}, \nullprob)$, preserves the following four properties:
\begin{squishlist}
	\item[\textbf{P1:}] head-size sequence $\left[|h_1|, \dotsc, |h_m|\right]$;
	\item[\textbf{P2:}] tail-size sequence $\left[|t_1|, \dotsc, |t_m|\right]$;
	\item[\textbf{P3:}] in-degree sequence $\left[\indeg{\observed}{v_1}, \dotsc, \indeg{\observed}{v_n}\right]$;
	\item[\textbf{P4:}] out-degree sequence $\left[\outdeg{\observed}{v_1}, \dotsc, \outdeg{\observed}{v_n}\right]$.
\end{squishlist}

Each $H \in \nullset^{\nameA}$ has the same head-size, tail-size, in-degree, and out-degree sequences of $\observed$.
Preserving \textbf{P1} and \textbf{P2} is equivalent to preserving the sequences of the out- and in-degrees of the vertices in $R$ in the bipartite graph representation \observedbip of \observed, and automatically preserves the sequence of the sizes of the hyperedges in \observed.
Preserving \textbf{P3} and \textbf{P4} corresponds to preserving the sequences of the in- and out-degrees of the vertices in $L$ in $\observedbip$, and automatically preserves the number of times each node is contained in a tail and a head of a hyperedge in \observed.
The in-degree, out-degree, head-size, and tail-size sequences of the directed hypergraph in \Cref{fig:toy}$\textbf{a}$ are illustrated in \Cref{fig:toy}$\textbf{b}$.

The {\nameA} can be regarded as a specific instance of the annotated hypergraph configuration model~\cite{chodrow2020annotated}, wherein the input is a degenerate hypergraph. In these hypergraphs, each node can assume multiple roles, which in our context, manifests as a node occupying both head and tail positions within a hyperedge.

The second null model proposed, called \nameClong (\nameC) and denoted as $\nullmodelC \doteq (\nullset^{\nameC}, \nullprob)$, preserves the following property:
\begin{squishlist}
	\item[\textbf{P5:}] JOINT $\biot{\observedbip}$.
\end{squishlist}

Preserving \textbf{P5} also preserves \textbf{P1}-\textbf{P4}.

In fact, for every $\bar{a} \in \left[0,\maxi_R\right]$, $\bar{b} \in \left[0,\maxo_R\right]$, $\bar{c} \in \left[0,\maxi_L\right]$, $\bar{d} \in \left[0,\maxo_L\right]$, 
it holds
\begin{flalign*}
	&\text{\textbf{P5} preserves \textbf{P1}:} \quad
	\sizeof{ \left\{ \alpha \in R \suchthat \indeg{\observed}{\alpha} = \bar{a} \right\} } = \sfrac{1}{\bar{a}} \sum_{i=0}^{\phantom{.}\maxi_L\phantom{.}} \sum_{j=0}^{\maxo_L} \sum_{l=0}^{\maxo_R} \biot{\observedbip}[i,j,\bar{a},l,+1] \\
	&\text{\textbf{P5} preserves \textbf{P2}:} \quad
	\sizeof{ \left\{ \alpha \in R \suchthat \outdeg{\observed}{\alpha} = \bar{b} \right\} } = \sfrac{1}{\bar{b}} \sum_{i=0}^{\phantom{.}\maxi_L\phantom{.}} \sum_{j=0}^{\maxo_L} \sum_{k=0}^{\phantom{.}\maxi_R\phantom{.}} \biot{\observedbip}[i,j,k,\bar{b},-1] \\
	&\text{\textbf{P5} preserves \textbf{P3}:} \quad
	\sizeof{ \left\{ v \in L \suchthat \indeg{\observed}{v} = \bar{c} \right\} } = \sfrac{1}{\bar{c}} \sum_{j=0}^{\maxo_L} \sum_{k=0}^{\phantom{.}\maxi_R\phantom{.}} \sum_{l=0}^{\maxo_R}\biot{\observedbip}[\bar{c},j,k,l,-1] \\
	&\text{\textbf{P5} preserves \textbf{P4}:} \quad
	\sizeof{ \left\{ v \in L \suchthat \outdeg{\observed}{v} = \bar{d} \right\} } = \sfrac{1}{\bar{d}} \sum_{i=0}^{\phantom{.}\maxi_L\phantom{.}} \sum_{k=0}^{\phantom{.}\maxi_R\phantom{.}} \sum_{l=0}^{\maxo_R} \biot{\observedbip}[i,\bar{d},k,l,+1]
\end{flalign*}

To simplify the visualization of the JOINT of the directed hypergraph in \Cref{fig:toy}$\textbf{a}$, \Cref{fig:toy}$\textbf{b}$ illustrates
\emph{(i)} for each edge direction $d$ (differentiated using arrows colored differently and pointing to the two directions) and head size $k$ (there is only the head size $1$), the 2-dimensional array of size $\maxi_L+1 \times \maxo_L+1$, whose $(i,j)$-th entry indicates the number of edges with direction $d$ connecting left vertices with in-degree $i$ and out-degree $j$ to right vertices with in-degree $k$;
and \emph{(ii)} for each edge direction $d$ and tail size $l$ (there are two tail sizes, i.e., $1$ and $2$), the 2-dimensional array of size $\maxi_L+1 \times \maxo_L+1$, whose $(i,j)$-th entry indicates the number of edges with direction $d$ connecting left vertices with in-degree $i$ and out-degree $j$ to right vertices with out-degree $l$.
The $(0,0)$ entry of each array is in the upper-left corner. 
In the example of \Cref{fig:toy}, there are $3$ left vertices with in-degree $1$ and out-degree $0$, i.e., $2$, $4$, and $5$, each of which has $1$ in-going edge from a right vertex with in-degree (and thus head size) $1$. Therefore, the bottom-left cell of the leftmost $2$-dimensional array for direction $\leftarrow$ contains the number $3$.
Similarly, there are $2$ left vertices with in-degree $1$ and out-degree $1$, i.e., $3$ and $6$, each of which with $1$ in-going edge from a right vertex with out-degree (and thus tail size) $2$. Therefore, the bottom-right cell of the rightmost $2$-dimensional array for direction $\leftarrow$ contains the number $2$.
For each direction $d$, the sum of the 2-dimensional arrays across head sizes equals the sum of 2-dimensional arrays across tail sizes. 
These summed arrays represent the number of edges of direction $d$ connected to left vertices with a specific in- and out-degree.

%% file: sections/exp.tex
\section{Results}

In this section, we present three distinct case studies that employ \algoA and \algoC, showcasing their versatility in analyzing various types of data models. 
While originally designed for generating random directed hypergraphs, these samplers extend their applicability to producing random undirected hypergraphs and (directed) bipartite graphs.
By conceptualizing an undirected hypergraph as a directed hypergraph where heads and tails coincide, \algoA produces undirected hypergraphs with prescribed node degree and hyperedge size distributions, while \algoC produces undirected hypergraphs with prescribed joint node degree and hyperedge size distributions. 
Moreover, by recognizing the lossless mapping between (directed) hypergraphs and (directed) bipartite graphs, \algoA and \algoC can produce random (directed) bipartite graphs with specified left and right degree sequences, and joint degree matrices.
The three case studies explore different domains, each utilizing a distinct data model. 
The first study delves into understanding group affinity within political parties through the analysis of sponsorship and co-sponsorship relations in the US Congress.
We reveal nuanced patterns that evade detection when solely examining unnormalized affinity values. 
The second study focuses on validating a recently proposed non-linear social contagion model for undirected hypergraphs, demonstrating how the JOINT can explain deviations from the theoretical framework in the observed data. 
Lastly, the third study investigates the impact of certain node properties preserved by our null models, namely degree and joint degree distribution, on the economic competitiveness of countries measures via metrics defined for bipartite country-to-product trade networks.
Here, we demonstrate that the JOINT adequately preserves rankings according to each measure of competitiveness considered. 
These case studies not only highlight the value of \algo as a lens but also yield valuable insights within each domain, thus enriching our understanding of these complex social systems.

\subsection{Group Affinity in Collaborative Hyper-Networks}\label{sec:app_congress}

The concept of \emph{homophily} describes an individuals' tendency to connect with those who share similar traits.
Previous studies have consistently found this inclination across various individual features, such as gender, age, ethnicity~\cite{moody2001race}, political views, and religious beliefs~\cite{loomis1946political}. 
From its origins in sociology~\cite{lazarsfeld1954friendship}, it later became a fundamental notion in network science, because of its natural relation to the connectedness of a system. 
Indeed, the focus of homophily research is to grasp how these similarities among individuals shape their network of interactions~\cite{mcpherson2001birds}.

Homophily can be extended to higher-order relations.
Called \emph{group affinity}~\cite{veldt2023combinatorial}, it measures the extent to which individuals in a certain class participate in groups with a certain number of individuals from the same class.
It offers insights into whether participation of an individual in a group is driven by a herding behavior conditional on trait similarity.

Here, we delve into the group affinity within the \emph{Republican} and \emph{Democratic} parties, known as \emph{partisanship}, using directed hypergraphs to represent sponsor-cosponsor relationships in Senate bills (\textsc{S-bills}) and House of Representatives bills (\textsc{H-bills}) from the $93^{\mathrm{rd}}$ to the $108^{\mathrm{th}}$ Congresses~\cite{fowler2006legislative}. 
We focus on bills and joint resolutions, given their potential to become law upon passage.
Each bill is represented as a directed hyperedge, with the bill's sponsor (the legislator who introduced the bill) forming the head, and the set of legislators supporting the bill as co-sponsors forming the tail.

In contrast to roll-call voting, where legislators must cast a vote, bill co-sponsorship data offers a nuanced view of collaboration behavior as they reflect voluntary expressions of interest in supporting specific bills, and reveal explicit cooperation that might not be fully captured in voting records.
Thus, co-sponsorship hyper-networks provide a rich account of legislative dynamics.
\Cref{tbl:data_senate_house} in \Cref{ax:data} reports some statistics of the hyper-networks corresponding to each session of the Congress, for both the House and the Senate.

Formally, we study group affinity in a hypergraph $H \doteq (V, E)$ whose nodes are partitioned in a set of classes $X_1, \dotsc, X_c$.
Let us consider hyperedges of the same size, i.e., we examine each $k$-uniform sub-hypergraph $H^k \doteq (V, E^k \doteq \{(h,t) \in E \, \wedge \, \sizeof{h} + \sizeof{t} = k\})$ for each size $k$, separately.
Taking inspiration from Veldt et al. \cite{veldt2023combinatorial}, we define a notion of the group affinity for directed hypergraphs.
For class $X_i$, the $(\alpha, \beta, k)$-affinity represents the extent to which a node of class $X_i$ belongs to the tail of a hyperedge of size $k$ where $\alpha$ of the $\beta$ nodes in the head are from class $X_i$:
\begin{equation}\label{eq:aff_h}
\textsf{A}_{\alpha, \beta, k}(X_i) = \frac{\sum_{v \in X_i}\indeg{H^k}{v, \alpha, \beta, X_i}}{\sum_{v \in X_i}\indeg{H^k}{v, \beta}}
\end{equation}

where 
$\indeg{H^k}{v, \alpha, \beta, X_i} = \sizeof{ \{ (h,t) \in E^k \suchthat \sizeof{ \left\{ u \in h \, \wedge \, u \in X_i \right\} } = \alpha \, \wedge \, \sizeof{ h } = \beta \, \wedge \, v \in t \} }$, and 
$\indeg{H^k}{v, \beta} = \sizeof{ \{ (h,t) \in E^k \suchthat \sizeof{ h } = \beta \, \wedge \, v \in t \} }$.


To determine whether the affinity score for a certain class $X_i$ is significantly high or low, we compare it against \textbf{(i)} the average score $\bar{\textsf{A}}_{\alpha, \beta, k}(X_i)$ measured in a collection of random hypergraphs sampled by \algoA and \algoC, and \textbf{(ii)} a baseline score adapted from Veldt et al.~\cite{veldt2023combinatorial}, which represents a null probability of $k$-interactions with head size $\beta$.
This baseline $(\alpha, \beta, k)$-affinity score for class $X_i$ is
\begin{align}\label{eq:base_h}
\textsf{B}_{\alpha, \beta, k}(X_i) = \frac{
	\overset{\textcolor{red}{\text{(1)}}}{\displaystyle\binom{|X_i|}{\alpha}}
	\overset{\textcolor{red}{\text{(2)}}}{\displaystyle\binom{n - |X_i|}{\beta - \alpha}}
	\overset{\textcolor{red}{\text{(3)}}}{\displaystyle\binom{n - 1}{k - \beta - 1}}}
{
\underset{\textcolor{red}{\text{(4)}}}{\displaystyle\binom{n}{\beta}}
\underset{\textcolor{red}{\text{(5)}}}{\displaystyle\binom{n-1}{k - \beta -1}}
}
= \frac{\displaystyle\binom{|X_i|}{\alpha}\displaystyle\binom{n - |X_i|}{\beta - \alpha}}{\displaystyle\binom{n}{\beta}}\,,
\end{align}
where $n$ is the total number of nodes in $H$,
\textcolor{red}{(1)} and \textcolor{red}{(2)} represent the number of ways to choose a head of size $\beta$ with $\alpha$ elements of class $X_i$, and the remaining elements from class different from $X_i$;
\textcolor{red}{(3)} represents the number of ways to choose a tail of size $k - \beta$, under the assumption that the same node can appear both in the head and in the tail of the hyperedge, having already selected one node of the tail; and
\textcolor{red}{(4)} and \textcolor{red}{(5)} represent the number of ways to form a $k$-size hyperedge with head size $\beta$, under the assumption that the same node can appear both in the head and in the tail of the hyperedge, having already selected one node of the tail.


The specific case where the head of each hyperedge has size $1$ is of particular practical interest for studying the co-sponsoring of congress bills, which are sponsored by a single member of Congress and supported by any number of co-sponsors. 
Then, \Cref{eq:aff_h} reduces to
\begin{align}\label{eq:affHone}
\textsf{A}_k(X_i) = \frac{\sum_{v \in X_i}\indeg{H^k}{v, X_i}}{\sum_{v \in X_i}\indeg{H^k}{v}}\,,
\end{align}
where 
$\indeg{H^k}{v, X_i} = \left|\left\{ ([u],t) \in E^k \suchthat u \in X_i \, \wedge \, v \in t \right\}\right|$.
%
%
\Cref{eq:affHone} can be seen as the probability that a node of class $X_i$ joins the tail of a hyperedge of size $k$, knowing that the head is of class $X_i$. 
%
As $\alpha = \beta = 1$, the baseline expressed by \Cref{eq:base_h} reduces to
\begin{equation}\label{eq:baseHone}
\textsf{B}_k(X_i) = \frac{|X_i|}{n}\,.
\end{equation}


In the case of directed hypergraphs with head-sequence $[1, \dotsc, 1]$, we also measure the \emph{homophily} $\mathsf{HO}$ of class $X_i$ as ~\cite{park2007distribution}:
\begin{equation}\label{eq:homo}
\mathsf{HO}(X_i) = \slfrac{ \mathsf{m}(X_i) }{ \bar{\mathsf{m}}(X_i) }\,
\end{equation}

where 
$\mathsf{m}(X_i) = \sfrac{1}{\sizeof{t}} \sum\limits_{\substack{e \doteq ([u], t) \in E\\u \in X_i}} \sizeof{ \{v \in t \suchthat v \in X_i\} }$
is measured in the observed hypergraph, 
and $\bar{\mathsf{m}}$ is the average across the samples generated by \algo.  



\begin{figure*}[t!]
    \begin{subfigure}{.48\linewidth}
      \centering
      \includegraphics[width=\linewidth]{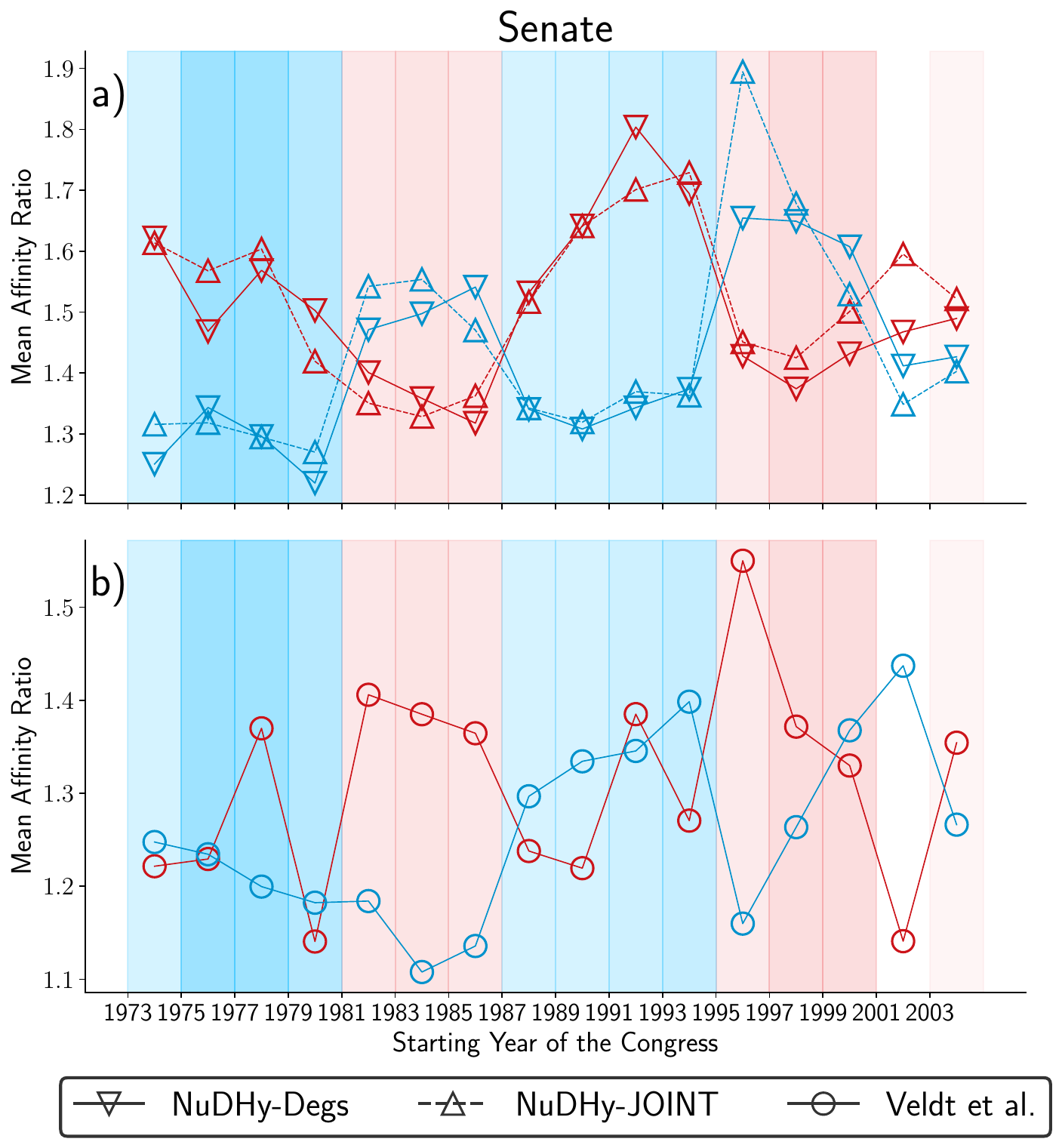}
    \end{subfigure}
    \begin{subfigure}{.50\linewidth}
      \centering
      \includegraphics[width=\linewidth]{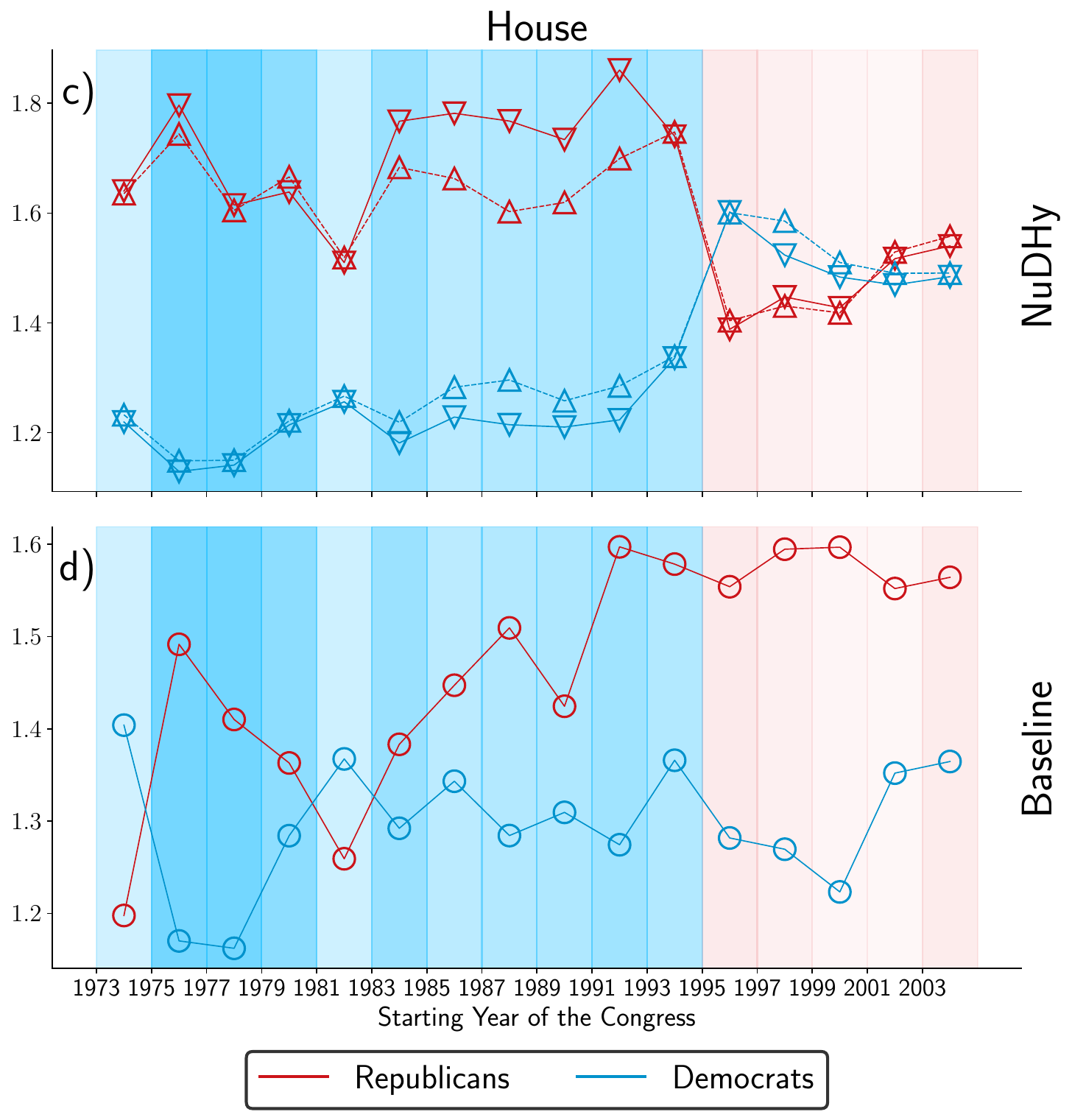}
    \end{subfigure}
    \caption{\textbf{Mean affinity ratios in the US Congress co-sponsored bills.}
    We show results for \Cref{eq:affHone} divided by the mean values in $33$ samples for \algoA and \algoC for the US Senate (\textsc{S-bills)}, panel (a)) and House (\textsc{H-bills)}, panel (c)). 
    For comparison, we show the values of \Cref{eq:affHone} divided by \Cref{eq:baseHone} for Veldt et al. again for the US Senate (panel (b)) and House (panel (d)). 
    The colors indicate \emph{Democrats} (blue) and \emph{Republicans} (red). 
    We report the average ratios over $k = 2, \cdots, 14$.}
    \label{fig:affinity_h_congress_sh}
\end{figure*}

\Cref{fig:affinity_h_congress_sh} illustrates the mean affinity ratios for \emph{Democrats} and \emph{Republicans} in each Congress, for \textsc{S-bills} and \textsc{H-bills}.
The mean affinity ratios for \algoA and \algoC are computed by averaging the terms $\textsf{A}_k(X_i)/\bar{\textsf{A}}_k(X_i)$ over all hyperedge sizes $k = 2, \ldots, 14$.
The mean affinity ratios for Veldt et al. are obtained by averaging the terms $\textsf{A}_k(X_i)/\textsf{B}_k(X_i)$ over $k = 2, \ldots, 14$.
For each Congress, the background color indicates which party held the majority (shades of red for Republicans and shades of blue for Democrats). 
The intensity of color corresponds to the size of the majority, with darker shades indicating a larger margin.

These plots show that we can draw similar conclusions when comparing the affinity values against the null models obtained by \algoA and \algoC, whereas the baseline scores offer divergent insights.
The panels corresponding to \algoA/\algoC reveal a clear trend (\Cref{fig:affinity_h_congress_sh}\textbf{a}-\textbf{c}): when one party holds the majority of the seats (indicated by the corresponding color in the background), the opposing party exhibits higher group affinity. 
This pattern indicates a more unified front, likely in pursuit of collecting the required minimum support to pass their bills.

In instances where Republicans held the majority, Democrats consistently maintained a group affinity $40\%$ to $60\%$ higher than expected, with the exception of the $104^{\mathrm{th}}$ Senate Con\-gress, coinciding with the first occurrence of a Republican majority in both chambers since 1953 and a government shutdown in the US.
Conversely, during Democratic majority periods, Republicans exhibited notably higher group affinity, particularly leading up to the $104^{\mathrm{th}}$ session, and especially in the House.
Data shows a lower number of bills sponsored by Republicans and a tendency to co-sponsor fewer bills. However, when they decide to co-sponsor a bill, it is more likely to be a bill presented by a Republican.
This pattern is consistent with past observations that ``Republicans have consistently valued doctrinal purity over pragmatic deal-making''~\cite{grossmann2015ideological}.

In contrast, the baseline yields generally lower affinity values and tends to attribute higher group affinity to the Republicans party, irrespective of the time period.
An exception is evident in the $107^{\mathrm{th}}$ Senate Congress starting in 2001, where the mean affinity ratio for Republicans only slightly surpasses $1$, whereas the ratio for Democrats is roughly $1.45$. 
During this session of the Congress, there is a discernible disparity in co-sponsorship tendencies between Republicans and Democrats. On average, a Republican member tended to co-sponsor fewer bills, averaging around $119$, while their Democratic counterparts engaged in a higher rate of co-sponsorship, averaging around $195$ bills. The baseline score, which fails to consider each party's relative prevalence and each legislator's individual co-sponsoring opportunities, inadequately acknowledges the significance of Republican co-sponsoring behaviors for bills sponsored by Republicans versus those sponsored by Democrats.
Our null models, instead, maintain these characteristics of the data intact, while randomizing the rest.

In addition, we also found a clear shift in co-sponsoring behavior within the House around the $104^{\mathrm{th}}$/$105^{\mathrm{th}}$ Congress (1995/1997). 
During this period, Democrats began to consistently co-sponsor a greater number of bills sponsored by Democrats compared to Republicans (see \Cref{tbl:data_senate_house} in \Cref{ax:data}), possibly hoping to increase the likelihood of the bills being passed.
\algo effectively models this shift, as reflected in the corresponding mean affinity ratios.

Party homophily has been studied also by Neal et al.~\cite{neal2022homophily}.
They represent bill co-spons\-orship data as a unipartite weighted graph, where legislators serve as nodes and edge weights indicate the number of bills co-sponsored by pairs of legislators.
To ascertain statistically significant connections, they employ a stochastic degree sequence model (SDSM). 
Despite using a data model that overlooks higher-order relations between legislators and using a simplified analytical framework (a thresholded weighted graph)~\cite{peel2022statistical}, they find results akin to our analysis. 
Specifically, both studies find evidence for differential homophily: the strength of Republicans' preference for collaborating with other Republicans differs from the strength of Democrats' preference for collaborating with other Democrats.

Differently from Neal et al., our work remains faithful to the original data.
Moreover, the use of randomized networks drawn from ensembles that retain some of the properties of the observed network is more suitable for identifying statistically significant connections~\cite{rottjers2021null}.

Finally, the results concerning \Cref{eq:homo} are presented in \Cref{ax:homophily}.
We observe that both parties exhibit an inclination toward associating with similar party members in co-sponsorship relations, and that the inverse relationship between the curves of Republicans and Democrats remains discernible, a pattern that remains unnoticed when solely examining the values of $\mathsf{m}$ measured in the observed hypergraphs.

\subsection{Contagion Processes in Contact Hyper-Networks}\label{sec:app_contagion}

The spread of information or diseases often transcends pairwise interactions and necessitates models that consider the collective influence of groups of individuals. 
For example, in the context of social and behavioral contagion, multiple studies have shown that exposure to multiple sources can be required~\cite{monsted2017evidence,karsai2014complex}.
Models of such complex contagion processes aim to capture group influences in social phenomena, such as norm adoption, rumor spread, and disease transmission.
These models embrace nonlinear connections between infection rates and sources of infection, which allows for mechanisms such as social reinforcement where multiple (or group) exposures have a larger collective impact than their sum.

More recently, multiple studies have proposed higher-order contagion models describing not only repeated interactions but rather genuine group interactions among agents. 
In these models, the substrates over which the process evolves are simplicial complexes~\cite{iacopini2019simplicial}, undirected hypergraphs~\cite{de2020social,st2022influential}, and directed hypergraphs featuring single-node tails~\cite{cui2023general}. 
In particular, undirected hypergraphs and simplicial complexes have proven more effective in modeling higher-order interactions between individuals. 
Conversely, single-tailed directed hypergraphs better model group influences on individuals.
The dynamic evolution of such contagion models is typically studied numerically on real-world hyper-networks, and compared to results obtained (both numerically and analytically) for the same dynamics on random hyper-networks \cite{st2021universal,iacopini2019simplicial,de2020social}.
To date, however, it is not clear what are the minimal constraints on such random hyper-networks required to reproduce the dynamical outcomes observed on the real-world hypergraphs. 
Here, using \algo, we highlight the role of structural correlations in shaping the dynamical outcomes of contagion processes. 
In particular, we show that stronger constraints (as implemented by \algoC) are required to faithfully reproduce results of super-linear contagion dynamics, while looser constraints on degrees and tail/head sizes (\algoA) are sufficient when the dynamics is pairwise (linear). 

We consider a hypergraph SIS contagion model~\cite{st2022influential} wherein the infection rate is a super-linear function of the number of infected nodes in the hyperedges.
Let $e$ be a hyperedge and $i_e$ be the number of infected nodes in $e$.
Then, each of the susceptible nodes in $e$ gets infected at rate $\beta(i_e) = \lambda i_e^{\nu}$, where $\nu$ is a parameter to regulate the non-linearity of the contagion process.
The model assumes that infections from different hyperedges are independent processes, and thus defines the total transition rate to the infected state of a node $v$ as the sum of the infection rates of all the hyperedges $E(v)$ containing $v$, i.e., $\sum\limits_{e \in E(v)}\beta(i_e)$.
Let $\mu$ denote the recovery rate (we set $\mu = 1$ in all the experiments).
Nodes undergo multiple transitions between susceptible and infected states. 
The contagion process is simulated using a Gillespie algorithm~\cite{gillespie2007stochastic}.
Starting with an initial density $\rho_0$ of infected nodes, the process unfolds via the two types of events (infection and recovery) occurring with probabilities proportional to their respective rates.
Once a hyperedge is selected for an infection event, a susceptible node in the hyperedge is chosen uniformly at random to transition to the infected state.
To obtain the density of infected nodes in the stationary state, we let the system evolve over a burn-in period $\tau_b  = 10$k.
Then, we sample $s = 10$k states separated by a decorrelation period $\tau_d = 1$.
Finally, we measure the mean and standard deviation of the density of infected nodes in these samples.
 
We compare the results of the simulations in the observed hypergraphs and the samples generated by \algo with the output of group-based approximate master equations (AMEs) that consider the ensemble of hypergraphs with the same distribution of hyperedge sizes and node degrees~\cite{st2022influential}.
We investigate two scenarios.
The first scenario involves undirected hypergraphs depicting face-to-face interactions among children in a primary school in Lyon, France~\cite{gemmetto2014mitigation} (\textsc{lyon}) and among students in a high school in Lyc\'{e}e Thiers, France~\cite{mastrandrea2015contact} (\textsc{high}).
These hypergraphs are characterized by nearly homogeneous hyperedge size distributions (between 60-70\% of the hyperedges have size $2$ and between 15-20\% of the hyperedges have size $3$) and bell-shaped node degree distributions centered around $11.79$ and $55.63$, respectively.
The second scenario involves email exchanges between members of a European research institution (\textsc{email-EU}) and between Enron employees (\textsc{email-Enron})~\cite{benson2018simplicial}.
These hypergraphs are characterized by heterogeneous hyperedge size distributions with mean hyperedge size $3$ and $3.42$, respectively, and max hyperedge size $18$ and $25$, respectively. The node degree distributions follow a heavy-tailed distribution.
The main characteristics of the four hypergraphs are reported in \Cref{tbl:contact_networks} in \Cref{ax:data}.

\begin{figure}[t!]
  \centering
  \includegraphics[width=\linewidth]{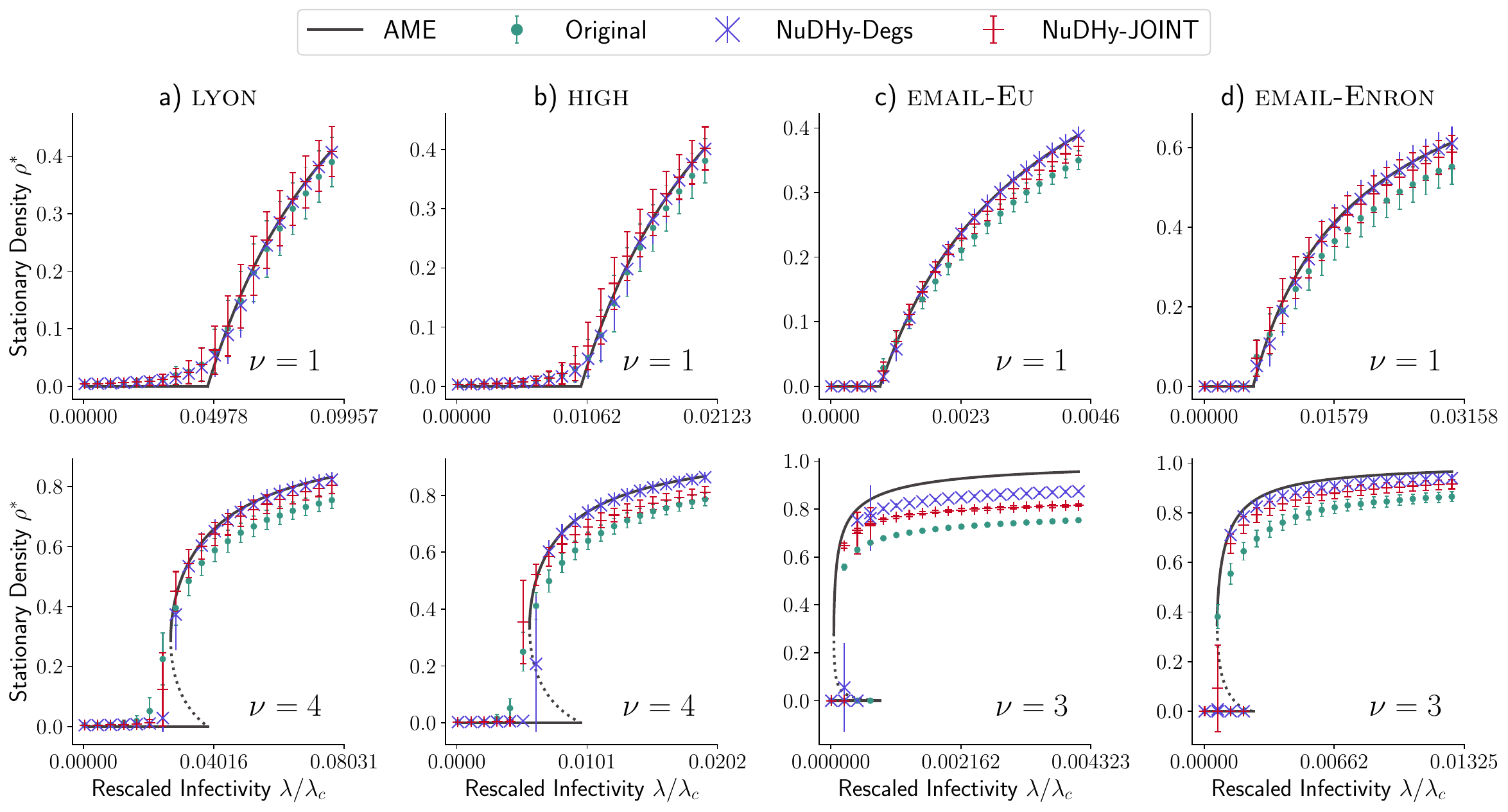}
  \caption{\textbf{Density of infected nodes in contact networks.} We show the values of $\rho^*$ in the stationary state of contagion dynamics on the observed hypergraph, and on $33$ samples generated by \algoA and \algoC, varying infection rate $\lambda$ and non-linearity parameter $\nu$, for \textsc{lyon}, \textsc{high}, \textsc{email-Eu}, and \textsc{email-Enron}. We report also the output of the AMEs as defined in~\cite{st2022influential}. The infection rate is rescaled with the invasion threshold $\lambda_c$. 
  Errors bars correspond to one standard deviation.
}
  \label{fig:rhos}
\end{figure}

\Cref{fig:rhos} display the average fraction $\rho^*$ of infected nodes in the stationary state of contagion dynamics on the observed hypergraphs and on $33$ samples generated by \algoA and \algoC, using $\rho_0 = 0.01$, and varying infection rate $\lambda$ and parameter $\nu$. The phase diagram reports also the output of the AMEs.
The infection rate is rescaled with the \emph{invasion threshold} $\lambda_c$, which is the minimum $\lambda$ above which the healthy state ($\rho^* = 0$) is unstable.
We consider both linear ($\nu = 1$) and super-linear ($\nu > 1$) contagions.
In the case of linear contagions, we observe two solutions for the stationary fraction of infected nodes: $\rho^*_1 = 0$ (absorbing state) and $\rho^*_2 > 0$ (endemic state). 
For the case of super-linear contagions, we chose a value of $\nu$ greater than the \emph{bistability threshold} $\nu_c$ reported in \Cref{tbl:contact_networks}.
The bistability threshold is the smallest non-linear exponent allowing for a discontinuous phase transition.
In this case we observe three solutions: $\rho^*_1= 0$ and $\rho^*_3 > 0$, which are locally stable, and $0 < \rho^*_2 < \rho^*_3$, which is unstable (dashed lines).
To obtain the lower branches in panels \textbf{c} and \textbf{d} in \Cref{fig:rhos} we run the ordinary simulation method described above. 
On the other hand, the upper branches in panels \textbf{c} and \textbf{d} and the branches in panels \textbf{a} and \textbf{b} are obtained with a quasi-stationary-state method~\cite{PhysRevE.71.016129}: we keep a history of $50$ past states from which a random state is used to replace the current state each time the absorbing state is reached.

Especially in the smaller datasets and for linear contagion processes, results for \algoA align well with the output of the analytical framework.
This is expected given that \algoA samples uniformly from the ensemble of hypergraphs with the same head and tail size sequences and the same in- and out-degree sequences, which are equivalent to the hyperedge size sequence and the node degree sequence when the input is an undirected hypergraph.
The disparities observed in the super-linear processes may potentially be attributed to a too-small value assigned to $\tau_b$. 

In contrast, the structural correlations present in the observed data lead to reductions in the stationary prevalence compared to the output of the AMEs.
These deviations display greater magnitude in the outputs of the super-linear contagions and in the presence of unstable regions, thus suggesting a higher influence of structural correlations within this type of contagion process.
Previous works~\cite{st2021universal} has shown that the correlations are especially important in the presence of nodes with large degrees.
In line with these works, we observe smaller discrepancies in \textsc{lyon}, where node degrees are more homogeneous.  
Conversely, the discrepancies in the lower branches in panels \textbf{a} and \textbf{b} in \Cref{fig:rhos} are due to the simulations being affected by finite-size effects---while AMEs assume an infinite-size system---and they become higher for the super-linear processes.

By looking at the curves for \algoC in \textsc{email-Eu} and \textsc{high}, we observe that part of the deviation in the super-linear simulations can be explained by the joint degree distribution.

In conclusion, our analysis shows that the fidelity to the original dynamics increases with the amount of preserved structural correlations, with \algoC offering the closest approximation, and the AME being the least accurate. 
While both \algoA and \algoC closely match the real dynamics for linear processes in smaller datasets, discrepancies emerge between their predictions in super-linear processes, with \algoC better approximating the real dynamics. 
This result highlights the role of higher-order structural correlations in non-linear contagion models, and thus highlights the importance of preserving the joint degree tensor when strongly non-linear processes or strong degree correlations are present (e.g. \textsc{email-Eu} and \textsc{high}).



\subsection{Economic Competitiveness in Trade Hyper-Networks}\label{sec:app_trade}

Economic complexity metrics are indicators that aim to capture the diversity and sophistication of a country's economy through its exported product basket.
The diversity and composition of a country's exported product basket, along with the complexity of the products therein, are the key properties exploited by these metrics to asses the competitiveness of countries.
In this analysis, we gauge the relative economic competitiveness 
of countries via three of these metrics: the Economic Complexity Index (ECI)~\cite{hidalgo2007product}, the Fitness~\cite{tacchella2012metrics,tacchella2013economic,cristelli2013measuring}, and the GENeralised Economic comPlexitY index (GENEPY)~\cite{sciarra2020reconciling}. 
We apply \algo alongside three additional null models purposefully designed for directed hypergraphs, with the aim to investigate which characteristics of the observed data are sufficient to replicate the ranking of countries based on these metrics.

Each of the three metrics is defined on an unweighted bipartite graph that represents the export relationships between countries and products: the \textit{bipartite country-product network}.
Nodes of one set represent countries, and nodes of the other set represent products.
Unweighted and undirected edges connect countries to their exported products.
Following previous literature in this field, we consider a country to be an exporter of a product if its \emph{Revealed Comparative Advantage} (RCA)~\cite{balassa1965trade} is greater than or equal to a minimum threshold $R^*$.
RCA measures the relative monetary importance of a product for a country among the export basket of the country compared to the global average.
We follow the standard economics literature and set $R^* = 1.0$~\cite{cristelli2013measuring}.
An RCA value greater than $R^*$ implies that the given country is advanced enough to compete in the global market for that product.
In addition, following the Atlas of Economic Complexity~\cite{harvard}, we only consider countries with a population above $1$ million and an average trade above $1$ billion USD.

Let $\mathsf{M}$ be the biadjacency matrix of the bipartite country-product network defined according to these criteria, and let $\mathsf{W}$ be a transformation matrix defined as 
$\mathsf{W}[c, p] = \mathsf{M}[c, p] / k_c h_p$, where $k_c$ is the degree of the left vertex $c$ (representing a country), and $h_p = \sum_c \mathsf{M}[c,p] / k_c$.
The \emph{country-to-country proximity matrix} between countries is then defined as follows:
\begin{align*}
\mathsf{X}[c,c^*] = 
\begin{cases}
	\sum_p \mathsf{W}[c,p]\mathsf{W}[c^*,p] & \text{ if } c \neq c^* ,\\
	0 & \text{ if } c = c^*.
\end{cases}
\end{align*}
The symmetric matrix $\mathsf{X}$ quantifies the similarities in the export baskets of countries.
Let us now recall the three metrics under study.

The \emph{Economic Complexity Index} (ECI) measures a country's complexity as the average complexity of the products it exports, and the complexity of a product as the average complexity of the countries that export it.
Thus, countries with a high ECI boast diversified export portfolios, featuring unique and sophisticated products, while those with a lower ECI export a more limited selection of common goods.
In terms of the biadjacency matrix $\mathsf{M}$, the ECI of a country and the \emph{Product Complexity Index} (PCI) of a product are defined by the following coupled equations:
\begin{equation}
\mathsf{ECI}(c) = \frac{1}{\sum_p \mathsf{M}[c,p]}\sum_p\mathsf{M}[c,p]\mathsf{PCI}(p),
\end{equation}
\begin{equation}
\mathsf{PCI}(p) = \frac{1}{\sum_c \mathsf{M}[c,p]}\sum_c\mathsf{M}[c,p]\mathsf{ECI}(p).
\end{equation}
The ECI index also possesses an alternative equivalent definition in terms of the eigenvector corresponding to the second largest eigenvalue of the country-to-country proximity matrix $\mathsf{X}$~\cite{inoua2023simple}.

The \emph{Fitness} $\mathsf{F}(c)$ of a country $c$ and the \emph{Quality} $\mathsf{Q}(p)$ of a product $p$ are defined according to the following coupled equations~\cite{tacchella2012metrics}:
\begin{equation}\label{eq:fitness}
\begin{dcases}
    \widetilde{\mathsf{F}}(c)^{(n)} & = \sum_p \mathsf{M}[c,p]\mathsf{Q}(p)^{(n-1)}\\
    \widetilde{\mathsf{Q}}(p)^{(n)} & = \frac{1}{\sum_c \mathsf{M}[c,p]\frac{1}{\mathsf{F}(c)^{(n-1)}}}
\end{dcases}
\rightarrow
\begin{dcases}
\mathsf{F}(c)^{(n)} & = \frac{\widetilde{\mathsf{F}}(c)^{(n)}}{\langle\widetilde{\mathsf{F}}(c)^{(n)}\rangle_{c}}\\
\mathsf{Q}(p)^{(n)} & = \frac{\widetilde{\mathsf{Q}}(p)^{(n)}}{\langle\widetilde{\mathsf{Q}}(p)^{(n)}\rangle_{p}}
\end{dcases}
\end{equation}
where $\langle\cdot\rangle_{x}$ denotes the arithmetic mean over the distribution of values for $x$.
The main difference introduced by the Fitness/Quality scores lies in a non-linear weighting of the fitness of the countries when computing the quality of a product, rather than using a simple average.
Fitness and Quality can be computed by solving \Cref{eq:fitness} with an iterative algorithm, initializing $\mathsf{F}(c)^{(0)} = 1$ for each country $c$, and $\mathsf{Q}(p)^{(0)} = 1$ for each product $p$.
The iterative algorithm converges to a single fixed point independently from the initial conditions~\cite{tacchella2012metrics,cristelli2013measuring,pugliese2016convergence}.
%

Finally, the GENEPY index is a combination of the eigenvectors of the country-to-country proximity matrix $\mathsf{X}$.
More precisely, the GENEPY index of a country $c$ is defined as  
\begin{equation}\label{eq:genepy}
\mathsf{G}(c) = \left(\sum_{i=1}^2 \lambda_{c,i} \mathsf{e}_{c,i}^2 \right)^2 + 2\sum_{i=1}^2\lambda_{c,i}^2 \mathsf{e}_{c,i}^2\,,
\end{equation}
where $\lambda_{c,i}$ is the $i$-th largest eigenvalue of the proximity matrix $\mathsf{X}$, and $\mathsf{e}_{c,i}$ is the corresponding eigenvector.

As a preliminary observation, note that the bipartite country-product network inherently represents a high-order structure~\cite{saracco2015randomizing}, as any hypergraph can be represented as a bipartite graph without loss of information.
Therefore, computing metrics on the bipartite country-product network corresponds to conducting higher-order analyses.

We perform a comparative analysis of country rankings based on ECI, Fitness, and GE\-NEPY computed from the observed data and from $33$ samples generated by \algo.
We consider international trade data for four years: 1995 (first year available), 2009 (global financial crisis), 2019 (COVID-19 outbreak), and 2020 (economic recession)~\cite{hausmann2014atlas}.
We consider a directed higher-order data representation where nodes represent countries and hyperedges represent products traded by them. 
Coherently with the construction of the bipartite country-product network, the head of each hyperedge includes countries that export the product with a \emph{Revealed Comparative Advantage}~\cite{balassa1965trade} greater than $1$; the tail of each hyperedge includes countries that import the product with an RCA greater than $1$; and we only consider countries with a population above $1$ million and an average trade above $1$ billion USD.
\Cref{tbl:datasets_eco} in \Cref{ax:eco_other} reports the characteristics of the resulting hypergraphs.
This directed hypergraph encoding perfectly represents the trade data and offers opportunities for studying the system more thoroughly.
For instance, while the country-product network only looks at the export side of the trades, the directed hypergraph also represents imports.
Higher-order representations thus offer a richer and more detailed description of the system on which more powerful metrics can be defined (although this specific task is outside the scope of the current work).

\begin{table}[t]
  \caption{Average Spearman's correlation and Kendall's Tau (KT) of the rankings of the countries obtained according to ECI/Fitness/GENEPY in $33$ samples with respect to the observed rankings, for \textsc{hs2019}. Standard deviations are reported in parentheses.}
  \label{tbl:ranks_genepy_eci}
  \resizebox{\columnwidth}{!}{
  \begin{tabular}{llccccc}
  \textbf{Score} & \textbf{Metric} & \base & \based & \nullm & \algoA & \algoC \\
    \midrule
    \multirow{2}{*}{ECI} & 
    Spearman & -0.144 (0.143) & -0.055 (0.120) & 0.007 (0.093) & 0.020 (0.121) & 0.964 (0.001)\\
    & KT & -0.092 (0.092) & -0.037 (0.079) & 0.005 (0.062) & 0.015 (0.082) & 0.848 (0.004)\\
  \midrule
    \multirow{2}{*}{Fitness} & 
    Spearman & 0.051 (0.075) & 0.237 (0.055) & -0.013 (0.075) & 0.981 (0.001) & 0.998 (0.000)\\
    & KT & 0.034 (0.051) & 0.160 (0.038) & -0.010 (0.052) & 0.886 (0.003) & 0.963 (0.001)\\
  \midrule
  \multirow{2}{*}{GENEPY} & 
    Spearman & 0.015 (0.078) & 0.230 (0.054) & -0.002 (0.097) & 0.941 (0.004) & 0.993 (0.000)\\
  & KT & 0.010 (0.054) & 0.156 (0.040) & -0.001 (0.040) & 0.801 (0.007) & 0.937 (0.002)\\
  \bottomrule
  \end{tabular}
  }
\end{table}

\begin{figure*}[!th]
    \centering
    \includegraphics[width=\linewidth]{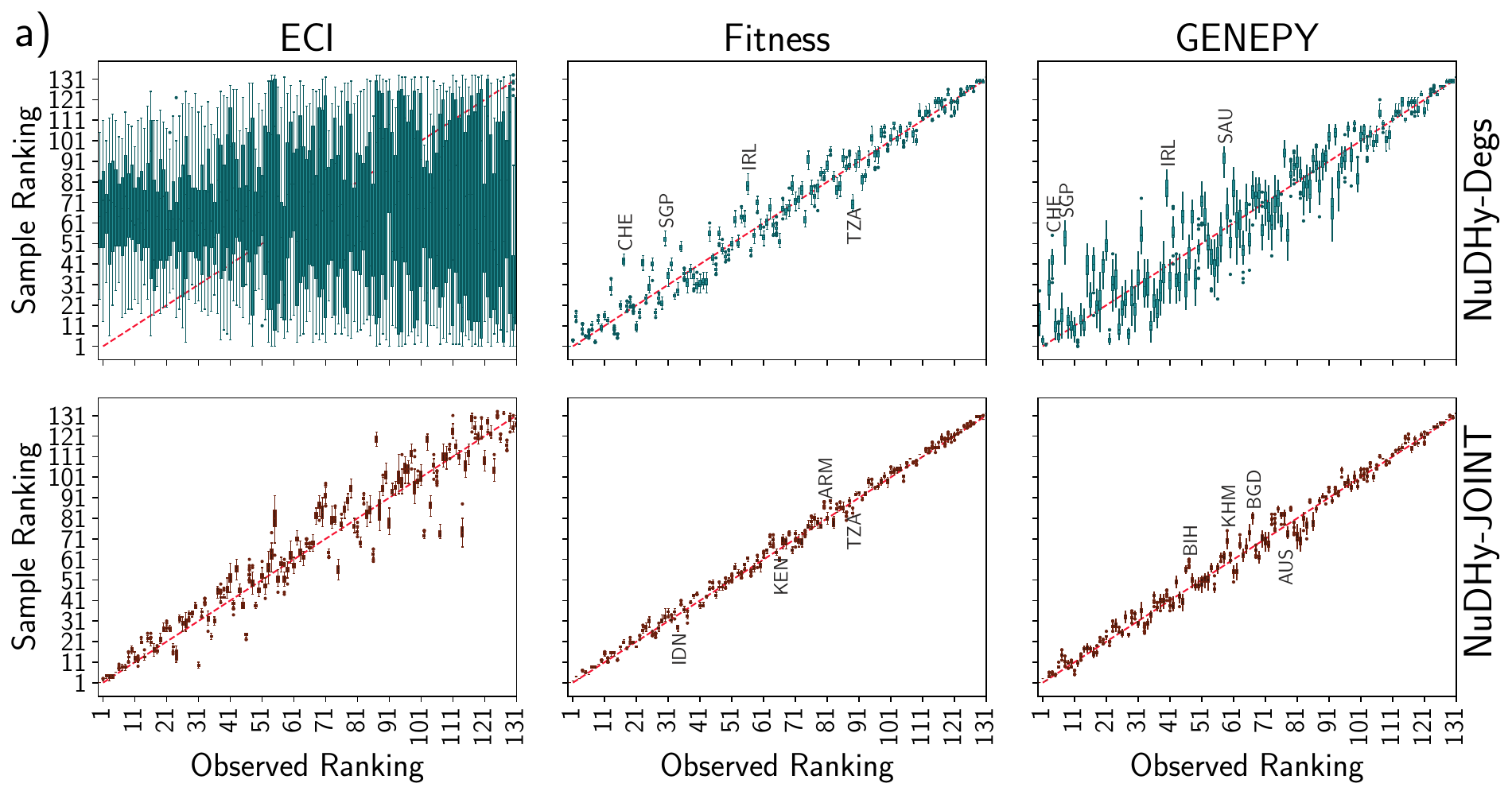}
    \includegraphics[width=\linewidth]{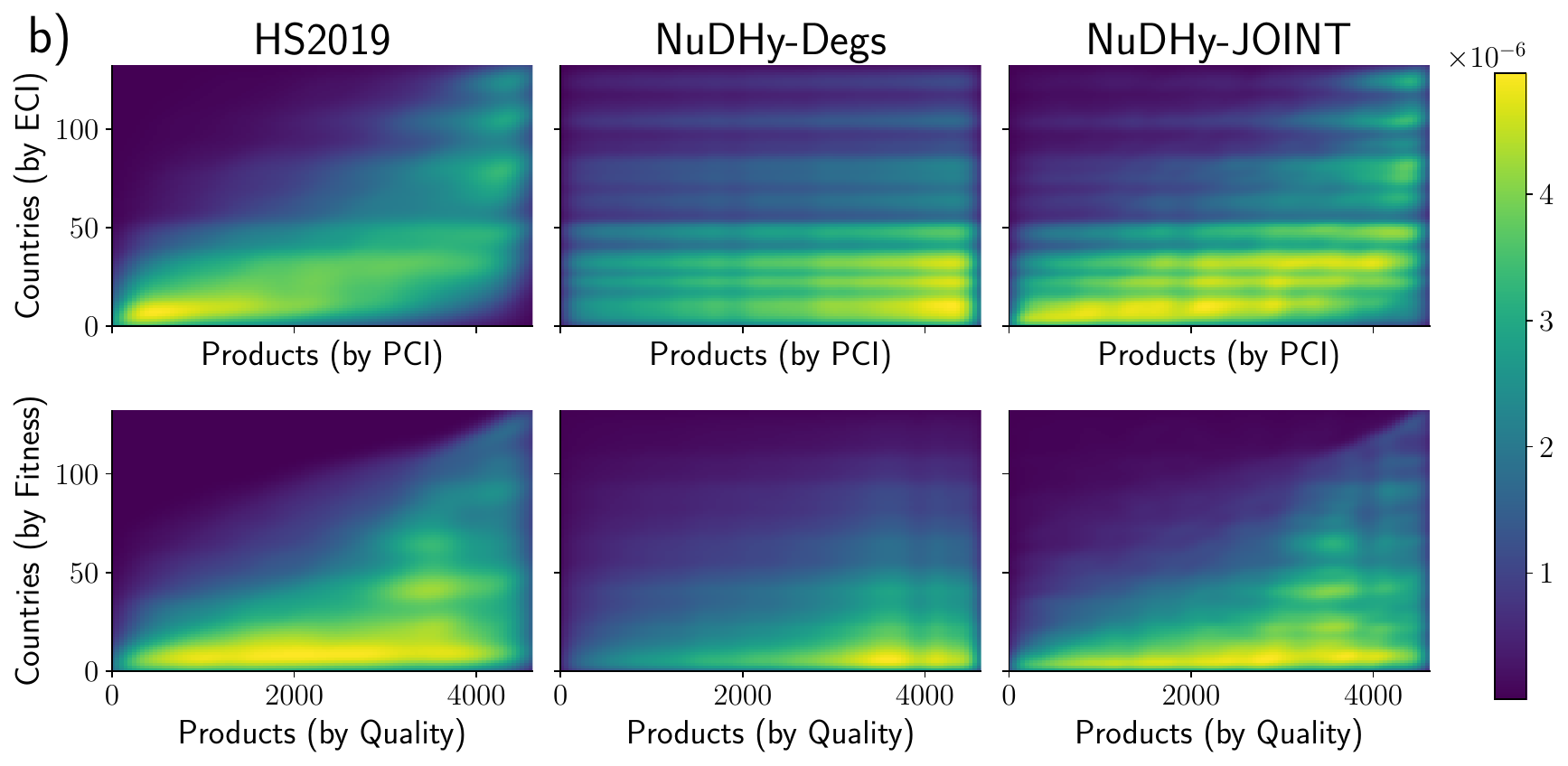}
    \caption{\textbf{Relative competitiveness in \textsc{hs2019}.} Panel a): rankings distributions based on ECI, Fitness, and GENEPY across 33 samples for \algoA (top) and \algoC (bottom) compared to the observed rankings, with annotated top-$4$ diverging ranks.
    Panel b): density plots of the KDE of the observed biadjacency matrix 
    $\mathsf{M}$ and of the aggregated matrices across $33$ samples of \algoA and \algoC. Countries are sorted by ECI/Fitness and products by PCI/Quality (descending). The lighter the color, the higher the density of edges.
    }
    \label{fig:hs2019}
\end{figure*}

An analysis similar to ours was presented in a previous study~\cite{straka2018ecology} by employing the Fitness score and the BiCM null model~\cite{saracco2015randomizing} for the bipartite country-product network. 
This null model maintains both the left and right degree sequences, but only in expectation (canonical ensemble).
The study revealed that, in general, for each country, the distribution of its ranks obtained from the samples has a mean value close to the observed rank, and a wide standard deviation.
We find a similar, albeit much stronger, result for \algo.
In the following, we discuss the main findings in \textsc{hs2019}.
Results for the other trade networks are qualitatively similar and are reported in \Cref{ax:eco_other}.

\mpara{ECI.}
For \algoA, both the Spearman and Kendall's Tau average correlation values of the rankings of countries are remarkably close to zero, and the standard deviation values for both coefficients are small (\Cref{tbl:ranks_genepy_eci}).
This result indicates independence between the observed ranking and the rankings provided by the samples.
\Cref{fig:hs2019}\textbf{a} (upper left) shows this pattern: the distributions of the ranks of each country across the samples tend to cluster around mid-ranking positions and exhibit a wide spread, with greater variance at the lower end of the observed ranking.
In other words, preserving the degrees via \algoA does not preserve the ECI.


In contrast, for \algoC both average correlation coefficients (Spearman and Ken\-dall's Tau) of the rankings of countries are significantly high ($\ge0.84$), and the standard deviation values for both coefficients are negligible (\Cref{tbl:ranks_genepy_eci}).
This observation suggests a dependence between the observed ranking and the rankings provided by the samples, as illustrated by the bottom left plot of \Cref{fig:hs2019}\textbf{a}.
The distributions of the ranks of each country across the samples are aligned with the observed rank of a country and present a narrow spread.

According to these results, preserving the JOINT is sufficient to preserve the ranking of countries based on their ECI score, while the degree sequence is insufficient.

\mpara{Fitness and GENEPY.}
These two measures behave quite similarly in our analysis.
For \algoA, the average correlation coefficients (Spearman and Kendall's Tau) of the rankings of countries are significantly high for both Fitness ($\ge0.88$) and GENEPY ($\ge0.8$), and the standard deviation values for both coefficients are negligible (\Cref{tbl:ranks_genepy_eci}).
This result indicates a dependence between the observed rankings and the rankings provided by the samples.
The middle and right plots of \Cref{fig:hs2019}\textbf{a} show that the distributions of the ranks of each country across the samples tend to be close to the observed rank of a country with a narrow spread.

For \algoC, the average correlation coefficients (Spearman and Kendall's Tau) of the rankings of countries are even higher ($\ge0.96$ for Fitness and $\ge0.93$ for GENEPY), and the standard deviation values for both coefficients are extremely small (\Cref{tbl:ranks_genepy_eci}).
There is a strong dependence between the observed ranking and the rankings provided by the samples, as shown by the bottom middle and right plots of \Cref{fig:hs2019}\textbf{a}.
The distributions of the country rank across the samples are aligned with the observed rank of a country, with a very limited spread.

According to these results, both the degree and joint degree sequence are sufficient to retain the ranking of countries based on their Fitness and GENEPY scores.

Figure~\ref{fig:hs2019}\textbf{b} displays density plots representing Kernel Density Estimations (KDEs) of the biadjacency matrices for the country-product network of $2019$. These matrices are derived from the observed data (first column) and from the aggregation of $33$ samples generated by \algoA and \algoC (second and third column).
Countries and products are arranged in descending order of ECI/Fitness and PCI/Quality, respectively.
The color intensity within each plot indicates the density of edges, with lighter colors indicating higher density.

As expected, countries with a high Fitness/ECI predominantly export products with high Quality/PCI, while those with lower Fitness/ECI focus solely on products with lower Quality/PCI.
Comparison with the corresponding plots for \algoA (middle columns of \Cref{fig:hs2019}\textbf{b}) indicates that the specialization process of countries cannot be fully explained by node degrees alone, as evidenced by the inability of \algoA to accurately capture the pattern observed in the real data. 
Conversely, plots derived from samples of \algoC reveals remarkably similar edge density distributions to those observed, regardless of the metrics used to sort rows and columns (first and third columns of \Cref{fig:hs2019}\textbf{b}). 

Overall, we find that preserving local properties of the hypergraph, either the degree sequences and hyperedge sizes in the case of \algoA or their joint tensor for \algoC, is sufficient to explain the rankings induced by most economic complexity measures.
As a consequence, it is likely that these measures primarily capture local network structure and do not fully leverage meso- and global-scale information.
Our suite of null models \algo can help explore the power of these structural metrics, and possibly develop more comprehensive ones that can leverage the natural higher-order representation of the underlying trade data.

\spara{Other null models.}
In this experiment, we also compare our null models with three null models for directed hypergraphs named \base, \based, and \nullm.
\base and \based are two different versions of \redi~\cite{kim2022reciprocity}: the first realistic generative model specifically designed for directed hypergraphs.
\redi extends the preferential attachment model~\cite{do2020structural} to directed hypergraphs, allowing the generation of random hypergraphs exhibiting reciprocal patterns akin to those observed in real directed hypergraphs.
The random hypergraphs generated by this model preserve, on average, the distribution of head and tail sizes.
The version of \redi dubbed \base preserves, on average, the distribution of the number of hyperedges in which each group of nodes appears, while the version dubbed \based preserves node degrees on average.
\nullm is a naive sampler that preserves the head and tail size distributions, but populates the hyperedges of the random hypergraph by drawing nodes uniformly at random from the set of nodes of the observed hypergraph.
Additional details on these models are provided in \Cref{ax:base}.

According to the results reported in \Cref{tbl:ranks_genepy_eci}, none among \base, \based, and \nullm, can explain the ranking of countries based on these three indexes.

%% file: sections/discussion.tex
\section{Discussion}\label{sec:discussion}

In this study, we introduced a suite of null models for directed hypergraphs, encompassing hypergraphs with the same in-degree, out-degree, head-size, and tail-size distributions, as well as the same JOINT of an observed hypergraph.
We demonstrated a lossless mapping from directed hypergraphs to directed bipartite graphs and proposed two MCMC samplers that efficiently sample from the corresponding micro-canonical graph ensembles.

Our approach fills a critical gap in the existing literature, which primarily focuses on canonical and micro-canonical bipartite graph ensembles~\cite{kannan1999simple,tabourier2011generating,strona2014fast,saracco2015randomizing,squartini2015unbiased,boroojeni2017generating,aksoy2017measuring,del2010efficient} and undirected hypergraph ensembles~\cite{saracco2022entropy,do2020structural,barthelemy2022class,guo2016non,wang2010evolving,chodrow2020configuration,zeng2023hyper,nakajima2021randomizing,miyashita2023randomizing}.

We conducted rigorous experiments and evaluations, highlighting the limitations of recent generative models, such as the one proposed by Kim et al.~\cite{kim2022reciprocity}, specifically designed for directed hypergraphs. 
The random hypergraphs generated by this (canonical) model preserve, on average, the distribution of head and tail sizes.
However, our findings revealed structural dissimilarities between generated hypergraphs and observed ones due to design choices aimed at improving sampler efficiency.

We then showed the importance of preserving stronger structural correlations (and hence the significance of the proposed null models) in three appropriate case studies, spanning various domains. 
First, we explored group affinity within political parties in the US Congress, revealing an inverse relationship between the affinity curves of Republicans and Democrats: when one party holds the majority of the seats, the opposing party exhibits higher group affinity.
This pattern becomes apparent only when the JOINT structural correlations are preserved.  

Second, we simulated linear and non-linear contagion processes in real and randomized hyper-networks, demonstrating the explanatory power of the JOINT in elucidating observed discrepancies between analytical contagion frameworks and simulations in real data. 
These results also suggest that our models could be used for more realistic data augmentation. 

Third, we compared the rankings of countries based on three economic complexity indices (ECI, Fitness, and GENEPY) computed in trade hyper-networks and their randomized counterparts, highlighting the nuanced information encoded in the degree sequences and the JOINT.
Our analysis revealed that both our null models accurately replicate the relative economic competitiveness of countries as measured by Fitness and GENEPY. 
However, for ECI, only the more restrictive null model \algoC succeeded in preserving the rankings.
These results demonstrate that retaining the local topological properties independently is insufficient to preserve the ranking of the countries based on their ECI score.
However, in all three cases, the local properties preserved by \algoC are sufficient to reproduce and explain the rankings, thus indicating that the metrics ignore mesoscale and global properties of the network.

Our findings emphasize the versatility and effectiveness of our proposed null models and samplers in uncovering intricate patterns across diverse disciplines.
These tools represent a powerful lens through which to examine higher-order complex systems.
They fill a significant gap in the analysis of higher-order networks, thus providing researchers in fields such as neuroscience, ecology, sociology, and economics with effective means for analysis and interpretation. 
Moreover, thanks to the efficiency of our samplers, our work empowers researchers to glean deeper insights also from more complex and larger datasets.
Finally, from a theoretical perspective, our results provide direct motivation for extending analytical descriptions of hyper-networks---and of processes taking place on them---to include more nuanced structural correlation patterns.

%% file: sections/algos.tex
\section{Sampling Algorithms}

This section describes two efficient sampling algorithms, \algoA and \algoC, designed for sampling from $\nullset^{\nameA}$ and $\nullset^{\nameC}$, respectively.
Both algorithms leverage the Metropolis-Hastings algorithm as part of the Markov Chain Monte Carlo approach, and employ targeted edge swap operations to traverse the Markov graph.
\algoA uses Parity Swap Operations (\Cref{th:pso}), while \algoC uses Restricted Parity Swap Operations (\Cref{th:rpso}). 
The sampling procedures for both algorithms are illustrated through pseudocode and detailed in \Cref{ax:algoa} and \Cref{ax:algoc}, respectively.
Finally, we experimentally study the mixing time of the samplers in \Cref{ax:convergence}.
The code is publicly available on GitHub.\footnote{\url{https://github.com/lady-bluecopper/NuDHy}}

\subsection{\algoA: An Efficient Sampler for \nameA}\label{sec:deg_pres}

\begin{figure}[th!]
	\centering
	\includegraphics[width=\columnwidth]{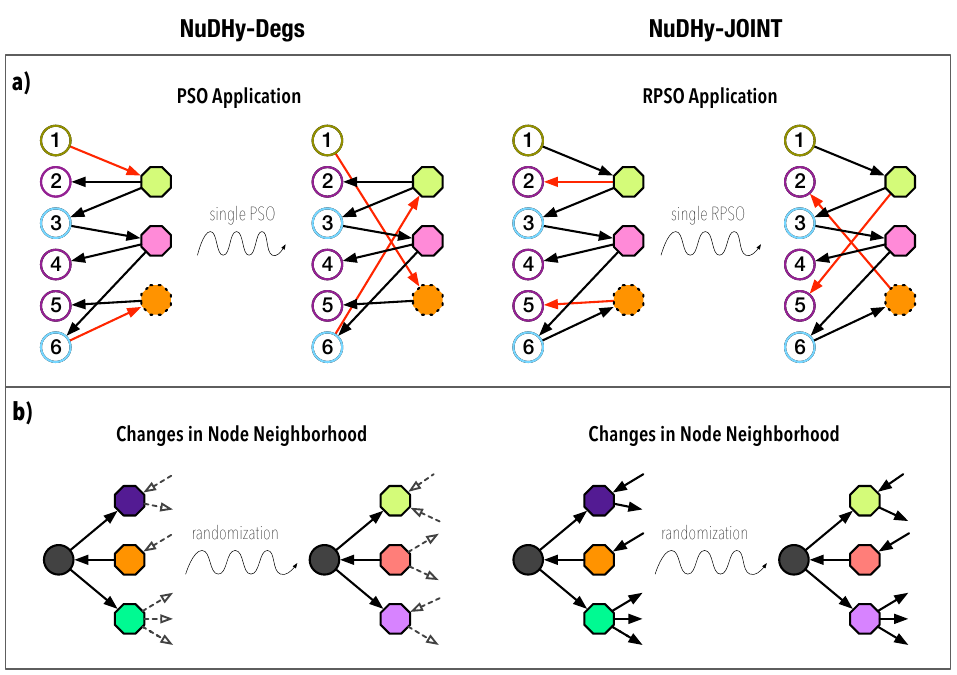}
	\caption{
	\textbf{a)} Bipartite graphs obtained from \Cref{fig:toy}\textbf{a} after the application of the PSO $(1, \protect\hedgeg, +1), (6, \protect\hedgeo, +1) \pso (1, \protect\hedgeo, +1), (6, \protect\hedgeg, +1)$ 
  and of the RPSO $(2, \protect\hedgeg, -1), (5, \protect\hedgeo, -1) \rpso (2, \protect\hedgeo, -1), (5, \protect\hedgeg, -1)$. 
  The edges involved in the swap operations are highlighted in red. Left nodes with the same in- and out-degree are outlined with the same color. Right nodes with the same in- and out-degree are outlined with the same pattern.
	\textbf{b)} Changes in the neighborhood of a left node after the application of a sequence of PSOs and of RPSOs. PSOs preserve the number of in-going and out-going edges of each node. RPSOs preserve also the in- and out-degree of the nodes connected to each node.
	}
	\label{fig:toy_2}
\end{figure}

We present a Markov Chain Monte Carlo algorithm, dubbed \algoA, that uses Metropolis-Hastings (MH) to sample from $\nullset^{\nameA}$ according to $\nullprob$.
We first define an edge swap operation that transforms a bipartite graph into another bipartite graph while preserving the degree sequences, and then describe the state space that this operation induces and over which the Markov chain is constructed.

\begin{lemma}[Parity Swap Operation, PSO]\label{th:pso}
Let $G \doteq (L, R, D)$ be a directed bipartite graph and $u \neq v \in L$, $\alpha \neq \beta \in R$ such that $\exists d \in \{+1,-1\}$ for which $e_1 \doteq (u,\alpha,d), e_2 \doteq (v,\beta,d) \in D$ and $e_3 \doteq (u,\beta,d), e_4 \doteq (v,\alpha,d) \notin D$.
Swapping $e_1$, $e_2$ with $e_3$, $e_4$ generates a directed bipartite graph $G'=(L, R, (D \setminus \{e_1, e_2\}) \cup \{e_3, e_4\})$ with the same left and right, in- and out-degree sequences as $G$. 
This swap operation, denoted as $e_1, e_2 \pso e_3, e_4$, is called \emph{parity swap operation} (PSO).
\end{lemma}
For directed unipartite graphs, this operation is known as \emph{checkerboard swap}~\cite{artzy2005generating}.
An example of PSO is shown in \Cref{fig:toy_2}\textbf{a} (left).

\begin{algorithm}[t!]
  \small
    \caption{\algoA}\label{alg:nudhy-a}
    \begin{algorithmic}[1]
    \Require Graph $G \doteq (L, R, D) \in \states^{\nameA}$, Number of Steps $s$
    \Ensure Graph sampled uniformly from $\states^{\nameA}$ 
    \RepeatN{$s$}
	  \State $\mathsf{out} \gets$ flip a biased coin with heads probability $|\fromL|/|D|$
	  \If{$\mathsf{out}$ is \emph{heads}}
	    \State $u$, $v$ $\gets$ distinct vertices drawn u.a.r. from $\posL$ \Comment{\begin{footnotesize} vertices in $L$ with out-degree $>0$ \end{footnotesize}}
	    \If{$\setdiff{u}{v} \neq \varnothing$} \Comment{\begin{footnotesize} $\{(\alpha,\beta) \suchthat \alpha \in \outneighs{u} \setminus \outneighs{v} \, \wedge \beta \in \outneighs{v} \setminus \outneighs{u}\}$ \end{footnotesize}}
	      \State $\alpha$, $\beta$ $\gets$ pair drawn u.a.r. from $\setdiff{u}{v}$
	      \State perform $(u,\alpha,+1), (v,\beta, +1) \pso (u, \beta, +1), (v, \alpha, +1)$ on $G$
	    \EndIf
	  \Else
	    \State $\alpha$, $\beta$ $\gets$ distinct vertices drawn u.a.r. from $\posR$ \Comment{\begin{footnotesize} vertices in $R$ with out-degree $>0$\end{footnotesize}}
	    \If{$\setdiff{\alpha}{\beta} \neq \varnothing$} \Comment{\begin{footnotesize} $\{(u,v) \suchthat u \in \outneighs{\alpha} \setminus \outneighs{\beta} \, \wedge v \in \outneighs{\beta} \setminus \outneighs{\alpha}\}$\end{footnotesize}}
	      \State $u$, $v$ $\gets$ pair drawn u.a.r. from $\setdiff{\alpha}{\beta}$
	      \State perform $(u,\alpha,-1), (v,\beta, -1) \pso (u, \beta, -1), (v, \alpha, -1)$ on $G$
	    \EndIf
	  \EndIf
    \End
    \State \Return $G$
    \end{algorithmic}
\end{algorithm}

The state space $\states^{\nameA}$ is a directed weighted graph.
Each vertex represents a bipartite directed graph with the same left and right, in- and out-degree sequences as $\observedbip$.
Each edge connects two graphs that can be transformed into each other via a PSO.
For any pair of graphs, there is at most one PSO that connects them, hence there are no parallel edges between vertices.
Moreover, we add self-loops from each vertex to itself.
All the graphs $G'$ that can be obtained by applying a PSO to $G$ are called the \emph{neighbors} of $G$ in $\states^{\nameA}$. 
We associate a weight $\neighdistr{G}(G')$ to each edge $(G,G')$ that represents the probability of transitioning to $G'$ starting from $G$.

A fundamental theorem of Markov chains states that an irreducible, aperiodic, finite Markov chain has a unique stationary distribution~\cite{rao2020principles}.
Therefore, the Markov chain converges to $\nullprob$ independently of the starting state.
Furthermore, we know that if the transition probability matrix is doubly stochastic, the Markov chain converges to the uniform distribution over its state space.\footnote{$\pi_G = \frac{1}{n}$ is stationary for all $G$ because $\left[\pi \cdot \neighdistr{G}\right]_G = \sum\limits_{G'}\pi_{G'}\neighdistr{G}(G') = \sum\limits_{G'}\frac{1}{n}\neighdistr{G}(G')=\frac{1}{n}\sum\limits_{G'}\neighdistr{G}(G')=\frac{1}{n}=\pi_G$.}
As a result, samples from the chain can be considered as uniform samples from the state space.
In our case, aperiodicity is guaranteed by the presence of self-loops over the vertices, while the double stochasticity of the transition matrix can be inferred from observing that \emph{(i)} each PSO is reversible, i.e., if $e_1, e_2 \pso e_3, e_4$ transforms $G$ into $G'$, then $e_3, e_4 \pso e_1, e_2$ transforms $G'$ into $G$; and that \emph{(ii)} the probability of going from $G$ to $G'$ is equal to the probability of going from $G'$ to $G$, i.e., $\neighdistr{G}(G') = \neighdistr{G'}(G)$. 
The definition of $\neighdistr{G}$ and the proof that the transition matrix $\{\neighdistr{G}(G')\}$ is doubly stochastic can be found in \Cref{ax:algoa}.
Finally, \Cref{ax:proofs} proves irreducibility by showing that $\states^{\nameA}$ is strongly connected.
From these results, we obtain that the stationary distribution is the uniform distribution.

\Cref{alg:nudhy-a} illustrates the sampling procedure of \algoA.
The algorithm performs a number of steps $s$ (input parameter) in the state space large enough that its output can be considered as a uniform sample from $\states^{\nameA}$.
Previous works has shown that $s = O\left(|E|\log\left(|E|\right)\right)$ is, in general, sufficient~\cite{viger2016efficient}.

\subsection{\algoC: An Efficient Sampler for $\nameC$}\label{sec:bjidt_pres}

We introduce an edge swap operation that transforms a bipartite graph into another bipartite graph with the same JOINT.

\begin{lemma}[Restricted Parity Swap Operation, RPSO]\label{th:rpso}
Let $G \doteq (L, R, D)$ be a directed bipartite graph and $u \neq v \in L$, $\alpha \neq \beta \in R$ such that $\exists d \in \{+1,-1\}$ for which $e_1 \doteq (u,\alpha,d), e_2 \doteq (v,\beta,d) \in D$ and $e_3 \doteq (u,\beta,d), e_4 \doteq (v,\alpha,d) \notin D$.

If 
$\indegG{u} = \indegG{v}\, \wedge \, \outdegG{u} = \outdegG{v} \;\vee\; \indegG{\alpha} = \indegG{\beta}\, \wedge \, \outdegG{\alpha} = \outdegG{\beta}$,
then swapping $e_1$, $e_2$ with $e_3$, $e_4$ generates a directed bipartite graph $G'=(L, R, (D \setminus \{e_1, e_2\}) \cup \{e_3, e_4\})$ with the same JOINT as $G$. 
This swap operation, denoted as $e_1, e_2 \rpso e_3, e_4$, is called \emph{restricted parity swap operation} (RPSO).
\end{lemma}

An example of RPSO is shown in \Cref{fig:toy_2}\textbf{a} (right).

\begin{algorithm}[!t]
  \small
    \caption{\algoC}\label{alg:nudhy-c}
    \begin{algorithmic}[1]
    \Require Graph $G \doteq (L, R, D) \in \states^{\nameC}$, Number of Steps $s$
    \Ensure Graph sampled uniformly from $\states^{\nameC}$ 
    \RepeatN{$s$}
      \State $\mathsf{out} \gets$ flip a biased coin with heads prob $|\fromL|/|D|$
      \State $\mathsf{out2} \gets$ flip a fair coin
      \If{$\mathsf{out2}$ is \emph{heads}}
      	\State $u$, $v$ $\gets$ different random left vertices with the same in- and out-degree
	    \If{$\mathsf{out}$ is \emph{heads} \textbf{and} $\setdiff{u}{v} \neq \varnothing$} \Comment{\begin{scriptsize} $\setdiff{u}{v} \doteq \{(\alpha,\beta) \suchthat \alpha \in \outneighs{u} \setminus \outneighs{v} \, \wedge \beta \in \outneighs{v} \setminus \outneighs{u}\}$\end{scriptsize}}
          	\State $\alpha$, $\beta$ $\gets$ pair drawn u.a.r. from $\setdiff{u}{v}$
          	\State perform $(u,\alpha, +1), (v,\beta, +1) \rpso (u, \beta, +1), (v, \alpha, +1)$ on $G$
	    \ElsIf{$\mathsf{out}$ is \emph{tails} \textbf{and} $\setidiff{u}{v} \neq \varnothing$} \Comment{\begin{scriptsize} $\setidiff{u}{v} \doteq \{(\alpha, \beta) \suchthat \alpha \in \inneighs{u} \setminus \inneighs{v} \, \wedge \, \beta \in \inneighs{v} \setminus \inneighs{u}\}$\end{scriptsize}}
          	\State $\alpha$, $\beta$ $\gets$ pair drawn u.a.r. from $\setidiff{u}{v}$
          	\State perform $(u,\alpha,-1), (v,\beta, -1) \rpso (u, \beta, -1), (v, \alpha, -1)$ on $G$
	    \EndIf
      \Else
      	\State $\alpha$, $\beta$ $\gets$ different random right vertices with the same in- and out-degree
      	\If{$\mathsf{out}$ is \emph{heads} \textbf{and} $\setidiff{\alpha}{\beta} \neq \varnothing$}
      		\State $u$, $v$ $\gets$ pair drawn u.a.r. from $\setidiff{\alpha}{\beta}$
      		\State perform $(u,\alpha, +1), (v,\beta, +1) \rpso (u, \beta, +1), (v, \alpha, +1)$ on $G$
      	\ElsIf{$\mathsf{out}$ is \emph{tails} \textbf{and} $\setdiff{\alpha}{\beta} \neq \varnothing$}
      		\State $u$, $v$ $\gets$ pair drawn u.a.r. from $\setdiff{\alpha}{\beta}$
      		\State perform $(u,\alpha,-1), (v,\beta, -1) \rpso (u, \beta, -1), (v, \alpha, -1)$ on $G$
      	\EndIf
      \EndIf
    \End
    \State \Return $G$
    \end{algorithmic}
\end{algorithm}

The state space $\states^{\nameC}$ is a directed weighted graph where each vertex is a bipartite directed graph with the same JOINT of $\observedbip$, and edges connect graphs that can be transformed into each other via an RPSO.
For each pair of vertices, there is at most one RPSO that can transform the first one into the second one, and self-loops are added to guarantee that the Markov chain is aperiodic.
\Cref{ax:algoc} defines a transition probability distribution $\neighdistr{G}$ over the set of neighbors of any $G \in \states^{\nameC}$ and proves that $\sum_{G' \in \states^{\nameC}}\neighdistr{G}(G') = 1$.
By observing that each RPSO is reversible and that the number of common in- and out-neighbors between any pair of nodes does not change after the application of an RPSO, we have that $\neighdistr{G}(G') = \neighdistr{G'}(G)$ and that the transition matrix $\{\neighdistr{G}(G')\}$ is doubly stochastic.
Finally, in \Cref{ax:proofs} we prove irreducibility by showing that $\states^{\nameC}$ is strongly connected.
From these results, we obtain that the stationary distribution is the uniform distribution.
\Cref{alg:nudhy-c} illustrates the sampling procedure of \algoC.

%% file: sections/proofs.tex
\section{Proofs}\label{ax:proofs}

\begin{lemma}\label{thm:nudhya_connected}
$\states^{\nameA}$ is strongly connected via PSOs.
\end{lemma}  
\begin{proof}
We prove that $\states^{\nameA}$ is strongly connected, by showing that any pair of distinct graphs $G_1 \doteq (L, R, D_1)$ and $G_2 \doteq (L, R, D_2)$ in $\states^{\nameA}$ can be transformed into one another through a sequence of PSOs.

We recall that PSOs swap only edges with the same direction.
Hence, we can independently address the edges in $\fromL_1$ and $\fromL_2$, and the edges in $\fromR_1$ and $\fromR_2$.
Given $\fromL_1$ and $\fromL_2$ (resp. $\fromR_1$ and $\fromR_2$), we construct a \emph{canonical path} from $\fromL_1$ to $\fromL_2$ (resp. from $\fromR_1$ to $\fromR_2$), following the procedure outlined in Section 2.1 of \cite{kannan1999simple}.
This canonical path delineates a series of \emph{switches}, i.e., edge swaps between pairs of edges.
Although the original procedure was tailored for undirected bipartite graphs with the same degree distribution, we can extend it to our scenario.
Since the edges in $\fromL_1$ and $\fromL_2$ (resp. $\fromR_1$ and $\fromR_2$) have the same direction, we can see $(L, R, \fromL_1)$ and $(L, R, \fromL_2)$ as undirected bipartite graphs.
Consequently, each PSO involving edges in $\fromL_1$ and $\fromL_2$ (resp. $\fromR_1$ and $\fromR_2$) corresponds to a switch operation.

Therefore, the sequence of PSOs transforming $G_1$ into $G_2$ simply consists in the concatenation of the switch sequence from $\fromL_1$ to $\fromL_2$, and the switch sequence from $\fromR_1$ to $\fromR_2$.
\end{proof}

\begin{lemma}\label{thm:nudhyc_connected}
$\states^{\nameC}$ is strongly connected via RPSOs.
\end{lemma}  
\begin{proof}
We prove that $\states^{\nameC}$ is strongly connected, by adapting the proof for the case of undirected bipartite graphs, proposed in \cite{czabarka2015realizations}.
Let $\mathcal{C}$ be the list of distinct tuples of in- and out-degrees of vertices in $G$, i.e., $c_i \doteq (k,l)$ iff $\exists v \in L, R \suchthat \indegG{v} = k \wedge \outdegG{v} = l$.
We define the \emph{degree spectrum} of $v \in L, R$ as the vector $s_G(v)$ where $s_G(v)[i]$ is the number of vertices with in-degree $c_i[0]$ and out-degree $c_i[1]$ to which $v$ is connected.
We partition $L$ into $V^L_1, \dots, V^L_{|\mathcal{C}|}$ such that each $V^L_i$ contains the vertices in $L$ with in-degree $c_i[0]$ and out-degree $c_i[1]$.
Similarly, we partition $R$ into $V^R_1, \dots, V^R_{|\mathcal{C}|}$.

For each $j$ such that $\left|V^L_j\right| \neq 0$ and each $i$, we set
\[
A^L_j(i) \doteq \frac{\biot{G}[c_j[0], c_j[1], c_i[0], c_i[1], +1]}{\left|V^L_j\right|}\enspace.
\]
Similarly, for each $j$ such that $\left|V^R_j\right| \neq 0$ and each $i$, we set
\[
A^R_j(i) \doteq \frac{\biot{G}[c_i[0], c_i[1], c_j[0], c_j[1], -1]}{|V^R_j|}\enspace.
\]
We say that $V^L_j$ (resp. $V^R_j$) is \emph{balanced} in $G$, if it is empty, or for each $V^R_i$ (resp. $V^L_i$) the edges connecting $V^L_j$ to $V^R_i$ (resp. $V^R_j$ to $V^L_i$) are as uniformly distributed on $V^L_j$ as possible. 
This happens when for all $v \in V^L_j$ and all $i$, $s_G(v)[i] \in \left\{\left\lfloor A^L_j(i)\right\rfloor, \left\lceil A^L_j(i)\right\rceil\right\}$. 
We say that $G$ is \emph{left-balanced} (resp. \emph{right-balanced}) if $V^L_i$ (resp. $V^R_i$) is balanced for all $i$.
Finally, for each $v \in L$ and $i$ such that $V^R_i \neq \varnothing$, we define $c^L_G(v,i) \doteq \left\lfloor\left|A^L_j(i) - s_G(v)[i]\right|\right\rfloor$, where $c_j = \left(\indegG{v}, \outdegG{v}\right)$, and for each $j$ we define $C_G^L(j) \doteq \sum_{v \in V^L_j}\sum_i c^L_G(v,i)$.
Similarly, for each $v \in R$ and $i$ such that $V^L_i \neq \varnothing$, we define $c^R_G(v,i)$ and for each $j$ we define $C_G^R(j)$.

We observe that we cannot swap two edges $(u,\alpha,d)$ and $(v,\beta,d')$ such that $d \neq d'$, because the resulting graph would have different in- and out-degree sequences. For example, if $D = \{(u,\alpha,+1)$ $,$ $(v,\beta,-1)\}$, by swapping $(u,\alpha,+1)$ and $(v,\beta,-1)$, the in-degree of $\alpha$ goes from $1$ to $0$, while its out-degree goes from $0$ to $1$.
As a result, we can treat the case where we swap edges with direction $1$ and the case where we swap edges with direction $-1$ independently.

To prove that $\states^{\nameC}$ is strongly connected, we first need to prove the following lemma for $C^L_G$. 
The proof for $C^R_G$ follows the same steps.
\begin{lemma}[Lemma 4~\cite{czabarka2015realizations}]\label{lemma:lemma4}
If $C^L_G(j) \neq 0$, then there exists $u,v \in V^L_j$ and a RPSO $(u,\alpha,+1)$ $,$ $(v,\beta,+1)$ $\rpso$ $(u,\beta,+1), (v,\alpha,+1)$ transforming $G$ into $G'$ such that $C^L_{G'}(j) < C^L_G(j)$ and for all $l \neq j$ $C^L_{G'}(l) = C^L_G(l)$.
\end{lemma}
\begin{proof}
We choose $u,v \in V^L_j$ such that $s_G(u)[i]$ is minimal and $s_G(v)[i]$ is maximal among the vertices in $V^L_j$.
Since $u$ has fewer neighbors in $V^R_i$ than $v$, there exists $\beta \in V^R_i$ such that $(v,\beta,+1) \in D$ but $(u,\beta,+1) \notin D$.
Since $u$ and $v$ have the same out-degrees and $s_G(v)[i] > s_G(u)[i]$, there exists a $k \neq i$ such that $s_G(u)[k] > s_G(v)[k]$, and hence there exists $\alpha \in V^R_k$ such that $(u,\alpha,+1) \in D$ but $(v,\alpha,+1) \notin D$.
Therefore, $(u,\alpha,+1)$ $,$ $(v,\beta,+1) \rpso (v,\alpha,+1)$ $,$ $(u,\beta,+1)$ is a RPSO.
Similarly to \cite{czabarka2015realizations}, by applying such RPSO we obtain a graph $G'$ such that $C^L_{G'}(j) < C^L_G(j)$.
Since the RPSO involves only left vertices with in- and out-degree combination $c_j$, all the $C^L_G(l)$ for $l \neq j$ remain unchanged, i.e., $C^L_{G'}(l) = C^L_G(l)$.
\end{proof}

Thanks to \Cref{lemma:lemma4}, we have the following corollary:
\begin{corollary}[Corollary 5~\cite{czabarka2015realizations}]\label{cor:cor5}
Given a directed bipartite graph $G \in \states^{\nameC}$, there exists a series of RPSOs transforming $G$ into a left-balanced (resp. right-balanced) graph $G' \in \states^{\nameC}$.
\end{corollary} 
\begin{proof}
We follow the proof of \cite{czabarka2015realizations}, which shows that successive applications of \Cref{lemma:lemma4} for each $i$ such that $C^L_{G}(i) \neq 0$ (resp. $C^R_{G}(i) \neq 0$), give a sequence of RPSOs that transforms $G$ into a left-balanced (resp. right-balanced) graph.
\end{proof}
We observe that when balancing the left side of $G$, we do not affect the values $C^R_{G}(i)$, and similarly, when balancing the right side of $G$, we do not affect the values $C^L_{G}(i)$.
As a result, we can first apply \Cref{cor:cor5} to transform $G$ into a left-balanced graph $G'$, and then apply \Cref{cor:cor5} to transform $G'$ into a right-balanced graph $G''$ that is left-balanced as well.

To prove that left- and right-balanced graphs can be connected via RPSOs, we first define two auxiliary bipartite graphs $\mathcal{A}^L(G,j) \doteq (U^L,P^L;E^L)$ and $\mathcal{A}^R(G,j) \doteq (U^R,P^R;E^R)$, such that $U^L$ (resp. $U^R$) contains a vertex $u_v$ for each $v \in V_j^L$ (resp. $V_j^R$), $P^L$ (resp. $P^R$) contains a vertex $p_i$ for each $V_i^R$ (resp. $V_i^L$) such that $A^L_j(i) \notin \mathbb{N}$ (resp. $A^R_j(i) \notin \mathbb{N}$), and $E^L$ (resp. $E^R$) contains an edge $(u_v,p_i)$ for each $v$ such that $s_G(v)[i] = \lfloor A^L_j(i) \rfloor + 1$ (resp. $s_G(v)[i] = \lfloor A^R_j(i) \rfloor + 1$). 
Then, we define the \emph{bipartite swap operation} (BSO) as the operation that, given an undirected bipartite graph $B \doteq (V,E)$, takes two edges $(u,a)$, $(v,b)$ $\in E$ such that $(u,b)$, $(v,a)$ $\notin E$, and generates a new bipartite graph $B' = (V, E \setminus \{(u,a), (v,b)\} \cup \{(u,b), (v,a)\})$. It can be easily seen that BSOs preserve the left and right degree distributions of $B$.

Graphs in $\states^{\nameC}$ and their auxiliary graphs are linked by the following lemma:
\begin{lemma}[Lemma 6~\cite{czabarka2015realizations}]\label{lemma:lemma6}
If there is a BSO transforming $\mathcal{A}^L(G,j)$ (resp. $\mathcal{A}^R(G,j)$) into $\mathcal{\hat{A}}^L(G,j)$, then there is a RPSO transforming $G$ into $G'$ such that $\mathcal{A}^L(G',j) = \mathcal{\hat{A}}^L(G,j)$ (resp. $\mathcal{A}^R(G',j) = \mathcal{\hat{A}}^R(G,j)$), and $s_G(v) = s_{G'}(v)$ for each vertex $v \notin V_j^L$ (resp. $v \notin V_j^R$).
\end{lemma}
\begin{proof}
We can straightforwardly use the proof in \cite{czabarka2015realizations} to prove the lemma for both the $\mathcal{A}^L(G,j)$ and the $\mathcal{A}^R(G,j)$ case.
\end{proof}

Thanks to \Cref{lemma:lemma6}, we can prove the following theorem, which establishes a relation between any two left-balanced (resp. right-balanced) graphs in $\states^{\nameC}$ :
\begin{theorem}[Theorem 7~\cite{czabarka2015realizations}]\label{th:th7}
If $G_1$ and $G_2$ are two left-balanced (resp. right-balanced) graphs in $\states^{\nameC}$, then there is a series of RPSOs transforming $G_1$ into $G_1'$, such that $s_{G_1'}(v) = s_{G_2}(v)$ for each $v \in L$ (resp. $v \in R$).
\end{theorem}
\begin{proof}
We prove the theorem for the left-balanced case. The other case can be handled in a similar way.
We construct a sequence of graphs $G^0_1, G^1_1, G^2_1, \dots, G^{|\mathcal{C}|}_1$ with $G_1^0 = G_1$ such that for each $i$ there is a sequence of RPSOs transforming $G^{i-1}_1$ to $G^i_1$ such that $s_{G^i_1}(v) = s_{G_2}(v)$ for each $v \in V^L_i$ and $s_{G^i_1}(v) = s_{G^{i-1}_1}(v)$ for each $v \notin V^L_i$. 
Since \textbf{(i)} all the vertices in $V^L_i$ have the same out-degree and \textbf{(ii)} $G_1$ and $G_2$ have the same JOINT, we have that $\mathcal{A}^L(G^{i-1}_1, i)$ and $\mathcal{A}^L(G_2, i)$ have the same degree sequences.
Therefore, we can apply the Ryser’s Theorem~\cite{ryser1963combinatorial} to obtain a sequence of BSOs transforming one into the other. 
Thanks to \Cref{lemma:lemma6}, we know that there exists a sequence of RPSOs transforming $G_1^{i-1}$ into the desired $G^i_1$.
The proof concludes by setting $G_1' = G^{|\mathcal{C}|}_1$.
\end{proof}

Given two left- and right-balanced graphs $G_1$ and $G_2$, we can apply \Cref{th:th7} two times: the first application transforms $G_1$ into $G_1'$ such that $s_{G_1'}(v) = s_{G_2}(v)$ for each $v \in L$, while the second application transforms $G_1'$ into $G_1''$ such that $s_{G_1''}(v) = s_{G_2}(v)$ for each $v \in R$.
We observe that the second application of the theorem does not affect the values $s_{G_1'}(v)$ for $v \in L$ because it acts only on the outgoing edges of vertices in $R$.

\begin{figure}[!t]
  \centering
  \includegraphics[width=\columnwidth]{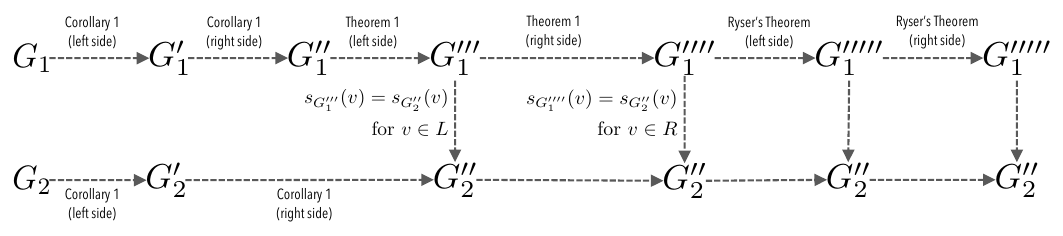}
  \caption{Steps to prove that $\states^{\nameC}$ is strongly connected.}
  \label{fig:proofN2}
\end{figure}

We can now prove that $\states^{\nameC}$ is strongly connected.
The main steps of the proof are illustrated in \Cref{fig:proofN2}.
Let $G_1$ and $G_2$ two non-isomorphic graphs in $\states^{\nameC}$.
We apply \Cref{cor:cor5} two times to transform $G_1$ and $G_2$ into left- and right-balanced graphs $G_1''$ and $G_2''$.
Next, we apply \Cref{th:th7} to transform $G_1''$ into a left-balanced realization $G_1'''$ such that $s_{G_1'''}(v) = s_{G_2''}(v)$ for each $v \in L$.
Next, we we apply \Cref{th:th7} to transform $G_1'''$ into a left- and right-balanced realization $G_1''''$ such that $s_{G_1''''}(v) = s_{G_2''}(v)$ for each $v \in R$.
We note that, the application of the theorem on the right side does not affect the values $s_{G_1'''}(u)$ for $u \in L$, and hence $G_1''''$ is both left and right-balanced. 
For $i, j$, let $G_{1,L}^{ij}$ and $G_{2,L}^{ij}$ the bipartite graphs consisting of edges in $G_1''''$ and $G_2''$, respectively, from vertices in $V^L_i$ to vertices in $V^R_j$.
Similarly, let $G_{1,R}^{ij}$ and $G_{2,R}^{ij}$ the bipartite graphs consisting of edges in $G_1''''$ and $G_2''$, respectively, from vertices in $V^R_i$ to vertices in $V^L_j$.
Let $T \in \{L, R\}$.
Since $s_{G_1''''}(v) = s_{G_2''}(v)$ for all $v \in T$, the degree sequences of $G_{1,T}^{ij}$ and $G_{2,T}^{ij}$ are the same.
Moreover, since all the left vertices in $G_{1,T}^{ij}$ and $G_{2,T}^{ij}$ have the same combination of in- and out-degree, a BSO in $G_{1,T}^{ij}$ is a RPSO in $G_1''''$.
Therefore, for each $i,j$ and each $T \in \{L,R\}$, we can apply the Ryser’s theorem~\cite{ryser1963combinatorial} to obtain a sequence of BSOs transforming $G_{1,T}^{ij}$ into $G_{2,T}^{ij}$, hence obtaining a sequence of RPSOs transforming $G_1''''$ into $G_2''$.
\end{proof}

%% file: sections/algodetails.tex
\section{Algorithmic Details of \algoA}\label{ax:algoa}

Let $\posL \doteq \{v \in L \suchthat \outdeg{G}{v} > 0\}$ the subset of left vertices with out-going edges,
$\posR \doteq \{\alpha \in R \suchthat \outdeg{G}{\alpha} > 0\}$ the subset of right vertices with out-going edges, and
$\setdiff{u}{v} \doteq \{(\alpha,\beta) \suchthat \alpha \in \outneighs{u} \setminus \outneighs{v} \, \wedge \beta \in \outneighs{v} \setminus \outneighs{u}\}$ be the set of pairs of out-neighbors of $u$ and $v$ that are not out-neighbors of $v$ and $u$, respectively.

\begin{algorithm}[ht!]
  \small
    \caption{\algoA}\label{alg:nudhy-a-details}
    \begin{algorithmic}[1]
    \Require Graph $G \doteq (L, R, D) \in \states^{\nameA}$, Number of Steps $s$
    \Ensure Graph sampled uniformly from $\states^{\nameA}$ 
    \RepeatN{$s$}
      \State $\mathsf{out} \gets$ flip a biased coin with heads probability $|\fromL|/|D|$
      \If{$\mathsf{out}$ is \emph{heads}}
        \State $u$, $v$ $\gets$ different vertices drawn u.a.r. from $\posL$
        \If{$|\setdiff{u}{v}| = 0$}
          \textbf{continue}\label{line:sl1} \Comment{\textit{self-loop}}
        \Else
          \State $\alpha$, $\beta$ $\gets$ pair drawn u.a.r. from $\setdiff{u}{v}$
          \State $d \gets +1$
        \EndIf
      \Else
        \State $\alpha$, $\beta$ $\gets$ different vertices drawn u.a.r. from $\posR$
        \If{$|\setdiff{\alpha}{\beta}| = 0$}
          \textbf{continue}\label{line:sl2} \Comment{\textit{self-loop}}
        \Else
          \State $u$, $v$ $\gets$ pair drawn u.a.r. from $\setdiff{\alpha}{\beta}$
          \State $d \gets -1$
        \EndIf
      \EndIf
      \State $G \gets $perform $(u,\alpha,d), (v,\beta, d) \pso (u, \beta, d), (v, \alpha, d)$ on $G$\label{line:apply} \Comment{\textit{transition always accepted}}
    \End
    \State \Return $G$
    \end{algorithmic}
\end{algorithm}
 
To sample a neighbor $G'$ of $G$, we first flip a biased coin that outputs \emph{heads} with probability $|\fromL|/|D|$ and \emph{tails} with probability $|\fromR|/|D|$\footnote{Any other probability can be used; the idea is to prefer the direction for which there are more valid PSOs.}.
If the outcome is \emph{heads}, we draw a pair of different vertices $u,v \in \posL$ uniformly at random between all pairs.
If $\setdiff{u}{v} = \varnothing$, we set $G' = G$ (self-loop).
Otherwise, we draw $(\alpha,\beta)$ from $\setdiff{u}{v}$ uniformly at random.
By construction, $(u,\alpha,+1), (v,\beta,+1) \pso (u,\beta,+1), (v,\alpha,+1)$ is a PSO, and thus we can set $G' = (L, R, D \setminus \{(u,\alpha,+1), (v,\beta,+1)\} \cup \{(u,\beta,+1), (v,\alpha,+1)\})$.
If the outcome is \emph{tails}, we draw a pair of different vertices $\alpha,\beta \in \posR$ uniformly at random between all pairs, and then follow the same procedure described for the \emph{heads} case.
This procedure induces a probability distribution $\neighdistr{G}$ over the set of neighbors of $G$.
Each directed edge $(G, G')$ in $\states^{\nameA}$ has thus weight $\neighdistr{G}{G'}$.
Let $(u,\alpha,d), (v,\beta,d) \pso (u,\beta,d), (v,\alpha,d)$ be the sampled PSO, and let $G'$ be the graph obtained by performing such PSO on $G$.

If $d = +1$, then
\begin{equation}\label{eq:neighdistleft}
\neighdistr{G}(G') = \frac{|\fromL|}{|D|}
           {\displaystyle\binom{|\posL|}{2}^{-1}}
                     \left(\left|\setdiff{u}{v}\right|\right)^{-1}\,.
\end{equation}

If $d = -1$, then 
    \begin{equation}\label{eq:neighdistright}
    \neighdistr{G}(G') = \frac{|\fromR|}{|D|}
               {\displaystyle\binom{|\posR|}{2}^{-1}}
                         \left(\left|\setdiff{\alpha}{\beta}\right|\right)^{-1}\,.
    \end{equation}

Let $G \in \states^{\nameA}$ and $G'$ a neighbor of $G$ chosen according to $\neighdistr{G}$.
A MH algorithm accepts the transition from a state $G$ to a new state $G'$ with probability 
$\min\left\{1, \frac{\nullprob(G')\neighdistr{G'}(G)}{\nullprob(G)\neighdistr{G}(G')}\right\}$,
otherwise, it sets $G' = G$.
However, in our case, $\alpha \in \neighsplus{G}{u} \setminus \neighsplus{G}{v}$ implies that $\left|\neighsplus{G}{u} \setminus \neighsplus{G}{v}\right| = \left|\neighsplus{G'}{u} \setminus \neighsplus{G'}{v}\right|$, and $\beta \in \neighsplus{G}{v} \setminus \neighsplus{G}{u}$ implies that $|\neighsplus{G}{v} \setminus \neighsplus{G}{u}| = |\neighsplus{G'}{v} \setminus \neighsplus{G'}{u}|$.
As a consequence, it holds that $\neighdistr{G}(G') = \neighdistr{G'}(G)$, and thus $\neighdistr{G}(G')/\neighdistr{G'}(G) = 1$.
We next show that $\sum_{G' \in \states^{\nameA}}\neighdistr{G}(G') = 1$, which ensures us that the transition matrix $\{\neighdistr{G}(G')\}$ is doubly stochastic.
From these results, we obtain that the stationary distribution is the uniform distribution, and thus $\nullprob(G')/\nullprob(G) = 1$. This simplifies our use of MH, as the algorithm accepts the transition to the new state with probability $1$.

Let $\mathsf{k}^L$ (resp. $\mathsf{k}^R$) be the set of pairs of left (resp. right) vertices $u,v$ such that $\setdiff{u}{v} \neq \varnothing$,
and let $\mathcal{N}^L(u,v)$ (resp. $\mathcal{N}^R(u,v)$) be the set of graphs $G'$ that can be obtained from $G$ by applying a PSO consisting in edges with direction $+1$ (resp. $-1$) and involving the pair of vertices $(u,v) \in \mathsf{k}^L$ (resp. $(u,v) \in \mathsf{k}^R$).
We recall that there is only one of such pairs for each graph $G'$ adjacent to $G$.
Each $(u,v) \in \mathsf{k}^L$ is part of $\left|\setdiff{u}{v}\right|$ PSOs, and thus
$\sum_{G' \in \mathcal{N}^L(u,v)}\neighdistr{G}(G') = \frac{|\fromL|}{|D|}{\displaystyle\binom{|\posL|}{2}^{-1}}$. 
Similarly, each $(u,v) \in \mathsf{k}^R$ is part of $\left|\setdiff{u}{v}\right|$ PSOs, and thus
$\sum_{G' \in \mathcal{N}^R(u,v)}\neighdistr{G}(G') = \frac{|\fromR|}{|D|}{\displaystyle\binom{|\posR|}{2}^{-1}}$.

Each time we sample a pair of left (resp. right) vertices not in $\mathsf{k}^L$ (resp. $\mathsf{k}^R$) we remain in the state $G$, and thus
\[\neighdistr{G}(G) = \frac{|\fromL|}{|D|}\frac{\displaystyle\binom{|\posL|}{2} - |\mathsf{k}^L|}{\displaystyle\binom{|\posL|}{2}} + \frac{|\fromR|}{|D|}\frac{\displaystyle\binom{|\posR|}{2} - |\mathsf{k}^R|}{\displaystyle\binom{|\posR|}{2}}\,.\]

Finally, for each $G'$ not adjacent to $G$, we have that $\neighdistr{G}(G') = 0$.
Summing all these terms, we obtain our result:

\begin{small}
\begin{align*}
\sum_{G' \in \states^{\nameA}}\neighdistr{G}(G') =& \sum_{(u,v) \in \mathsf{k}^L}\frac{|\fromL|}{|D|}{\displaystyle\binom{|\posL|}{2}^{-1}}
+ \sum_{(u,v) \in \mathsf{k}^R}\frac{|\fromR|}{|D|}{\displaystyle\binom{|\posR|}{2}^{-1}}\\
&+ \frac{|\fromL|}{|D|}\left[\displaystyle\binom{|\posL|}{2} - |\mathsf{k}^L|\right]{\displaystyle\binom{|\posL|}{2}^{-1}} 
+ \frac{|\fromR|}{|D|}\left[\displaystyle\binom{|\posR|}{2} - |\mathsf{k}^R|\right]{\displaystyle\binom{|\posR|}{2}^{-1}}\\
=& \frac{|\fromL|}{|D|}{\displaystyle\binom{|\posL|}{2}^{-1}}|\mathsf{k}^L|
+ \frac{|\fromR|}{|D|}{\displaystyle\binom{|\posR|}{2}^{-1}}|\mathsf{k}^R|\\
&+ \frac{|\fromL|}{|D|}\left[\displaystyle\binom{|\posL|}{2} - |\mathsf{k}^L|\right]{\displaystyle\binom{|\posL|}{2}^{-1}} 
+ \frac{|\fromR|}{|D|}\left[\displaystyle\binom{|\posR|}{2} - |\mathsf{k}^R|\right]{\displaystyle\binom{|\posR|}{2}^{-1}}\\
=& \frac{|\fromL|}{|D|} + \frac{|\fromR|}{|D|} = 1
\end{align*}
\end{small}

\subsection{A (Theoretically) Faster Sampler for \nameA}

\Cref{alg:nudhy-a-details} may get stuck in the same state $G$ for multiple iterations, due to the fact that it selects pairs of vertices that cannot form any PSO (line~\ref{line:sl1} and line~\ref{line:sl2}).
Consequently, achieving convergence may necessitate a large number of iterations.

In an effort to improve mixing time, we explored an alternative implementation of the MH approach for sampling from $\nullmodelA$. This variant, named \algoB, randomly samples pairs of edges at each iteration until it identifies a pair forming a PSO. It then accepts the transition to the new state $G'$ resulting from the applied PSO with a probability of $\min\left\{1, \frac{\neighdistr{G'}(G)}{\neighdistr{G}(G')}\right\}$.
In this case, $\neighdistr{G}(G') = \degG^{-1}$, where $\degG$ is the number of PSOs applicable to $G$.
Moreover, as different numbers of PSOs can be applied to different states $G \in \states^{\nameA}$, the equality $\neighdistr{G}(G') = \neighdistr{G'}(G)$ no longer holds.
As a result, the acceptance probability of \algoB is $\min\left\{1, \frac{\degG}{\mathsf{d}(G')}\right\}$.

\begin{algorithm}[!ht]
  \small
    \caption{\algoB}\label{alg:nudhy-b}
    \begin{algorithmic}[1]
    \Require Graph $G \doteq (L, R, D) \in \states^{\nameA}$, Number of Steps $s$
    \Ensure Graph sampled uniformly from $\states^{\nameA}$ 
    \State $G' \gets G$
    \RepeatN{$s$}
      \While{$(u, \alpha, d_1), (v, \beta, d_2)$ do not form a PSO}
        \State $(u, \alpha, d_1)$, $(v, \beta, d_2)$ $\gets$ edges drawn u.a.r. from $D$
      \EndWhile
      \State $G' \gets $apply $(u,\alpha,d_1), (v,\beta, d_2) \pso (u, \beta, d_1), (v, \alpha, d_2)$ to $G$
      \State $p \gets $ random real number in $[0,1]$
      \If{$p \leq \min\left(1, \degG/\mathsf{d}(G')\right)$}
        $G \gets G'$
      \EndIf
    \End
    \State \Return $G$
    \end{algorithmic}
\end{algorithm}

\Cref{alg:nudhy-b} illustrates the whole procedure.
Due to the absence of self-loops, \algoB does not get stuck in problematic states like \algoA.
As a result, it \emph{should} theoretically converge at a faster rate.
Nevertheless, our experimental findings contradict this expectation. We demonstrated that the computational overhead incurred by calculating the number of potential PSOs at each state counterbalances this anticipated advantage. This holds true even with the optimizations detailed in the next section.
Scalability experiments on the datasets in \Cref{tbl:datasets} in \Cref{ax:data} showed that \algoB is found to be at least one order of magnitude slower than \algoA.

\subsubsection{An Efficient Way to Compute the Degree of the Current State}

In this section we show how \algoB efficiently computes $\degG$ and keeps its value up-to-date as it moves in $\states^{\nameA})$.
Informally, $\degG$ is given by the number of edge pairs that \textbf{(i)} have the same direction, \textbf{(ii)} do not share any endpoint, and \textbf{(iii)} are not endpoints of a \emph{caterpillar}.
A caterpillar is a configuration of $3$ directed edges with the same direction involving four vertices.
Two edges involved in a caterpillar cannot form a PSO, because at least one of the edges that would be created by the PSO already exists in the graph.
Thus, $\degG$ can be computed as follows~\cite{gionis2007assessing}:

\begin{equation}\label{eq:degree}
\degG = J(G) - \typeB{G} + 2\butterfly{G}
\end{equation}

where
\begin{align*}
J(G) =& \frac{1}{2}\left(|\fromL| (|\fromL| + 1) - \sum_{v \in L}{\sqoutdeg{v}} - \sum_{\alpha \in R}{\sqindeg{\alpha}} \right)\\
+& \frac{1}{2}\left(|\fromR| (|\fromR| + 1) - \sum_{v \in L}{\sqindeg{v}} - \sum_{\alpha \in R}{\sqoutdeg{\alpha}} \right)
\end{align*}
is the number of disjoint pairs of edges with the same direction, 
\begin{align*}
\typeB{G} = \sum_{j=1}^{\maxo_L}\sum_{k=1}^{\maxi_R} (j-1)(k-1)\left(\sum_{i=1}^{\maxi_L}\sum_{l=1}^{\maxo_R}\sum_{d \in \{1,-1\}} \biot{G}[i,j,k,l,d] \right)
\end{align*}
is the number of caterpillars, and
\begin{align*}
\butterfly{G} = \sum_{\substack{u,v \in L\\u \neq v}}\binom{|\outneighs{v} \cap \outneighs{u}|}{2} + \sum_{\substack{\alpha,\beta \in R\\\alpha \neq \beta}}\binom{|\outneighs{\alpha} \cap \outneighs{\beta}|}{2}
\end{align*}
is the number of complete $k(2,2)$ graphs in $G$, a.k.a. \emph{butterflies}, formed by edges with the same direction.
Each $k(2,2)$ graph consists in $4$ caterpillars but includes only $2$ pairs of non-swappable edges. Thus, we must add $2\butterfly{G}$ to $\typeB{G}$, to obtain the correct count.

The change in $\typeB{G}$ due to the PSO $(u,\alpha,+1),$ $(v,\beta,+1)$ $\pso$ $(u,\beta,+1),$ $(v,\alpha,+1)$ is
\begin{align*}
\Delta(\overrightarrow{\typeB{G}}) = \left(|\outneighs{u}|-|\outneighs{v}|\right)\left(|\inneighs{\beta}|-|\inneighs{\alpha}|\right)\,,
\end{align*}
while the change in $\butterfly{G}$ is
\begin{align*}
\Delta(\overrightarrow{\butterfly{G}}) &=
\sum_{\substack{w \in \inneighs{\beta} \setminus \inneighs{\alpha}\\w \neq v}}|\outneighs{u} \cap \outneighs{w}| - \left(|\outneighs{v} \cap \outneighs{w}|-1\right)\\
&+ \sum_{\substack{w \in \inneighs{\alpha} \setminus \inneighs{\beta}\\w \neq u}}|\outneighs{v} \cap \outneighs{w}| - \left(|\outneighs{u} \cap \outneighs{w}|-1\right)\,.
\end{align*}

Similarly, the change in $\typeB{G}$ after the application of $(u,\alpha,-1),$ $(v,\beta,-1)$ $\pso$ $(u,\beta,-1),$ $(v,\alpha,-1)$ is
\begin{align*}
\Delta(\overleftarrow{\typeB{G}}) = \left(|\outneighs{\alpha}|-|\outneighs{\beta}|\right)\left(|\inneighs{v}|-|\inneighs{u}|\right)\,,
\end{align*}
while the change in $\butterfly{G}$ is
\begin{align*}
\Delta(\overleftarrow{\butterfly{G}}) &=
\sum_{\substack{w \in \outneighs{\beta} \setminus \outneighs{\alpha} \\w \neq v}}|\inneighs{u} \cap \inneighs{w}| - \left(|\inneighs{v} \cap \inneighs{w}|-1\right)\\
&+ \sum_{\substack{w \in \outneighs{\alpha} \setminus \outneighs{\beta}\\w \neq u}}|\inneighs{v} \cap \inneighs{w}| - \left(|\inneighs{u} \cap \inneighs{w}|-1\right)\,.
\end{align*}

\section{Algorithmic Details of \algoC}\label{ax:algoc}

Let first define some useful quantities:
\begin{squishlist}
  \item $\forall 0 \le i \le \maxi_L\,,1 \le j \le \maxo_L\,,\quad L^+_{i,j} \doteq \{v \in L \suchthat \outdegG{v} = j \, \wedge \, \indegG{v} = i\}\,;$
  \item $\forall 0 \le i \le \maxi_R\,,1 \le j \le \maxo_R\,,\quad R^+_{i,j} \doteq \{\alpha \in R \suchthat \outdegG{\alpha} = j \, \wedge \, \indegG{\alpha} = i\}\,;$
  \item $\forall 1 \le i \le \maxi_L\,,0 \le j \le \maxo_L\,,\quad L^-_{i,j} \doteq \{v \in L \suchthat \outdegG{v} = j \, \wedge \, \indegG{v} = i\}\,;$
  \item $\forall 1 \le i \le \maxi_R\,,0 \le j \le \maxo_R\,,\quad R^-_{i,j} \doteq \{\alpha \in R \suchthat \outdegG{\alpha} = j \, \wedge \, \indegG{\alpha} = i\}\,;$
  \item $\forall u,v \in L \cup R\,,\quad \setdiff{u}{v} \doteq \{(\alpha,\beta) \suchthat \alpha \in \outneighs{u} \setminus \outneighs{v} \, \wedge \beta \in \outneighs{v} \setminus \outneighs{u}\}\,;$
  \item $\forall u,v \in L \cup R\,,\quad \setidiff{u}{v} \doteq \{(\alpha, \beta) \suchthat \alpha \in \inneighs{u} \setminus \inneighs{v} \, \wedge \, \beta \in \inneighs{v} \setminus \inneighs{u}\}\,.$
\end{squishlist}

\begin{algorithm}[!ht]
  \small
    \caption{\algoC}\label{alg:nudhy-c-details}
    \begin{algorithmic}[1]
    \Require Graph $G \doteq (L, R, D) \in \states^{\nameC}$, Number of Steps $s$
    \Ensure Graph sampled uniformly from $\states^{\nameC}$ 
    \RepeatN{$s$}
      \State $\mathsf{out} \gets$ flip a biased coin with heads prob $|\fromL|/|D|$
      \State $\mathsf{out2} \gets$ flip a fair coin
      \If{$\mathsf{out}$ is \emph{heads} \textbf{and} $\mathsf{out2}$ is \emph{heads}}
        \State $i,j \gets$ ints drawn with prob $\vartheta(i,j)$ from $[0,\maxi_L]$ and $[1,\maxo_L]$ 
        \State $u$, $v$ $\gets$ different vertices drawn u.a.r. from $L^+_{ij}$
        \If{$\setdiff{u}{v} = \varnothing$}
          \textbf{continue} \Comment{\textit{self-loop}}
        \Else
          \State $\alpha$, $\beta$ $\gets$ pair drawn u.a.r. from $\setdiff{u}{v}$
          \State $d \gets +1$
        \EndIf
      \ElsIf{$\mathsf{out}$ is \emph{tails} \textbf{and} $\mathsf{out2}$ is \emph{heads}}
        \State $i,j \gets$ ints drawn with prob $\eta(i,j)$ from $[1,\maxi_L]$ and $[0,\maxo_L]$ 
        \State $u$, $v$ $\gets$ different vertices drawn u.a.r. from $L^-_{ij}$
        \If{$\setidiff{u}{v} = \varnothing$}
          \textbf{continue} \Comment{\textit{self-loop}}
        \Else
          \State $\alpha$, $\beta$ $\gets$ pair drawn u.a.r. from $\setidiff{u}{v}$
          \State $d \gets -1$
        \EndIf
      \ElsIf{$\mathsf{out}$ is \emph{heads} \textbf{and} $\mathsf{out2}$ is \emph{tails}}
        \State $i,j \gets$ ints drawn with prob $\phi(i,j)$ from $[1,\maxi_R]$ and $[0,\maxo_R]$ 
        \State $\alpha$, $\beta$ $\gets$ different vertices drawn u.a.r. from $R^-_{ij}$
        \If{$\setidiff{\alpha}{\beta} = \varnothing$}
          \textbf{continue} \Comment{\textit{self-loop}}
        \Else
          \State $u$, $v$ $\gets$ pair drawn u.a.r. from $\setidiff{\alpha}{\beta}$
          \State $d \gets +1$
        \EndIf
      \Else
        \State $i,j \gets$ ints drawn with prob $\nu(i,j)$ from $[0,\maxi_R]$ and $[1,\maxo_R]$ 
        \State $\alpha$, $\beta$ $\gets$ different vertices drawn u.a.r. from $R^+_{ij}$
        \If{$\setdiff{\alpha}{\beta} = \varnothing$}
          \textbf{continue} \Comment{\textit{self-loop}}
        \Else
          \State $u$, $v$ $\gets$ pair drawn u.a.r. from $\setdiff{\alpha}{\beta}$
          \State $d \gets -1$
        \EndIf
      \EndIf
      \State $G \gets $perform $(u,\alpha,d), (v,\beta, d) \rpso (u, \beta, d), (v, \alpha, d)$ on $G$ \Comment{\textit{transition always accepted}}
    \End
    \State \Return $G$
    \end{algorithmic}
\end{algorithm}

To sample a neighbor $G'$ of $G$, we first flip two coins.
The first one is an fair coin, while the second one is a biased coin that outputs \emph{heads} with probability $|\fromL|/|D|$.
Let denote with $\langle\mathsf{fair}, \mathsf{biased}\rangle$ the tuple consisting in the outcomes of the two coins.
We distinguish four cases.

\mpara{Case $\langle\mathsf{heads}, \mathsf{heads}\rangle$.}
We draw a pair of integers $0 \le i \le \maxi_L$ and $1 \le j \le \maxo_L$ with probability
\begin{equation}\label{eq:probability-hh}
  \vartheta(i,j) \doteq {\displaystyle \binom{|L^+_{i,j}|}{2}}\left({\sum_{k=0}^{\maxi_L}\sum_{l=1}^{\maxo_L}\displaystyle\binom{|L^+_{k,l}|}{2}}\right)^{-1},
\end{equation}
and then draw a pair $(u,v)$ of distinct vertices from $L^+_{i,j}$ uniformly at random.
If $|\setdiff{u}{v}| = 0$, we set $G' = G$. 
Otherwise, we draw $(\alpha,\beta)$ from $\setdiff{u}{v}$ uniformly at random. 
By construction, $(u,\alpha,+1),$ $(v,\beta,+1)$ $\rpso$ $(u,\beta,+1),$ $(v,\alpha,+1)$ is a RPSO, and we can set $G'$ to be the graph obtained by performing it on $G$.

\mpara{Case $\langle\mathsf{tails}, \mathsf{heads}\rangle$.}
We draw a pair of integers $1 \le i \le \maxi_L$ and $0 \le j \le \maxo_L$ with probability
\begin{equation}\label{eq:probability-th}
  \eta(i,j) \doteq \binom{|L^-_{i,j}|}{2}\left(\sum_{k=1}^{\maxi_L}\sum_{l=0}^{\maxo_L}\binom{|L^-_{k,l}|}{2}\right)^{-1},
\end{equation}
and then draw a pair $(u,v)$ of distinct elements from $L^-_{i,j}$ uniformly at random.
If $|\setidiff{u}{v}| = 0$, we set $G' = G$. 
Otherwise, we draw $(\alpha, \beta)$ from $\setidiff{u}{v}$ uniformly at random. 
By construction, $(u,\alpha,-1),$ $(v,\beta,-1)$ $\rpso$ $(u,\beta,-1),$ $(v,\alpha,-1)$ is a RPSO, and we can set $G'$ to be the graph obtained by performing it on $G$.

\mpara{Case $\langle\mathsf{heads}, \mathsf{tails}\rangle$.}
We draw a pair of integers $1 \le i \le \maxi_R$ and $0 \le j \le \maxo_R$ with probability
\begin{equation}\label{eq:probability-ht}
  \phi(i,j) \doteq \binom{|R^-_{i,j}|}{2}\left(\sum_{k=1}^{\maxi_R}\sum_{l=0}^{\maxo_R}\binom{|R^-_{k,l}|}{2}\right)^{-1},
\end{equation}
and then we draw a pair $(\alpha,\beta)$ of distinct elements from $R^-_{i,j}$ uniformly at random. 
If $|\setidiff{\alpha}{\beta}| = 0$, we set $G' = G$. 
Otherwise, we draw $(u,v)$ from $\setidiff{\alpha}{\beta}$ uniformly at random. 
By construction, $(u,\alpha,+1),(v,\beta,+1) \rpso (u,\beta,+1),(v,\alpha,+1)$ is a RPSO, and we can set $G'$ to be the graph obtained by performing it on $G$.

\mpara{Case $\langle\mathsf{tails}, \mathsf{tails}\rangle$.}
We draw a pair of integers $0 \le i \le \maxi_R$ and $1 \le j \le \maxo_R$ with probability
\begin{equation}\label{eq:probability-tt}
  \nu(i,j) \doteq \binom{|R^+_{i,j}|}{2}\left(\sum_{k=0}^{\maxi_R}\sum_{l=1}^{\maxo_R}\binom{|R^+_{k,l}|}{2}\right)^{-1},
\end{equation}
and then we draw a pair $(\alpha,\beta)$ of distinct elements from $R^+_{i,j}$ uniformly at random. 
If $|\setdiff{\alpha}{\beta}|= 0$, we set $G' = G$. 
Otherwise, we draw $(u,v)$ from $\setdiff{\alpha}{\beta}$ uniformly at random. 
By construction, $(u,\alpha,-1),(v,\beta,-1) \rpso (u,\beta,-1),(v,\alpha,-1)$ 
 is a RPSO, and we can set $G'$ to be the graph obtained by performing it on $G$.
This procedure induces a probability distribution $\neighdistr{G}$ over the set of neighbors of $G$. 
Each directed edge $(G, G')$ in $\states^{\nameC}$ has thus weight $\neighdistr{G}(G')$.
Let $(u,\alpha,d), (v,\beta,d) \rpso (u,\beta,d), (v,\alpha,d)$ be the sampled RPSO and $G'$ be the graph obtained by performing such RPSO on $G$.
We define the following two events:
\begin{align*}
  \mathsf{E}_L \doteq& \text{``vertices $u$ and $v$ have the same in- and out-degree'';}\\
  \mathsf{E}_R \doteq& \text{``vertices $\alpha$ and $\beta$ have the same in- and out-degree''.}
\end{align*}

Then, $\neighdistr{G}(G')$ takes one of the following six values:
\begin{itemize}
  \item if $d = +1$ and only $\mathsf{E}_L$ holds, then
    \begin{equation}\label{eq:neighdistrleft1}
      \neighdistr{G}(G') = 
      \frac{|\fromL|}{|D|}\frac{1}{2}\left({\displaystyle \sum_{k=0}^{\maxi_L}\sum_{l=1}^{\maxo_L}\binom{|L^+_{k,l}|}{2}}|\setdiff{u}{v}|\right)^{-1};
    \end{equation}
    \item if $d = +1$ and only $\mathsf{E}_R$ holds, then
    \begin{equation}\label{eq:neighdistrright1}
      \neighdistr{G}(G') = 
      \frac{|\fromL|}{|D|}\frac{1}{2}\left({\displaystyle \sum_{k=1}^{\maxi_R}\sum_{l=0}^{\maxo_R}\binom{|R^-_{k,l}|}{2}}|\setidiff{\alpha}{\beta}|\right)^{-1};
    \end{equation}
    \item if $d = +1$ and both $\mathsf{E}_L$ and $\mathsf{E}_R$ hold, then $\neighdistr{G}(G')$ is the sum of \cref{eq:neighdistrleft1} and \cref{eq:neighdistrright1};
    \item if $d = -1$ and only $\mathsf{E}_L$ holds, then
    \begin{equation}\label{eq:neighdistrleft-1}
      \neighdistr{G}(G') = 
      \frac{|\fromR|}{|D|}\frac{1}{2}\left({\displaystyle \sum_{k=1}^{\maxi_L}\sum_{l=0}^{\maxo_L}\binom{|L^-_{k,l}|}{2}}|\setidiff{u}{v}|\right)^{-1};
    \end{equation}
    \item if $d = -1$ and only $\mathsf{E}_R$ holds, then
    \begin{equation}\label{eq:neighdistrright-1}
      \neighdistr{G}(G') = 
      \frac{|\fromR|}{|D|}\frac{1}{2}\left({\displaystyle \sum_{k=0}^{\maxi_R}\sum_{l=1}^{\maxo_R}\binom{|R^+_{k,l}|}{2}}|\setdiff{\alpha}{\beta}|\right)^{-1};
    \end{equation}
    \item if $d=-1$ and both $\mathsf{E}_L$ and $\mathsf{E}_R$ hold, then $\neighdistr{G}(G')$ is the sum of \cref{eq:neighdistrleft-1} and \cref{eq:neighdistrright-1}.
\end{itemize}

Also in this case we have that $\neighdistr{G}(G') = \neighdistr{G'}(G)$, and thus we just need to prove that $\sum_{G' \in \states^{\nameC}}\neighdistr{G}(G') = 1$, to get that the acceptance probability of $\algoC$ is $1$.

In the following, let $\tilde{L}^+_{ij}$ be the set of pairs of vertices $u,v \in L^+_{ij}$ such that $\setdiff{u}{v} \neq \varnothing$, $\neg \tilde{L}^+_{ij} = L^+_{ij} \setminus \tilde{L}^+_{ij}$, and $\tilde{\mathsf{X}}^+_{uv}$ be the set of pairs of vertices $(\alpha, \beta) \in \setdiff{u}{v}$ such that $\indegG{\alpha}=\indegG{\beta}$ and $\outdegG{\alpha} = \outdegG{\beta}$.
We define $\tilde{L}^-_{ij}$, $\tilde{R}^+_{ij}$, $\tilde{R}^-_{ij}$, $\neg \tilde{L}^-_{ij}$, $\neg \tilde{R}^+_{ij}$, $\neg \tilde{R}^-_{ij}$, and $\tilde{\mathsf{X}}^-_{uv}$ in the same way.

Let consider the subset $\states^{\nameC}_+$ of graphs $G' \in \states^{\nameC}$ that can be obtained from $G$ applying a RPSO of type $(u, \alpha, +1), (v, \beta, +1)$ $\rpso$ $(u,\beta, +1), (v, \alpha, +1)$. 
Since two states in $\states^{\nameC}$ can be connected by at most one RPSO, it holds that

\begin{scriptsize}
\begin{align*}
\sum_{G' \in \states^{\nameC}_+}\neighdistr{G}(G') =&
\sum_{i=0}^{\maxi_L} \sum_{j=1}^{\maxo_L} \sum_{(u,v) \in \tilde{L}^+_{ij}}\frac{|\fromL|}{|D|}\frac{1}{2}\left[\sum_{k=0}^{\maxi_L} \sum_{l=1}^{\maxo_L} {\displaystyle\binom{|L^+_{kl}|}{2}^{-1}}\frac{|\tilde{\mathsf{X}}_{uv}^+|}{|\setdiff{u}{v}|}
+\left(\sum_{k=1}^{\maxi_R} \sum_{l=0}^{\maxo_R} \displaystyle\binom{|R^-_{kl}|}{2}\sum_{(\alpha,\beta) \in \tilde{\mathsf{X}}_{uv}^+}|\setidiff{\alpha}{\beta}|\right)^{-1}\right]\\
&+\sum_{i=0}^{\maxi_L} \sum_{j=1}^{\maxo_L} \sum_{(u,v) \in \tilde{L}^+_{ij}}\frac{|\fromL|}{|D|}\frac{1}{2}\sum_{k=0}^{\maxi_L} \sum_{l=1}^{\maxo_L} \binom{|L^+_{kl}|}{2}^{-1}\frac{|\setdiff{u}{v}|-|\tilde{\mathsf{X}}_{uv}^+|}{|\setdiff{u}{v}|}\\
&+ \sum_{i=1}^{\maxi_R} \sum_{j=0}^{\maxo_R} \sum_{(\alpha,\beta) \in \tilde{R}^-_{ij}}\frac{|\fromL|}{|D|}\frac{1}{2}\sum_{k=1}^{\maxi_R} \sum_{l=0}^{\maxo_R} \binom{|R^-_{kl}|}{2}^{-1}\frac{|\setidiff{\alpha}{\beta}| - |\tilde{\mathsf{X}}^-_{\alpha\beta}|}{|\setidiff{\alpha}{\beta}|}\\
=& \frac{|\fromL|}{|D|}\frac{1}{2}\sum_{i=0}^{\maxi_L} \sum_{j=1}^{\maxo_L} |\tilde{L}^+_{ij}|\left(\sum_{k=0}^{\maxi_L} \sum_{l=1}^{\maxo_L} \binom{|L^+_{kl}|}{2}\right)^{-1}
+ \sum_{i=1}^{\maxi_R} \sum_{j=0}^{\maxo_R} \sum_{(\alpha,\beta) \in \tilde{R}^-_{ij}}\frac{|\fromL|}{|D|}\frac{1}{2}\sum_{k=1}^{\maxi_R} \sum_{l=0}^{\maxo_R} \binom{|R^-_{kl}|}{2}^{-1}\frac{|\tilde{\mathsf{X}}^-_{\alpha\beta}|}{|\setidiff{\alpha}{\beta}|}\\
&+ \sum_{i=1}^{\maxi_R} \sum_{j=0}^{\maxo_R} \sum_{(\alpha,\beta) \in \tilde{R}^-_{ij}}\frac{|\fromL|}{|D|}\frac{1}{2}\sum_{k=1}^{\maxi_R} \sum_{l=0}^{\maxo_R} \binom{|R^-_{kl}|}{2}^{-1}\frac{|\setidiff{\alpha}{\beta}| - |\tilde{\mathsf{X}}^-_{\alpha\beta}|}{|\setidiff{\alpha}{\beta}|}\\
=&\frac{|\fromL|}{|D|}\frac{1}{2}\sum_{i=0}^{\maxi_L} \sum_{j=1}^{\maxo_L} |\tilde{L}^+_{ij}|\left(\sum_{k=0}^{\maxi_L} \sum_{l=1}^{\maxo_L} \binom{|L^+_{kl}|}{2}\right)^{-1}
+\frac{|\fromL|}{|D|}\frac{1}{2}\sum_{i=1}^{\maxi_R} \sum_{j=0}^{\maxo_R} |\tilde{R}^-_{ij}|\left(\sum_{k=1}^{\maxi_R} \sum_{l=0}^{\maxo_R} \binom{|R^-_{kl}|}{2}\right)^{-1}.
\end{align*}
\end{scriptsize}
The first term is the sum of the transition probabilities for the graphs reachable by RPSO where both source and destination vertices have the same in- and out-degree;
the second term refers to the cases where only the source vertices have the same in- and out-degree;
and the third term refers to the cases where only the destination vertices have the same in- and out-degree.
The first equality is obtained by observing that each $(\alpha,\beta)$ in the first term appears exactly $\tilde{\mathsf{X}}^-_{\alpha\beta}$ times in the summation, because it is considered for each pair of $(u,v) \in \setidiff{\alpha}{\beta}$ with the same degree.  

Similarly, the sum of the transition probabilities to the subset $\states^{\nameC}_-$ of graphs $G' \in \states^{\nameC}$ that can be obtained from $G$ applying a RPSO of type $(u, \alpha, -1),$ $(v, \beta, -1)$ $\rpso$ $(u,\beta, -1),$ $(v, \alpha, -1)$ is equal to
\begin{scriptsize}
\begin{align*}
\sum_{G' \in \states^{\nameC}_-}\neighdistr{G}(G') &=
\frac{|\fromR|}{|D|}\frac{1}{2}\sum_{i=1}^{\maxi_L} \sum_{j=0}^{\maxo_L} |\tilde{L}^-_{ij}|\left(\sum_{k=1}^{\maxi_L} \sum_{l=0}^{\maxo_L} \binom{|L^-_{kl}|}{2}\right)^{-1}
+ \frac{|\fromR|}{|D|}\frac{1}{2}\sum_{i=0}^{\maxi_R} \sum_{j=1}^{\maxo_R} |\tilde{R}^+_{ij}|\left(\sum_{k=0}^{\maxi_R} \sum_{l=1}^{\maxo_R} \binom{|R^+_{kl}|}{2}\right)^{-1}.
\end{align*}
\end{scriptsize}

Every time the pair of vertices drawn $(u,v)$ (resp. $(\alpha, \beta)$) belongs to $\neg \tilde{L}^+_{ij}$ in the case $\langle\mathsf{heads}, \mathsf{heads}\rangle$, or to $\neg \tilde{L}^-_{ij}$ in the case $\langle\mathsf{tails}, \mathsf{heads}\rangle$ (resp. $\neg \tilde{R}^+_{ij}$ in the case $\langle\mathsf{tails}, \mathsf{tails}\rangle$, or $\neg \tilde{R}^-_{ij}$ in the case $\langle\mathsf{heads}, \mathsf{tails}\rangle$), we transition to the graph itself, i.e. we perform a \emph{self-loop}. 
This gives the following self-loop probability:
\begin{scriptsize}
\begin{align*}
\neighdistr{G}(G)
=&\sum_{i=0}^{\maxi_L} \sum_{j=1}^{\maxo_L} \sum_{(u,v) \in \neg \tilde{L}^+_{ij}}\frac{|\fromL|}{|D|}\frac{1}{2}\sum_{k=0}^{\maxi_L} \sum_{l=1}^{\maxo_L} {\displaystyle \binom{|L^+_{kl}|}{2}^{-1}}
+ \sum_{k=1}^{\maxi_R} \sum_{l=0}^{\maxo_R} \sum_{(\alpha,\beta) \in \neg \tilde{R}^-_{ij}}\frac{|\fromL|}{|D|}\frac{1}{2}\sum_{k=1}^{\maxi_R} \sum_{l=0}^{\maxo_R} \binom{|R^-_{kl}|}{2}^{-1}\\
&+ \sum_{i=1}^{\maxi_L} \sum_{j=0}^{\maxo_L} \sum_{(u,v) \in \neg \tilde{L}^-_{ij}}\frac{|\fromR|}{|D|}\frac{1}{2}\sum_{k=1}^{\maxi_L} \sum_{l=0}^{\maxo_L} \binom{|L^-_{kl}|}{2}^{-1}
+ \sum_{k=0}^{\maxi_R} \sum_{l=1}^{\maxo_R} \sum_{(\alpha,\beta) \in \neg \tilde{R}^+_{ij}}\frac{|\fromR|}{|D|}\frac{1}{2}\sum_{k=0}^{\maxi_R} \sum_{l=1}^{\maxo_R} \binom{|R^+_{kl}|}{2}^{-1}\\
=& \sum_{i=0}^{\maxi_L} \sum_{j=1}^{\maxo_L} |\neg \tilde{L}^+_{ij}|\frac{|\fromL|}{|D|}\frac{1}{2}\sum_{k=0}^{\maxi_L} \sum_{l=1}^{\maxo_L} \binom{|L^+_{kl}|}{2}^{-1}
+ \sum_{k=1}^{\maxi_R} \sum_{l=0}^{\maxo_R} |\neg \tilde{R}^-_{ij}|\frac{|\fromL|}{|D|}\frac{1}{2}\sum_{k=1}^{\maxi_R} \sum_{l=0}^{\maxo_R} \binom{|R^-_{kl}|}{2}^{-1}\\
&+ \sum_{i=1}^{\maxi_L} \sum_{j=0}^{\maxo_L} |\neg \tilde{L}^-_{ij}|\frac{|\fromR|}{|D|}\frac{1}{2}\sum_{k=1}^{\maxi_L} \sum_{l=0}^{\maxo_L} \binom{|L^-_{kl}|}{2}^{-1}
+ \sum_{k=0}^{\maxi_R} \sum_{l=1}^{\maxo_R} |\neg \tilde{R}^+_{ij}|\frac{|\fromR|}{|D|}\frac{1}{2}\sum_{k=0}^{\maxi_R} \sum_{l=1}^{\maxo_R} \binom{|R^+_{kl}|}{2}^{-1}\,.
\end{align*}
\end{scriptsize}

By observing that each of the $\binom{|L^+_{ij}|}{2}$ pairs of vertices in $L^+_{ij}$ is either in $\tilde{L}^+_{ij}$ or in $\neg \tilde{L}^+_{ij}$, 
the sum of all the previous terms gives our result:
\begin{small}
\begin{align*}
\sum_{G' \in \states^{\nameC}}\neighdistr{G}(G') &= \sum_{G' \in \states^{\nameC}_+}\neighdistr{G}(G') + \sum_{G' \in \states^{\nameC}_-}\neighdistr{G}(G') + \neighdistr{G}(G)
= \frac{|\fromL|}{|D|}+ \frac{|\fromR|}{|D|} = 1\,.
\end{align*}
\end{small}

\subsection{A (Theoretically) Faster Sampler for \nameC}

We studied an edge-sampling-based MH algorithm also for sampling from $\nullmodelC$.
This algorithm, dubbed {\algoD}, samples pairs of edges uniformly at random until it finds a pair of edges forming a RPSO.
Its transition probability from a state $G$ to a state $G'$ is $\neighdistr{G}(G') = \degGR^{-1}$, where $\degGR$ is the number of RPSOs applicable to $G$.
The acceptance probability of {\algoD} is thus 
$\min\left\{1, \frac{\degGR}{\mathsf{d}^\mathrm{R}(G')}\right\}$.
\Cref{alg:nudhy-d} illustrates the procedure.

\begin{algorithm}[ht!]
	\small
    \caption{\algoD}\label{alg:nudhy-d}
    \begin{algorithmic}[1]
    \Require Graph $G \doteq (L, R, D) \in \states^{\nameC}$, Number of Steps $s$
    \Ensure Graph sampled uniformly from $\states^{\nameC}$ 
    \State $G' \gets G$
    \RepeatN{$s$}
    	\While{$(u, \alpha, d_1), (v, \beta, d_2)$ do not form a RPSO}
      	\State $(u, \alpha, d_1)$, $(v, \beta, d_2)$ $\gets$ edges drawn u.a.r. from $D$
      \EndWhile
        \State $G' \gets $apply $(u,\alpha,d_1), (v,\beta, d_2) \pso (u, \beta, d_1), (v, \alpha, d_2)$ to $G$
        \State $p \gets $ random real number in $[0,1]$
        \If{$p \leq \min\left(1, \degGR/\mathsf{d}^\mathrm{R}(G')\right)$}
          $G \gets G'$
        \EndIf
    \End
    \State \Return $G$
    \end{algorithmic}
\end{algorithm}

In the following, we discuss how \algoD computes $\degGR$ and keeps its value up-to-date as it moves in the Markov graph.
We recall that a pair of edges forms a RPSO if they \textbf{(i)} have the same direction, \textbf{(ii)} do not share any endpoint, \textbf{(iii)} are not endpoints of a caterpillar, and \textbf{(iv)} either the sources or the destinations of the two edges have the same in- and out-degree.
Thus, $\degGR$ can be computed as follows:
\begin{small}
\begin{equation}\label{eq:degree2}
\degGR = \overrightarrow{J}^\mathrm{R}(G) + \overleftarrow{J}^\mathrm{R}(G) - \overrightarrow{C}^\mathrm{R}(G) - \overleftarrow{C}^\mathrm{R}(G) + 2\overrightarrow{\butterfly{G}}^\mathrm{R} + 2\overleftarrow{\butterfly{G}}^\mathrm{R}\,.
\end{equation}
\end{small}
In the following, we introduce each term in \Cref{eq:degree2} and show how to calculate it efficiently.
Let $\mathsf{P}_L \doteq \left[\maxi_L\right] \times \left[\maxo_L\right]$ and
$\mathsf{P}_R \doteq \left[\maxi_R\right] \times \left[\maxo_R\right]$. 
For $i \in [0, \maxi_L]$, $j \in [0, \maxo_L]$, $l \in [0, \maxi_R]$, and $k \in [0, \maxo_R]$, we define the following quantities:
\begin{squishlist} 
\item $\overrightarrow{\biotG}^{ij} \doteq \sum\limits_{(l,k) \in \mathsf{P}_R}\biotG[i,j,l,k,+1]$ and $\overrightarrow{\biotG}^{lk} \doteq \sum\limits_{(i,j) \in \mathsf{P}_L}\biotG[i,j,l,k,+1]$;
\item $\overleftarrow{\biotG}^{ij} \doteq \sum\limits_{(l,k) \in \mathsf{P}_R}\biotG[i,j,l,k,-1]$ and $\overleftarrow{\biotG}^{lk} \doteq \sum\limits_{(i,j) \in \mathsf{P}_L}\biotG[i,j,l,k,-1]$;
\item $L_{ij} \doteq \{v \in L \suchthat \outdegG{v} = j \, \wedge \, \indegG{v}=i\}$;
\item $R_{ij} \doteq \{\alpha \in R \suchthat \outdegG{\alpha} = j\, \wedge \, \indegG{\alpha} = i\}$.
\end{squishlist}

The term $\overrightarrow{J}^\mathrm{R}(G)$ counts the disjoint pairs of candidate swappable edges with direction $+1$:
\begin{align}
\label{eq:A}
\overrightarrow{J}^\mathrm{R}(G) =&
\sum_{i=0}^{\maxi_L}\sum_{j=0}^{\maxo_L}\binom{\overrightarrow{\biotG}^{ij} + 1}{2}\\
\label{eq:B}
+& \sum_{l=0}^{\maxi_R}\sum_{k=0}^{\maxo_R}\binom{\overrightarrow{\biotG}^{lk} + 1}{2}\\
\label{eq:C}
-& \sum_{i=0}^{\maxi_L}\sum_{l=0}^{\maxi_R}\sum_{j=0}^{\maxo_L}\sum_{k=0}^{\maxo_R}\binom{\biotG[i,j,l,k,+1]}{2}\\
\label{eq:D}
-& \sum_{v \in L}\binom{\outdegG{v}}{2}\\
\label{eq:E}
-& \sum_{\alpha \in R}\binom{\indegG{\alpha}}{2}\\
-& 2|\fromL|
\end{align}
where \cref{eq:A} is the number of pairs of edges in $\fromL$ from vertices with the same in- and out-degree;
where \cref{eq:B} is the number of pairs of edges in $\fromL$ to vertices with the same in- and out-degree;
where \cref{eq:C} is the number of pairs of distinct edges in $\fromL$ from vertices with the same in- and out-degree to vertices with the same in- and out-degree (we subtract the term because these pairs are counted twice);
where \cref{eq:D} is the number of pairs of edges in $\fromL$ with the same source (we do not want wedges);
and where \cref{eq:A} is the number of pairs of edges in $\fromL$ with the same destination (we do not want wedges).
We also subtract $2|\fromL|$ because we count twice the pairs of edges of the type $(e,e)$.

Similarly, $\overleftarrow{J}^\mathrm{R}(G)$ counts those with direction $-1$:
\begin{align*}
\overleftarrow{J}^\mathrm{R}(G) &= 
\sum_{i=0}^{\maxi_L}\sum_{j=0}^{\maxo_L}\binom{\overleftarrow{\biotG}^{ij} + 1}{2} + 
\sum_{l=0}^{\maxi_R}\sum_{k=0}^{\maxo_R}\binom{\overleftarrow{\biotG}^{lk} + 1}{2}
- \sum_{i=0}^{\maxi_L}\sum_{l=0}^{\maxi_R}\sum_{j=0}^{\maxo_L}\sum_{k=0}^{\maxo_R}\binom{\biotG[i,j,l,k,-1]}{2}\\
&- \sum_{v \in L}\binom{\indegG{v}}{2} - \sum_{\alpha \in R}\binom{\outdegG{\alpha}}{2} - 2|\fromR|\,.
\end{align*}

The term $\overrightarrow{C}^\mathrm{R}(G)$ is the number of configurations of $3$ directed edges with direction $+1$ involving four vertices such that the left and/or the right vertices have the same in- and out-degree:
\begin{align}
\label{eq:F}
\overrightarrow{C}^\mathrm{R}(G) =& \sum_{(u,\alpha,+1)\in \fromL}\left(\outdegG{u}_{\upharpoonright \alpha} - 1\right)\left(\indegG{\alpha} - 1\right)\\
\label{eq:G} 
+& \sum_{(u,\alpha,+1)\in \fromL}\left(\outdegG{u} - 1\right)\left(\indegG{\alpha}_{\upharpoonright u} - 1\right)\\
\label{eq:H}
-& \sum_{(u,\alpha,+1)\in \fromL}\left(\outdegG{u}_{\upharpoonright \alpha} - 1\right)\left(\indegG{\alpha}_{\upharpoonright u} - 1\right)
\end{align}
where \cref{eq:F} is the number of configurations where the right vertices have the same in- and out-degree ($\outdegG{u}_{\upharpoonright \alpha}$ is the number of out-neighbors of $u$ with the same in- and out-degree of $\alpha$);
where \cref{eq:G} is the number of configurations where the left vertices have the same in- and out-degree ($\indegG{\alpha}_{\upharpoonright u}$ is the number of in-neighbors of $\alpha$ with the same in- and out-degree of $u$);
and where \cref{eq:H} is the number of configurations where the left vertices have the same in- and out-degree and the right vertices have the same in- and out-degree (we subtract the term because these configurations are counted twice).

Similarly, $\overleftarrow{C}^\mathrm{R}(G)$ is the number of configurations of $3$ directed edges with direction $-1$ involving four vertices such that the left and/or the right vertices have the same in- and out-degree:
\begin{align*}
\overleftarrow{C}^\mathrm{R}(G) =& \sum_{(u,\alpha,-1)\in \fromR}\left(\indegG{u}_{\upharpoonright \alpha} - 1\right)\left(\outdegG{\alpha} - 1\right)\\
+& \sum_{(u,\alpha,-1)\in \fromR}\left(\indegG{u} - 1\right)\left(\outdegG{\alpha}_{\upharpoonright u} - 1\right)\\
-& \sum_{(u,\alpha,-1)\in \fromR}\left(\indegG{u}_{\upharpoonright \alpha} - 1\right)\left(\outdegG{\alpha}_{\upharpoonright u} - 1\right)\,.
\end{align*}

The term $\overrightarrow{\butterfly{G}}^\mathrm{R}$ is the number of butterflies whose edges have direction $+1$ and such that the left and/or the right vertices have the same in- and out-degree:
\begin{align}
\label{eq:I}
\overrightarrow{\butterfly{G}}^\mathrm{R} =& \sum_{(i,j)\in\mathsf{P}_L}\sum_{\substack{u,v \in L_{ij}\\u \neq v}}\binom{|\outneighs{u} \cap \outneighs{v}|}{2}\\
\label{eq:L}
+& \sum_{(i,j)\in\mathsf{P}_R}\sum_{\substack{\alpha,\beta \in R_{ij}\\\alpha \neq \beta}}\binom{|\inneighs{\alpha} \cap \inneighs{\beta}|}{2}\\
\label{eq:M}
-& \sum_{(i,j)\in\mathsf{P}_L}\sum_{\substack{u,v \in L_{ij}\\u \neq v}}\sum_{(l,k)\in \mathsf{P}_R}\binom{|\outneighs{u}_{\upharpoonright lk} \cap \outneighs{v}_{\upharpoonright lk}|}{2}
\end{align}
where \cref{eq:I} is the number of butterflies where the left vertices have the same in- and out-degree;
where \cref{eq:L} is the number of butterflies where the right vertices have the same in- and out-degree;
and where \cref{eq:M} is the number of butterflies where the left vertices have the same in- and out-degree and the right vertices have the same in- and out-degree ($\outneighs{v}_{\upharpoonright lk}$ is the set of out-neighbors of $v$ with in-degree $l$ and out-degree $k$). We subtract this term because these butterflies are counted twice. 

Similarly, the term $\overleftarrow{\butterfly{G}}^\mathrm{R}$ is the number of butterflies whose edges have direction $-1$ and such that the left and/or the right vertices have the same in- and out-degree:
\begin{align*}
\overleftarrow{\butterfly{G}}^\mathrm{R} =& \sum_{(i,j)\in \mathsf{P}_L}\sum_{\substack{u,v \in L_{ij}\\u \neq v}}\binom{|\inneighs{u} \cap \inneighs{v}|}{2}\\
+& \sum_{(i,j)\in \mathsf{P}_R}\sum_{\substack{\alpha,\beta \in R_{ij}\\\alpha \neq \beta}}\binom{|\outneighs{\alpha} \cap \outneighs{\beta}|}{2}\\
-& \sum_{(i,j)\in \mathsf{P}_L}\sum_{\substack{u,v \in L_{ij}\\u \neq v}}\sum_{(l,k)\in \mathsf{P}_R}\binom{|\inneighs{u}_{\upharpoonright lk} \cap \inneighs{v}_{\upharpoonright lk}|}{2}\,.
\end{align*}

We now show how to obtain the degree of the next state from the degree of the previous state.
The change in $\overrightarrow{C}^\mathrm{R}(G)$ after the application of $(u,\alpha,+1), (v,\beta,+1)$ $\rpso$ $(u,\beta,+1), (v,\alpha,+1)$ is
\begin{align*}
\Delta(\overrightarrow{C}^\mathrm{R}(G)) &= \left(\outdegG{u}_{\upharpoonright \beta} - \outdegG{v}_{\upharpoonright \beta}\right)\left(\indegG{\beta}-\indegG{\beta}_{\upharpoonright v}\right)\\ 
+& \left(\outdegG{v}_{\upharpoonright \alpha} - \outdegG{u}_{\upharpoonright \alpha}\right)\left(\indegG{\alpha}-\indegG{\alpha}_{\upharpoonright v}\right)
\end{align*}
if only $u$ and $v$ have the same in- and out-degree;
\begin{align*}
\Delta(\overrightarrow{C}^\mathrm{R}(G)) &= \left(\outdegG{v}_{\upharpoonright \beta} - \outdegG{v}\right)\left(\indegG{\beta}_{\upharpoonright v}-\indegG{\alpha}_{\upharpoonright v}\right)\\ 
+& \left(\outdegG{u}_{\upharpoonright \beta} - \outdegG{u}\right)\left(\indegG{\alpha}_{\upharpoonright u}-\indegG{\beta}_{\upharpoonright u}\right)
\end{align*}
if only $\alpha$ and $\beta$ have the same in- and out-degree; and
\begin{align*}
\Delta(\overrightarrow{C}^\mathrm{R}(G)) &= \left(\outdegG{v}_{\upharpoonright \beta} - \outdegG{u}_{\upharpoonright \beta}\right)\left(\indegG{\beta}_{\upharpoonright v}-\indegG{\alpha}_{\upharpoonright v}\right)
\end{align*}
if both $u$,$v$ and $\alpha$,$\beta$ have the same in- and out-degree.

The change in $\overleftarrow{C}(G)$ due to the RPSO $(u,\alpha,-1), (v,\beta,-1)$ $\rpso$ $(u,\beta,-1),$ $(v,\alpha,-1)$ is
\begin{align*}
\Delta(\overleftarrow{C}^\mathrm{R}(G)) &= \left(\indegG{u}_{\upharpoonright \beta} - \indegG{v}_{\upharpoonright \beta}\right)\left(\outdegG{\beta}-\outdegG{\beta}_{\upharpoonright v}\right)\\ 
+& \left(\indegG{v}_{\upharpoonright \alpha} - \indegG{u}_{\upharpoonright \alpha}\right)\left(\outdegG{\alpha}-\outdegG{\alpha}_{\upharpoonright v}\right)
\end{align*}
if only $u$ and $v$ have the same in- and out-degree;
\begin{align*}
\Delta(\overleftarrow{C}^\mathrm{R}(G)) &= \left(\indegG{v}_{\upharpoonright \beta} - \indegG{v}\right)\left(\outdegG{\beta}_{\upharpoonright v}-\outdegG{\alpha}_{\upharpoonright v}\right)\\ 
+& \left(\indegG{u}_{\upharpoonright \beta} - \indegG{u}\right)\left(\outdegG{\alpha}_{\upharpoonright u}-\outdegG{\beta}_{\upharpoonright u}\right)
\end{align*}
if only $\alpha$ and $\beta$ have the same in- and out-degree; and
\begin{align*}
\Delta(\overleftarrow{C}^\mathrm{R}(G)) &= \left(\indegG{v}_{\upharpoonright \beta} - \indegG{u}_{\upharpoonright \beta}\right)\left(\outdegG{\beta}_{\upharpoonright v}-\outdegG{\alpha}_{\upharpoonright v}\right)
\end{align*}
if both $u$, $v$ and $\alpha$, $\beta$ have the same in- and out-degree.

Let denote with 
\begin{squishlist}
	\item $\outneighs{u}_{\upharpoonright z}$ the subset of out-neighbors of $u$ with same in- and out-degree of $z$;
	\item $\inneighs{u}_{\upharpoonright z}$ the subset of in-neighbors of $u$ with same in- and out-degree of $z$;
	\item $L_{\upharpoonright u}$ the subset of left vertices with same in- and out-degree of $u$;
	\item $\overrightarrow{\gamma}(u,v)$ the number of common out-neighbors of $u$ and $v$;
	\item $\overleftarrow{\gamma}(u,v)$ the number of common in-neighbors of $u$ and $v$;
	\item $\overrightarrow{\gamma}(u,v)_{\upharpoonright w}$ the number of common out-neighbors of $u$ and $v$ with in- and out-degree as $w$;
	\item $\overleftarrow{\gamma}(u,v)_{\upharpoonright w}$ the number of common in-neighbors of $u$ and $v$ with in- and out-degree as $w$.
\end{squishlist}

Then, the change in $\overrightarrow{\butterfly{G}}^\mathrm{R}$ after the application of $(u,\alpha,+1),$ $(v,\beta,+1)$ $\rpso$ $(u,\beta,+1),$ $(v,\alpha,+1)$ is
\begin{align}
\label{eq:b1}
\Delta(\overrightarrow{\butterfly{G}}^\mathrm{R}) =& 
\sum_{\substack{w \in \inneighs{\alpha}_{\upharpoonright v} \setminus \inneighs{\beta}\\w \neq u}}\overrightarrow{\gamma}(v,w)\\
\label{eq:b2}
+& \sum_{\substack{\delta \in \outneighs{v}_{\upharpoonright \alpha} \setminus \outneighs{u}\\\delta \neq \beta}}\overleftarrow{\gamma}(\alpha,\delta)\\
\label{eq:b3}
-& \sum_{\substack{w \in \inneighs{\beta}_{\upharpoonright v} \setminus \inneighs{\alpha}\\w \neq v}}\overrightarrow{\gamma}(v,w)-1\\
\label{eq:b4}
-& \sum_{\substack{\delta \in \outneighs{v}_{\upharpoonright \beta} \setminus \outneighs{u}\\\delta \neq \beta}}\overleftarrow{\gamma}(\beta,\delta)-1\\
\label{eq:b5}
+& \sum_{\substack{w \in \inneighs{\beta}_{\upharpoonright u}\setminus \inneighs{\alpha}\\w \neq v}}\overrightarrow{\gamma}(u,w)\\
\label{eq:b6}
+& \sum_{\substack{\delta \in \outneighs{u}_{\upharpoonright \beta}\setminus \outneighs{v}\\\delta \neq \alpha}}\overleftarrow{\gamma}(\beta,\delta)\\
\label{eq:b7}
-& \sum_{\substack{w \in \inneighs{\alpha}_{\upharpoonright u}\setminus \inneighs{\beta}\\ w \neq u}}\overrightarrow{\gamma}(u,w)-1\\
\label{eq:b8}
-& \sum_{\substack{\delta \in \outneighs{u}_{\upharpoonright \alpha}\setminus \outneighs{v}\\ \delta \neq \alpha}}\overleftarrow{\gamma}(\alpha,\delta)-1\\
\label{eq:b9}
-& \sum_{\substack{w \in L_{\upharpoonright v}\\w \neq v,u}}(\overrightarrow{\gamma}(v,w)_{\upharpoonright \alpha}-1)
- \sum_{\substack{w \in L_{\upharpoonright u}\\w \neq v,u}}(\overrightarrow{\gamma}(u,w)_{\upharpoonright \beta}-1)\\
\label{eq:b11}
+& \sum_{\substack{w \in L_{\upharpoonright v}\\w \neq v,u}}(\overrightarrow{\gamma}(v,w)_{\upharpoonright \beta}-1)
+ \sum_{\substack{w \in L_{\upharpoonright u}\\w \neq v,u}}(\overrightarrow{\gamma}(u,w)_{\upharpoonright \alpha}-1)\,.
\end{align}
\Cref{eq:b1} indicates the number of additional butterflies where the left vertices have the same in- and out-degree that $v$ will form after the swap (the number of common neighbors between $v$ and $w$ increases by $1$, unless $w$ shares $\beta$ with $v$);
\Cref{eq:b2} indicates the number of additional butterflies where the right vertices have the same in- and out-degree that $\alpha$ will form after the swap (the number of common neighbors between $\alpha$ and $\delta$ increases by $1$, unless $\delta$ shares $u$ with $\alpha$);
\Cref{eq:b3} indicates the number of butterflies where the left vertices have the same in- and out-degree that $v$ will not form anymore after the swap (the number of common neighbors between $v$ and $w$ decreases by $1$, unless $w$ is neighbor of $\alpha$);
\Cref{eq:b4} indicates the number of butterflies where the right vertices have the same in- and out-degree that $\beta$ will not form anymore after the swap (the number of common neighbors between $\beta$ and $\delta$ decreases by $1$, unless $\delta$ is neighbor of $u$);
\Cref{eq:b5} indicates the number of additional butterflies where the left vertices have the same in- and out-degree that $u$ will form after the swap (the number of common neighbors between $u$ and $w$ increases by $1$, unless $w$ shares $\alpha$ with $u$);
\Cref{eq:b6} indicates the number of additional butterflies where the right vertices have the same in- and out-degree that $\beta$ will form after the swap (the number of common neighbors between $\beta$ and $\delta$ increases by $1$, unless $\delta$ shares $v$ with $\beta$);
\Cref{eq:b7} indicates the number of butterflies where the left vertices have the same in- and out-degree that $u$ will not form anymore after the swap (the number of common neighbors between $u$ and $w$ decreases by $1$, unless $w$ is neighbor of $\beta$);
\Cref{eq:b8} indicates the number of butterflies where the right vertices have the same in- and out-degree that $\alpha$ will not form anymore after the swap (the number of common neighbors between $\alpha$ and $\delta$ decreases by $1$, unless $\delta$ is neighbor of $v$);
\Cref{eq:b9} indicates the number of butterflies counted twice in \Cref{eq:b1} and \Cref{eq:b2} (resp. \Cref{eq:b5} and \Cref{eq:b6}) because they involve both left and right vertices with the same in- and out-degree; 
\Cref{eq:b11} indicates the number of butterflies removed two times in \Cref{eq:b3} and \Cref{eq:b4} (resp. \Cref{eq:b7} and \Cref{eq:b8}) because they involve both left and right vertices with the same in- and out-degree.

Similarly, the change in $\overleftarrow{\butterfly{G}}^\mathrm{R}$ due to the RPSO $(u,\alpha,-1),$ $(v,\beta,-1)$ $\rpso$ $(u,\beta,-1),$ $(v,\alpha,-1)$ is 
\begin{align*}
\Delta(\overleftarrow{\butterfly{G}}^\mathrm{R}) &=
\sum_{\substack{w \in \outneighs{\alpha}_{\upharpoonright v} \setminus \outneighs{\beta}\\w \neq u}}\overleftarrow{\gamma}(v,w) - \sum_{\substack{w \in \outneighs{\beta}_{\upharpoonright v} \setminus \outneighs{\alpha}\\w \neq v}}\left(\overleftarrow{\gamma}(v,w)-1\right)\\
&+ \sum_{\substack{\delta \in \inneighs{v}_{\upharpoonright \alpha} \setminus \inneighs{u}\\\delta \neq \beta}}\overrightarrow{\gamma}(\alpha,\delta) - \sum_{\substack{\delta \in \inneighs{u}_{\upharpoonright \alpha}\setminus \inneighs{v}\\ \delta \neq \alpha}}\left(\overrightarrow{\gamma}(\alpha,\delta)-1\right)\\
&+ \sum_{\substack{\delta \in \inneighs{u}_{\upharpoonright \beta}\setminus \inneighs{v}\\\delta \neq \alpha}}\overrightarrow{\gamma}(\beta,\delta) - \sum_{\substack{\delta \in \inneighs{v}_{\upharpoonright \beta} \setminus \inneighs{u}\\\delta \neq \beta}}\left(\overrightarrow{\gamma}(\beta,\delta)-1\right)\\
&+ \sum_{\substack{w \in \outneighs{\beta}_{\upharpoonright u}\setminus \outneighs{\alpha}\\w \neq v}}\overleftarrow{\gamma}(u,w) - \sum_{\substack{w \in \outneighs{\alpha}_{\upharpoonright u}\setminus \outneighs{\beta}\\ w \neq u}}\left(\overleftarrow{\gamma}(u,w)-1\right)\\
&+ \sum_{\substack{w \in L_{\upharpoonright v}\\w \neq v,u}}\overleftarrow{\gamma}(v, w)_{\upharpoonright \beta} - \overleftarrow{\gamma}(v,w)_{\upharpoonright \alpha} + \sum_{\substack{w \in L_{\upharpoonright u}\\w \neq v,u}}\overleftarrow{\gamma}(u, w)_{\upharpoonright \alpha} - \overleftarrow{\gamma}(u,w)_{\upharpoonright \beta}\,. 
\end{align*}

%% file: sections/data.tex
\section{Data}\label{ax:data}

We showcased the flexibility of \algo, considering both directed and undirected hypergraphs from various domains.
All the datasets used in our analyses are publicly available on GitHub~\footnote{\url{https://github.com/lady-bluecopper/NuDHy}}.

\Cref{tbl:data_senate_house} reports the main characteristics of the directed hypergraphs representing sponsor-cosponsor relationships in Senate bills (\textsc{S-bills}) and House bills (\textsc{H-bills}) from the $93^{\mathrm{rd}}$ to the $108^{\mathrm{th}}$ Congresses. We exploited these datasets in the group affinity analysis presented in \Cref{sec:app_congress}.

\begin{table}[ht!]
\caption{Characteristics of \textsc{H-bills} and \textsc{S-bills} per session: starting year, majority of seats, num. of legislators, num. of Republicans and of Democrats, num. of bills sponsored by a Republican and by a Democrat, mean num. of bills sponsored by a Republican and by a Democrat, and mean num. of bills co-sponsored by a Republican and by a Democrat.}
\label{tbl:data_senate_house}
\begin{subtable}{\linewidth}
\caption{House}
\resizebox{\columnwidth}{!}{
\begin{tabular}{rccrrrrrrrrr}
\toprule
 & Start & Majority & Legislators & Rep. & Dem. & Rep. Bills & Dem. Bills & Sp. Rep. & Sp. Dem. & Cosp. Rep. & Cosp. Dem. \\
\midrule
\textbf{93} & 1973 & Dem. & 440 & 192 & 241 & 1612 & 3237 & 9.006 & 13.953 & 82.574 & 126.751 \\
\textbf{94} & 1975 & Dem. & 441 & 144 & 291 & 1556 & 3927 & 11.358 & 14.385 & 88.231 & 129.823 \\
\textbf{95} & 1977 & Dem. & 440 & 143 & 292 & 1758 & 4043 & 12.468 & 15.257 & 123.143 & 140.949 \\
\textbf{96} & 1979 & Dem. & 438 & 156 & 276 & 965 & 2144 & 6.307 & 8.152 & 155.975 & 156.119 \\
\textbf{97} & 1981 & Dem. & 440 & 191 & 243 & 1230 & 2017 & 6.543 & 8.732 & 176.933 & 188.984 \\
\textbf{98} & 1983 & Dem. & 439 & 164 & 269 & 1023 & 2332 & 6.354 & 8.969 & 215.174 & 281.838 \\
\textbf{99} & 1985 & Dem. & 438 & 181 & 252 & 1222 & 2305 & 6.983 & 9.486 & 241.687 & 318.203 \\
\textbf{100}  & 1987 & Dem. & 441 & 177 & 258 & 1302 & 2367 & 7.750 & 9.430 & 252.737 & 329.531 \\
\textbf{101}  & 1989 & Dem. & 440 & 174 & 259 & 1370 & 2660 & 7.874 & 10.391 & 287.596 & 356.031 \\
\textbf{102}  & 1991 & Dem. & 441 & 167 & 267 & 1323 & 2578 & 8.067 & 9.954 & 264.347 & 317.478 \\
\textbf{103}  & 1993 & Dem. & 441 & 176 & 258 & 1286 & 2134 & 7.565 & 8.570 & 226.028 & 223.569 \\
\textbf{104}  & 1995 & Rep. & 439 & 230 & 204 & 1640 & 1041 & 7.354 & 5.627 & 150.013 & 142.913 \\
\textbf{105}  & 1997 & Rep. & 444 & 226 & 207 & 1865 & 1294 & 8.216 & 6.811 & 168.740 & 209.708 \\
\textbf{106}  & 1999 & Rep. & 437 & 223 & 211 & 2176 & 1537 & 9.982 & 7.319 & 192.879 & 284.566 \\
\textbf{107}  & 2001 & Rep. & 442 & 221 & 211 & 2049 & 1756 & 9.230 & 8.566 & 170.009 & 298.207 \\
\textbf{108}  & 2003 & Rep. & 439 & 229 & 204 & 2055 & 1718 & 9.215 & 8.422 & 166.415 & 298.222 \\
\bottomrule
\end{tabular}
}
\end{subtable}
\begin{subtable}{\linewidth}
\caption{Senate}
\resizebox{\columnwidth}{!}{
\begin{tabular}{rccrrrrrrrrr}
\toprule
 & Start & Majority & Legislators & Rep. & Dem. & Rep. Bills & Dem. Bills & Sp. Rep. & Sp. Dem. & Cosp. Rep. & Cosp. Dem. \\
\midrule
\textbf{93} & 1973 & Dem. & 101 & 44 & 54 & 571 & 1054 & 13.595 & 18.491 & 105.833 & 140.000 \\
\textbf{94} & 1975 & Dem. & 100 & 37 & 60 & 526 & 943 & 14.216 & 15.459 & 101.703 & 101.984 \\
\textbf{95} & 1977 & Dem. & 104 & 38 & 61 & 521 & 942 & 13.711 & 14.952 & 99.395 & 88.969 \\
\textbf{96} & 1979 & Dem. & 101 & 41 & 58 & 526 & 866 & 13.150 & 14.931 & 109.927 & 94.797 \\
\textbf{97} & 1981 & Rep. & 101 & 53 & 46 & 882 & 612 & 16.642 & 13.304 & 137.167 & 167.109 \\
\textbf{98} & 1983 & Rep. & 101 & 54 & 46 & 1068 & 648 & 19.418 & 14.087 & 176.345 & 221.391 \\
\textbf{99} & 1985 & Rep. & 101 & 53 & 46 & 1150 & 722 & 21.698 & 15.696 & 192.593 & 244.149 \\
\textbf{100}  & 1987 & Dem. & 101 & 45 & 55 & 770 & 1120 & 16.739 & 20.364 & 243.826 & 285.364 \\
\textbf{101} & 1989 & Dem. & 100 & 45 & 55 & 760 & 1288 & 16.889 & 23.418 & 266.822 & 300.527 \\
\textbf{102} & 1991 & Dem. & 102 & 44 & 56 & 722 & 1241 & 16.409 & 21.772 & 240.795 & 263.655 \\
\textbf{103} & 1993 & Dem. & 101 & 43 & 57 & 535 & 985 & 12.159 & 17.589 & 165.091 & 168.228 \\
\textbf{104} & 1995 & Rep. & 102 & 53 & 47 & 778 & 431 & 14.679 & 8.979 & 103.481 & 75.708 \\
\textbf{105} & 1997 & Rep. & 100 & 55 & 45 & 912 & 564 & 16.582 & 12.533 & 123.764 & 126.467 \\
\textbf{106} & 1999 & Rep. & 102 & 55 & 45 & 1087 & 822 & 19.411 & 18.267 & 154.036 & 200.543 \\
\textbf{107} & 2001 & Dem. & 101 & 50 & 50 & 778 & 1084 & 15.878 & 21.680 & 119.080 & 195.245 \\
\textbf{108} & 2003 & Rep. & 100 & 51 & 48 & 953 & 926 & 18.686 & 19.292 & 116.941 & 206.188 \\
\bottomrule
\end{tabular}
}
\end{subtable}
\end{table}

\Cref{tbl:contact_networks} reports the main characteristics of the undirected hypergraphs representing \textbf{(i)} face-to-face interactions among children in a primary school in Lyon, France~\cite{gemmetto2014mitigation} (\textsc{lyon}) and among students in a high school in Lyc\'{e}e Thiers, France~\cite{mastrandrea2015contact} (\textsc{high}), and \textbf{(ii)} email exchanges between members of a European research institution (\textsc{email-EU}) and between Enron employees (\textsc{email-Enron})~\cite{benson2018simplicial}.
We exploited these datasets in the non-linear contagion analysis presented in \Cref{sec:app_contagion}.

\begin{table}[!th]
    \small
    \centering
    \caption{Characteristics of the contact networks: number of nodes, number of hyperedges, max size of a hyperedge, mean and std hyperedge size, mean and std node degree, number of steps performed by \algo, invasion thresholds for the linear ($l$) and super-linear ($sl$) case, and bistability threshold.}
    \label{tbl:contact_networks}
    \resizebox{\columnwidth}{!}{
  \begin{tabular}{lrrrrrrrrrrr}
  \toprule
  \textbf{Network} & \multicolumn{1}{c}{$\mathbf{|V|}$} & \multicolumn{1}{c}{$\mathbf{|E|}$} & \multicolumn{1}{c}{$\mathbf{d}$} & \multicolumn{1}{c}{$\mu_{|e|}$} & \multicolumn{1}{c}{$\sigma_{|e|}$} & \multicolumn{1}{c}{$\mu_{\mathrm{deg}}$} & \multicolumn{1}{c}{$\sigma_{\mathrm{deg}}$} & $s$ & $\lambda_c^l$ & $\lambda_c^{sl}$ & $\nu_c$\\
  \midrule
  \textsc{lyon} & \num{243} & \num{1188} & \num{5} & \num{2.40} & \num{0.52} & \num{11.79} & \num{5.59} & \num{57060} & 0.0474 & 0.0382 & 2.5415\\
  \textsc{high} & \num{327} & \num{7818} & \num{5} & \num{2.33} & \num{0.53} & \num{55.63} & \num{27.06} & \num{363840} & 0.0101 & 0.0096 & 2.4337\\
  \textsc{email-Enron} & \num{143} & \num{1512} & \num{18} & \num{3.00} & \num{1.95} & \num{31.82} & \num{24.22} & \num{227500} & 0.0060 & 0.0025 & 1.3182 \\
  \textsc{email-Eu} & \num{998} & \num{25027} & \num{25} & \num{3.42} & \num{2.84} & \num{85.91} & \num{114.23} & \num{1714740} & 0.0009 & 0.0008 & 1.2313\\
  \bottomrule
  \end{tabular}
  }
\end{table}

\Cref{tbl:datasets_eco} reports the main characteristics of the directed hypergraphs generated from international trade data~\cite{hausmann2014atlas} of four years: 1995, 2009, 2019, and 2020.
The head of each hyperedge includes countries that export the product, while the tail consists in countries that import the product. We follow the standard economics literature~\cite{cristelli2013measuring} and consider a country to be an exporter of a product if its \emph{Revealed Comparative Advantage}~\cite{balassa1965trade} is greater than $1$, and to be an importer of a product if its \emph{Revealed Comparative Disadvantage}~\cite{krugman2009international} is greater than $1$.
We follow~\cite{harvard} and include only countries with population above $1$ million and average trade above $1$ billion USD.
We exploited these datasets in the economic complexity analysis presented in \Cref{sec:app_trade}.

\begin{table}[!th]
    \small
    \centering
    \caption{Characteristics of the trade hyper-networks: number of nodes, number of hyperedges, max size of a hyperedge, average head size, average tail size, average node in-degree, average node out-degree, and number of steps of \algo.}
    \label{tbl:datasets_eco}
    \resizebox{\columnwidth}{!}{
  \begin{tabular}{lrrrrrrrr}
  \toprule
  \textbf{Dataset} & \multicolumn{1}{c}{$\mathbf{|V|}$} & \multicolumn{1}{c}{$\mathbf{|E|}$} & \multicolumn{1}{c}{$\mathbf{d}$} & \multicolumn{1}{c}{$\bar{|h|}$} & \multicolumn{1}{c}{$\bar{|t|}$} & \multicolumn{1}{c}{$\overline{\indeg{}{v}}$} & \multicolumn{1}{c}{$\overline{\outdeg{}{v}}$} & \multicolumn{1}{c}{$s$}\\
  \midrule
  \textsc{hs1995} & \num{129} & \num{5}k & \num{107} & \num{14.76} & \num{32.95} & \num{1.29}k & \num{576.3} & \num{12}m\\
  \textsc{hs2009} & \num{133} & \num{4.9}k & \num{113} & \num{16.43} & \num{35.59} & \num{1.3}k & \num{604.4} & \num{13}m\\
  \textsc{hs2019} & \num{133} & \num{4.6}k & \num{120} & \num{16.24} & \num{37.31} & \num{1.29}k & \num{563.9}& \num{12}m\\ 
  \textsc{hs2020} & \num{133} & \num{4.6}k & \num{118} & \num{15.92} & \num{37.32} & \num{1.29}k & \num{552.4} & \num{12}m\\ 
  \bottomrule
  \end{tabular}
  }
\end{table}

In the extended analysis presented in \Cref{ax:additional}, we consider the following datasets, whose characteristics are summarized in \Cref{tbl:datasets}.
\textbf{\textsc{ecoli}}~\cite{shen2018genome} is a metabolic hypergraph constructed with data from the Kyoto Encyclopedia of Genes and Genomes (KEGG) pathway database. Specifically, we consider the pathway \emph{eco01100} of Escherichia coli.
\textbf{\textsc{iaf1260b}} and \textbf{\textsc{ijo1366}}~\cite{kim2022reciprocity} model chemical reactions among genes. Each node is a gene, and each hyperedge is a reaction.
\textbf{\textsc{dblp-9}}~\cite{tang2008arnetminer} and \textbf{\textsc{cit-sw}}~\cite{kim2022reciprocity} are DBLP citation hypergraphs, where nodes are authors and hyperedges are  citations. The head of each hyperedge is a paper represented as its set of authors, while the tail is the set of authors of the paper cited by the head. 
\textbf{\textsc{math}}~\cite{kim2022reciprocity} is a question-answering dataset obtained from the MathOverflow website. Each node is a user, and each hyperedge is a post. The questioner is the head and the answerers are the tail.
\textbf{\textsc{enron}}~\footnote{\url{(https://www.cs.cmu.edu/~enron)}} is a communication network of emails sent. For each directed hyperedge, the head is the sender and the tail contains the recipients.
\textbf{\textsc{ord}}~\cite{kearnes2021open} is a hypergraph modeling chemical reactions from the Open Reaction Database (ORD). Each hyperedge is a chemical reaction, where the head is the set of reagents and the tail is the set of products generated from the reaction.

\begin{table}[!th]
    \small
    \centering
    \caption{Characteristics of the datasets: number of nodes, number of hyperedges, max size of a hyperedge, average head size, average tail size, average node in-degree, average node out-degree, and number of steps of \algo.}
    \label{tbl:datasets}
    \resizebox{\columnwidth}{!}{
  \begin{tabular}{lrrrrrrrr}
  \toprule
  \textbf{Dataset} & \multicolumn{1}{c}{$\mathbf{|V|}$} & \multicolumn{1}{c}{$\mathbf{|E|}$} & \multicolumn{1}{c}{$\mathbf{d}$} & \multicolumn{1}{c}{$\bar{|h|}$} & \multicolumn{1}{c}{$\bar{|t|}$} & \multicolumn{1}{c}{$\overline{\indeg{}{v}}$} & \multicolumn{1}{c}{$\overline{\outdeg{}{v}}$} & \multicolumn{1}{c}{$s$}\\
  \midrule
  \textsc{ecoli} & \num{702} & \num{923} & \num{9} & \num{2.02} & \num{2.26} & \num{3.99} & \num{3.55} & \num{79}k\\
  \textsc{iaf1260b} & \num{1.7}k & \num{2}k & \num{12} & \num{2.27} & \num{2.00} & \num{2.72} & \num{3.60} & \num{178}k\\
  \textsc{ijo1366} & \num{1.8}k & \num{2.2}k & \num{13} & \num{2.27} & \num{2.03} & \num{2.75} & \num{3.54} & \num{193}k\\
  \textsc{cit-sw} & \num{16}k & \num{53}k & \num{19} & \num{2.72} & \num{2.93} & \num{10.65} & \num{13.17} & \num{6}m\\
  \textsc{math} & \num{35}k & \num{94}k & \num{213} & \num{1.00} & \num{1.78} & \num{3.57} & \num{10.54} & \num{5.2}m\\
  \textsc{dblp-9} & \num{21}k & \num{95}k & \num{30} & \num{2.44} & \num{2.46} & \num{26.45} & \num{12.21} & \num{9.3}m\\
  \textsc{enron} & \num{57}k & \num{149}k & \num{30} & \num{1.00} & \num{3.98} & \num{13.12} & \num{6.27} & \num{15}m\\
  \textsc{ord} & \num{633}k & \num{479}k & \num{20} & \num{4.52} & \num{1.03} & \num{1.13} & \num{6.64} & \num{53}m\\
  \bottomrule
  \end{tabular}
  }
\end{table}

For each dataset and each sampler, we generated $33$ samples by performing $s = 20w$ steps in the corresponding Markov graph, where $w$ is the number of edges in the bipartite graph representation of the dataset.

\section{Baselines}\label{ax:base}

In \Cref{ax:additional} we compared \algo with \redi~\cite{kim2022reciprocity}, the first realistic generative model specifically designed for directed hypergraphs.
\redi extends the preferential attachment model proposed in \cite{do2020structural} to directed hypergraphs, allowing the generation of random hypergraphs exhibiting reciprocal patterns akin to those observed in real directed hypergraphs.
This generator requires three input parameters: the number of nodes, a proportion $\beta_1$ of reciprocal hyperedges, and the extent $\beta_2$ of reciprocity between a hyperedge and its reciprocal counterpart.
The random hypergraph generated preserves, on average, the distribution of head and tail sizes, and the distribution of number of hyperedges in which each group of nodes appears.

However, \redi suffers from three main limitations.
Firstly, it does not account for situations where a node is related to itself, which is common in real-world hypergraphs such as citation networks. 
Secondly, precise tuning of $\beta_1$ and $\beta_2$ is necessary to ensure that the reciprocity measured in the random samples (see \Cref{eq:reciprocity}) aligns with that of the observed hypergraph. 
Achieving this tuning can be challenging due to the computational complexity of calculating reciprocity ($O(2^{|E|})$) and the possibility of reciprocal hyperedges occurring by chance. Thus, $\beta_1$ effectively represents a "lower bound" for the actual ratio of reciprocated hyperedges. 
To circumvent this limitation, we set $\beta_1 = \beta_2 = 0$ as recommended by the authors of \redi.
Thirdly, the procedures used by \redi to preserve the distribution of the number of hyperedges in which each node group participates are computationally intensive, resulting in significantly increased runtime as the hypergraph size grows. 
To address this issue, \cite{kim2022reciprocity} introduced an alternative generator that, instead of node group degrees, preserves node degrees on average.
In our analyses, we considered both the generator that preserves node group degrees (dubbed \base) and the one that preserves node degrees (dubbed \based).
We note that for the largest dataset in \Cref{tbl:datasets}, \base was not able to generate a random sample within 24 hours.

Furthermore, we implemented also a naive sampler, \nullm, that preserves the head and tail size distributions, but populates the hyperedges of the random hypergraph by drawing nodes uniformly at random from the set of nodes in the observed hypergraph.

%% file: sections/convergence.tex
\section{Convergence}\label{ax:convergence}

To study the convergence of \algo, we follow a procedure similar to the 
one proposed by~\cite{gionis2007assessing}. 
Given a hypergraph $H \doteq (V, E)$, we generate two datasets $\dataset$ of \emph{itemsets}:
in $\dataset_{V,1}$ each node in $V$ is an item and each itemset is the head of a hyperedge in $E$; and 
in $\dataset_{V,-1}$ each node in $V$ is an item and each itemset is the tail of a hyperedge in $E$.
The mixing time of the \algo samplers is then estimated by looking at the convergence of the \emph{Average Relative Support Difference (ARSD)}, defined for each dataset $\dataset$ as follows:

\begin{align}\label{eq:arsd}
\mathrm{ARSD}(\dataset^s) = \frac{1}{\left|\textsc{FI}_{f,l}(\dataset)\right|}\sum_{A \in \textsc{FI}_{f,l}(\dataset)}\frac{\left|\supp{\dataset}{A} - \supp{\dataset^s}{A}\right|}{\supp{\dataset}{A}}\,,
\end{align}

where $A$ is a set of items, $\supp{\dataset}{A}$ is the number of itemsets in $\dataset$ containing $A$, $\textsc{FI}_{f,l}(\dataset)$ is the collection of top-$f$ frequent itemsets $A$ with length $|A| \geq l$,
and $\dataset^s$ is the dataset generated from the hypergraph obtained by the sampler after $s$ steps.
A frequent itemset in $\dataset_{V,1}$ represents a frequent co-occurrence of nodes in the heads of the hyperedges; whereas 
a frequent itemset in $\dataset_{V,-1}$ represents a frequent co-occurrence of vertices in the tails.

In our experiments we set $f = 20$ and $l = 3$.
\Cref{fig:convergence} illustrates the ARSD values for each dataset in \Cref{tbl:datasets}, for $s = \lceil k \cdot w \rceil$ with $k \in [0, 50]$ and $w = \sum_{A \in D}|A|$, for \algoA and \algoC.
We observe that, in general, \algoC requires fewer steps to converge, especially for larger hypergraphs.

\begin{figure}[!th]
    \centering
    \includegraphics[width=\linewidth]{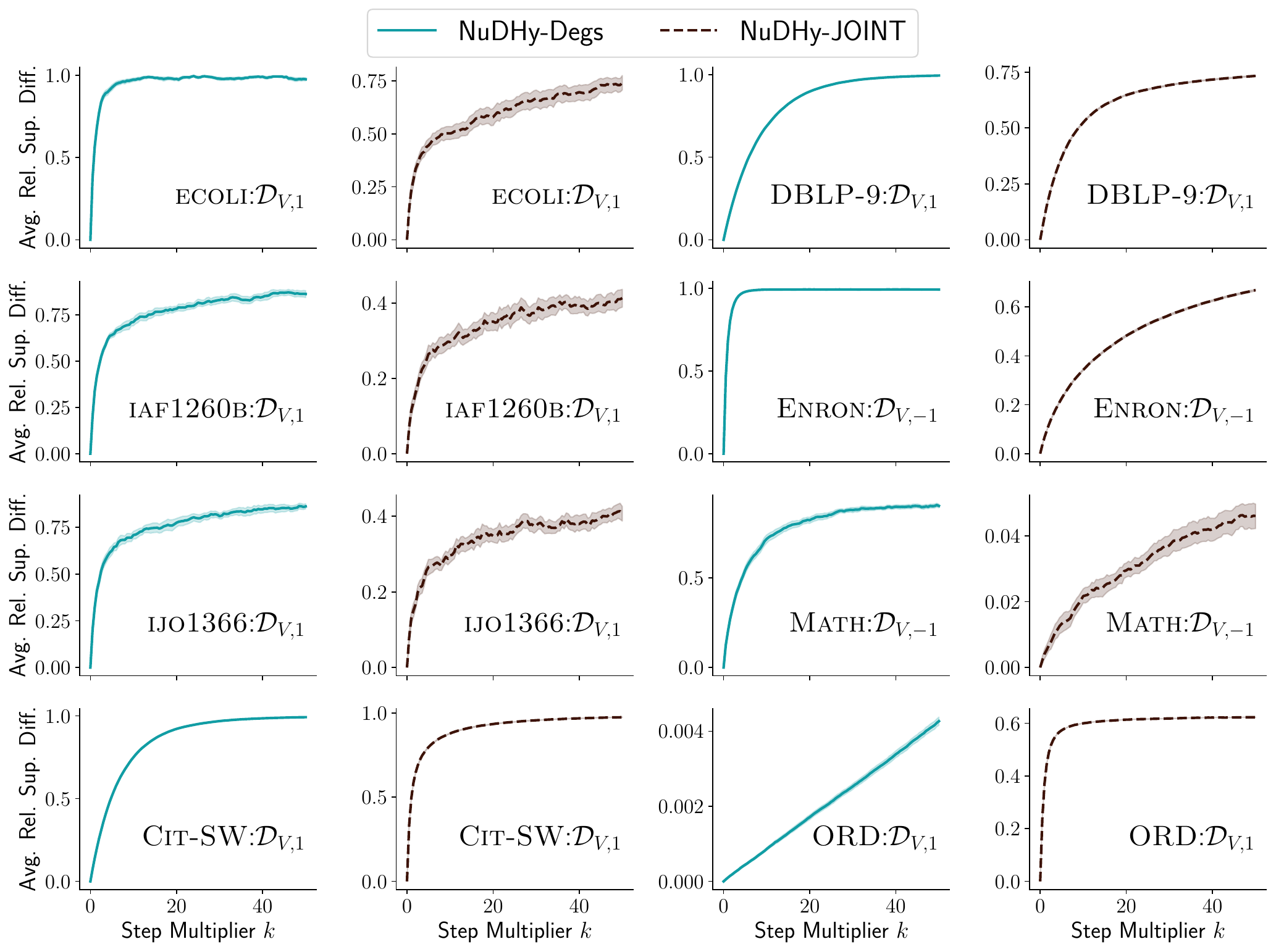}
    \caption{Convergence experiment for each dataset, for \algoA and \algoC.}
    \label{fig:convergence}
\end{figure}

We note that although \algoB and \algoD (not shown in the charts) require fewer steps to reach convergence due to their acceptance probability mechanism, which prioritizes visiting dense regions more frequently and sparse regions less often, their overall speed is hindered by the need to compute the degree of a state, which has a time complexity of $O(|V|^2)$. 
This computation adds a significant overhead, and as a result, the time saved by performing fewer steps is counterbalanced by the increased complexity of the step operations. Consequently, \algoB ends up being slower than \algoA, and \algoD slower than \algoC.
As an example, to perform the same amount of steps, for the small datasets (\textsc{ecoli}, \textsc{iaf1260b}, \textsc{ijo1366}), we report a time increase for \algoB from $5.86$x to $13.49$x; whereas for the mid-size datasets (\textsc{cit-sw}, \textsc{math}, \textsc{dblp-9}, \textsc{enron}), it is from $23.59$x to $187.2$x. 

%% file: sections/appendix.tex
\section{Extended Experimental Evaluation of \algo}\label{ax:additional}

In this section, we delve into the experimental evaluation of our framework, which considered various metrics either specifically designed or extended to accommodate the unique characteristics of directed hypergraphs. Our primary objective is to discern which samplers produce hypergraphs that closely mimic the structural patterns of the original hypergraph.
Our second objective is to shed light on the factors influencing the observed characteristics and, in particular, understand whether the properties of the original hypergraph preserved by the random hypergraphs generated by the samplers are among such factors.
To this aim, we compare the values of the metrics calculated in the original hypergraph against those derived from $33$ random hypergraphs generated by each sampler.

\subsection{Higher-order Reciprocity}

Reciprocity~\cite{newman2002email} is a statistic originally defined for directed graphs, and later on extended to directed hypergraphs~\cite{kim2022reciprocity}. 
It quantifies how mutually nodes are linked and thus allows a better understanding of the organizing principles of the hypergraph.
The reciprocity of a hyperedge $e$ is measured considering a set $R_e$ of hyperedges sharing vertices with $e$.
Informally, if the head (resp. tail) of $e$ overlaps with the tails (resp. heads) of the hyperedges in $R_e$ to a greater extent, then it is more reciprocal. 
Given the hyperedge $e \doteq (h,t)$, let $p_u(v)$ denote the transition probability from a node $u \in h$ to each node $v$, i.e., the probability of a random walker transiting from $u$ to $v$ when she moves to a uniform random tail-node of a uniform random arc among the hyperedges in $R_e$.
In addition, let $p_u^*$ be the optimal transition probability, i.e., the probability when $R_e = \{(t, h)\}$.
Then, the reciprocity of $e$ in the directed hypergraph $H \doteq (V,E)$ is
\[
\mathsf{r}(e, R_e) \doteq \left(\frac{1}{|R_e|}\right)^\alpha\left(1-\frac{\sum_{u \in h} \mathcal{L}(p_u, p_u^*)}{|h| \mathcal{L}_\mathrm{max}}\right)\,,
\]
where $\alpha \in (0,1]$ is a penalization term~\footnote{In our experiments we used the default value $1$e-$6$.} for large sets $R_e$, $\mathcal{L}(p_u, p_u^*)$ is the Jensen-Shannon Divergence between the distributions $p_u$ and $p_u^*$, $\mathcal{L}_\mathrm{max}$ is the maximum value taken by $\mathcal{L}$, and 
\[
R_e \doteq \arg\max_{R \subseteq E}\mathsf{r}(e, R)\,.
\]
Finally, the \emph{reciprocity} of $H$ is simply the average reciprocity of its hyperedges:
\begin{equation}\label{eq:reciprocity}
\mathsf{r}(H) \doteq \frac{1}{|E|}\sum\limits_{e \in E}\mathsf{r}(e, R_e)\,.
\end{equation}

The goal of this analysis is to demonstrate the ability of \algo to produce random hypergraphs with realistic reciprocity patterns, even though reciprocity is not inherently enforced in their design. Additionally, we aim to showcase their efficiency by generating such hypergraphs in significantly less time compared to \base and \based.

\begin{figure}[th!]
    \centering
    \begin{subfigure}{.49\linewidth}
    \centering
      \includegraphics[width=\linewidth]{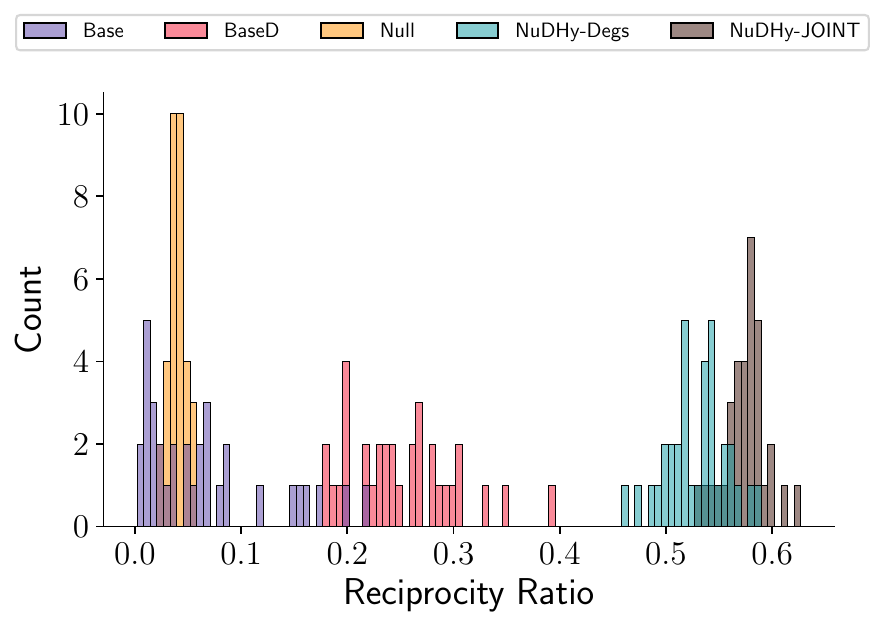}
      \caption{\textsc{ecoli}}
    \end{subfigure}
    \begin{subfigure}{.49\linewidth}
    \centering
      \includegraphics[width=\linewidth]{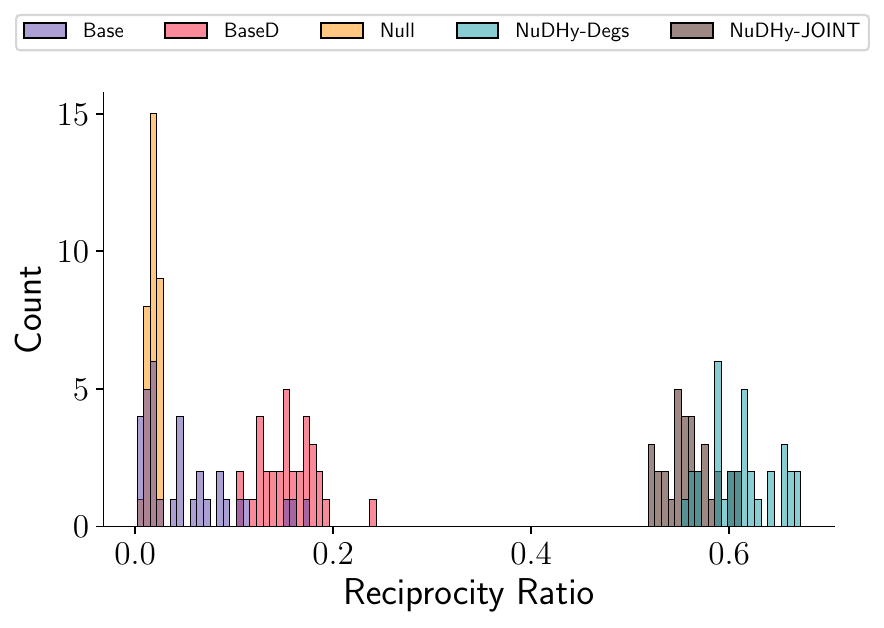}
      \caption{\textsc{iaf1260b}}
    \end{subfigure}
    \caption{Distribution of the ratios between the reciprocity of $33$ samples and the observed reciprocity, for each sampler, for \textsc{ecoli} (left) and \textsc{iaf1260b} (right).}
    \label{fig:reciprocity}
\end{figure}

\begin{figure}[th!]
    \centering
    \includegraphics[width=\linewidth]{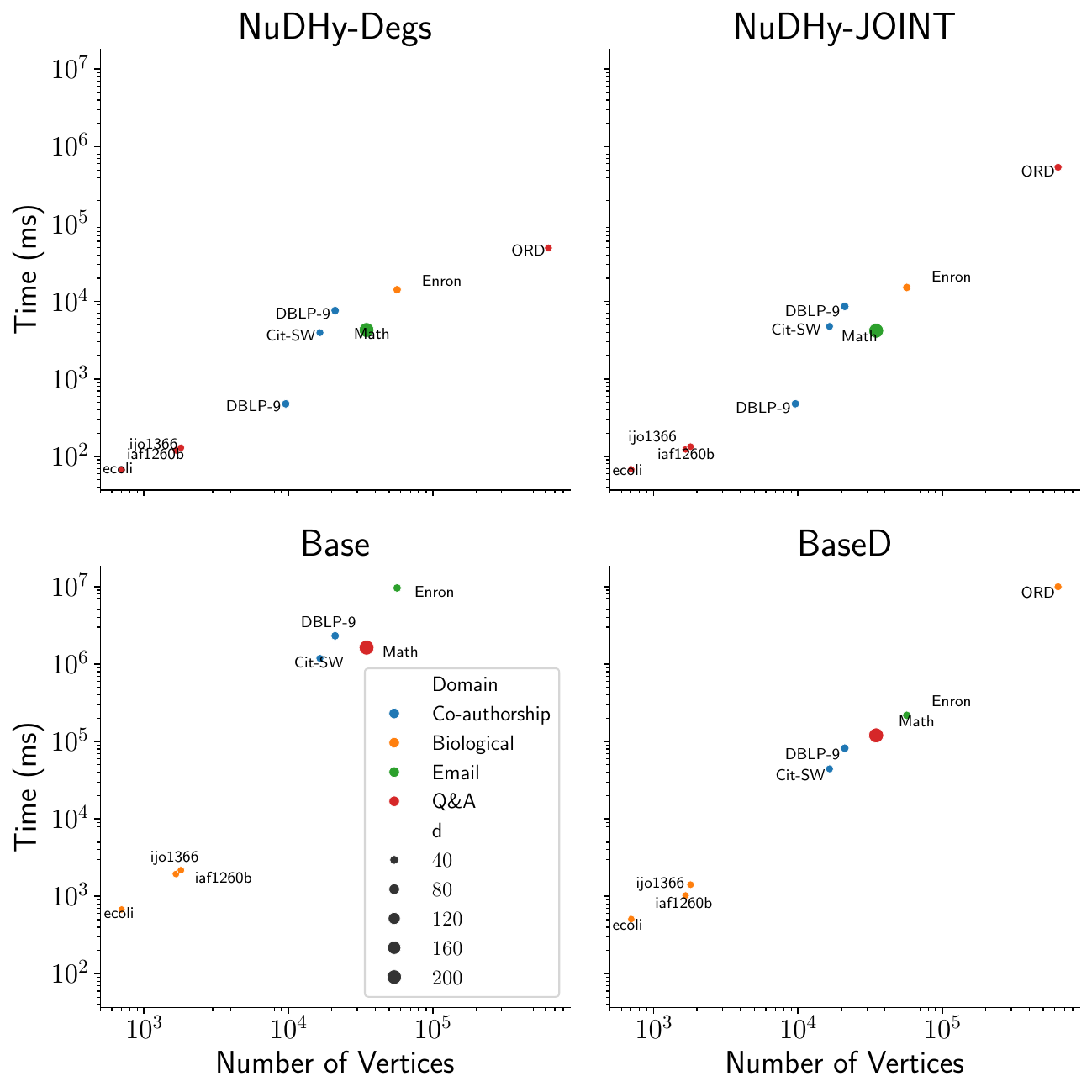}
    \caption{Mean time (ms) required to generate a sample using each of the samplers, for datasets of different size and different domain. Values are in log scale.}
    \label{fig:runtime}
\end{figure}

\Cref{fig:reciprocity} illustrates the distributions of reciprocity values obtained from $33$ samples generated by each sampler, where each value is divided by the reciprocity of the observed hypergraph.
Across all cases, the mean reciprocity values of our samples closely approximate the observed reciprocity, outperforming the mean reciprocity values derived from the samples generated by our competitors.
Specifically, the mean reciprocity among the samples generated by \algoA range from 47\% lower to 37\% lower compared to the observed reciprocity. For \algoC, this range is from 33\% to 18\%. 
In contrast, \base and \based exhibit mean reciprocity values that deviate more significantly, with values from 95\% to 93\% for \base, and from 98\% to 75\% for \based.
By looking at \Cref{fig:runtime}, we can observe that our samplers efficiently complete the generation process in few seconds across most datasets. In contrast, \base and \based require at least one order of magnitude more time to terminate the generation, with \base being unable to complete the generation process within a 24-hour time frame for the largest datasets.

Similar mean reciprocity values among the various samplers are evident only in the case of \textsc{enron}, where \algoA and \base lead to a mean reciprocity 50\% lower than the observed reciprocity, whereas \algoC 35\% lower.
This particular case stands out from the others due to the unique characteristic of all heads having size 1. 
As a result, the fact that \base preserves the \emph{undirected} group degrees (i.e., number of hyperedges that contain the group either in the head or in the tail) leads to samples that approximately maintain the \emph{directed} group degrees (i.e, number of heads and number of tails that contain the group). Consequently, the samples generated by \base exhibit similar characteristics, on average, to those produced by \algoA.

\subsection{Node Centrality}

As there are neither intuitive nor efficient measures of centrality specifically designed for directed hypergraphs, we measure the centrality of a node in a directed hypergraph, by first transforming the hypergraph into a graph, and then running two classic algorithms for node centrality in graphs: PageRank~\cite{tran2019pagerank} and HITS~\cite{kleinberg1999authoritative}.
PageRank determines the importance of a node based on its link structure, so that a node is important if it is linked to by important nodes.
HITS identifies two types of important nodes: \emph{hubs} and \emph{authorities}.
Hubs are nodes connected to many relevant nodes; while authorities are nodes with many in-going edges from hubs.
We run the PageRank algorithm on the directed weighted graph generated by creating a node for each node in the directed hypergraph and an edge from each node in a head of a hyperedge to each node in the tail of such hyperedge.
The edge weight is the number of hyperedges in which the source node is in the head and the destination node is in the tail.
On the other hand, we run the HITS algorithm on the directed bipartite graph representation of the directed hypergraph.

\begin{table}[!th]
  \small
    \centering
    \caption{Mean and standard deviation (in parentheses) of the Normalized Discounted Cumulative Gain (nDCG) of the rankings of the nodes obtained according to the Authority (Auth), Hub, and PageRank (Page) scores in $33$ samples, for each sampler.}
    \label{tab:centrality}
    \begin{tabular}{lcccccc}
    \toprule
  & Score & {\base} & {\based} & {\nullm} & {\algoA} & {\algoC} \\
  \midrule
  \multirow{3}{*}{\rotatebox{90}{\scriptsize{\textsc{ecoli}}}} & 
  Auth & 0.325 (0.081) & 0.333 (0.074) & 0.235 (0.030) & 0.588 (0.311) & 0.707 (0.370)\\
  & Hub & 0.326 (0202) & 0.354 (0.185) & 0.211 (0.018) & 0.633 (0.357) & 0.618 (0.357)\\
  & Page & 0.731 (0.053) & 0.748 (0.041) & 0.573 (0.021) & 0.928 (0.012) & 0.933 (0.009)\\
  \midrule
  \multirow{3}{*}{\rotatebox{90}{\scriptsize{\textsc{iaf1260b}}}} & 
  Auth & 0.346 (0.006) & 0.380 (0.047) & 0.347 (0.009)& 0.558 (0.157)& 0.588 (0.170)\\
  & Hub & 0.212 (0.011) & 0.290 (0.054)  & 0.213 (0.023)& 0.910 (0.225)& 0.866 (0.282)\\
  & Page & 0.475 (0.006)& 0.525 (0.024) & 0.482 (0.009) & 0.957 (0.005)& 0.960 (0.003)\\
  \midrule
  \multirow{3}{*}{\rotatebox{90}{\scriptsize{\textsc{ijo1366}}}} & 
  Auth & 0.526 (0.012) & 0.523 (0.008) & 0.527 (0.013) & 0.549 (0.023) & 0.539 (0.022)\\
  & Hub & 0.385 (0.012) & 0.374 (0.012) & 0.395 (0.018) & 0.380 (0.007) & 0.380 (0.012)\\
  & Page & 0.465 (0.008) & 0.462 (0.007) & 0.472 (0.011) & 0.959 (0.004) & 0.960 (0.002) \\
  \midrule
  \multirow{3}{*}{\rotatebox{90}{\scriptsize{\textsc{cit-sw}}}} &
  Auth & 0.695 (0.003) & 0.708 (0.010) & 0.694 (0.010) & 0.735 (0.017) & 0.739 (0.019) \\
  & Hub & 0.139 (0.035) & 0.197 (0.073) & 0.124 (0.004) & 0.890 (0.246) & 0.363 (0.010)\\
  & Page & 0.835 (0.002) & 0.849 (0.001) & 0.826 (0.002) & 0.972 (0.001) & 0.979 (0.001)\\
  \midrule
  \multirow{3}{*}{\rotatebox{90}{\scriptsize{\textsc{math}}}} &
  Auth & 0.220 (0.015) & 0.208 (0.004) & 0.206 (0.005) & 0.983 (0.001) & 0.987 (0.001)\\
  & Hub & 0.213 (0.036) & 0.195 (0.008) & 0.194 (0.007) & 0.451 (0.213) & 0.410 (0.178)\\
  & Page  & 0.609 (0.080) & 0.572 (0.002) & 0.569 (0.002) & 0.974 (0.003) & 0.977 (0.001)\\
  \midrule
  \multirow{3}{*}{\rotatebox{90}{\scriptsize{\textsc{dblp-9}}}} & 
  Auth  & 0.338 (0.066) & 0.345 (0.027) & 0.311 (0.004) & 0.567 (0.188) & 0.538 (0.158)\\
  & Hub & 0.211 (0.018) & 0.278 (0.074) & 0.196 (0.004) & 0.582 (0.122) & 0.920 (0.233)\\
  & Page & 0.738 (0.005) & 0.770 (0.004) & 0.696 (0.003) & 0.959 (0.001) & 0.967 (0.001)\\
  \midrule
  \multirow{3}{*}{\rotatebox{90}{\scriptsize{\textsc{enron}}}} & 
  Auth & 0.285 (0.015) & 0.397 (0.012) & 0.263 (0.004) & 0.798 (0.001) & 0.982 (0.000)\\
  & Hub & 0.180 (0.017) & 0.183 (0.029) & 0.156 (0.006) & 0.485 (0.214) & 0.498 (0.197)\\
  & Page & 0.792 (0.002) & 0.820 (0.002) & 0.779 (0.001) & 0.934 (0.002) & 0.979 (0.001)\\
  \midrule
  \multirow{3}{*}{\rotatebox{90}{\scriptsize{\textsc{ord}}}} & 
  Auth & $-$ & 0.803 (0.003) & 0.802 (0.000) & 0.813 (0.008) & 0.815 (0.009)\\
  & Hub & $-$ & 0.502 (0.061) & 0.178 (0.002) & 0.998 (0.000) & 0.992 (0.000)\\
  & Page & $-$  & 0.905 (0.003) & 0.903 (0.000) & 0.983 (0.000) & 0.982 (0.000)\\
  \bottomrule
  \end{tabular}
\end{table}

In this experiment we investigate whether node centrality is solely determined by node degrees or if natural patterns of node connections give rise to hubs and authorities. To assess this, we conducted the HITS and PageRank algorithms on each sample (as explained above), calculating the average scores for each node across all samples. We then evaluated the rankings generated by each score in terms of normalized Discounted Cumulative Gain (nDCG). 
nDCG~\cite{jarvelin2002cumulated} is a metric used in information retrieval to evaluate the quality of rankings. It quantifies how well a ranking aligns with the relevance of its items, under the assumption that highly relevant items should appear earlier in the list. nDCG computes a gain for each item based on its relevance and discounts this gain as items appear lower in the ranking. The gains are accumulated to form a cumulative gain, which is then normalized by the ideal cumulative gain, representing a perfect ranking.
In our experiments, we use the ranking from the observed hypergraph as measure of node relevance. 

The results for each dataset and sampler are presented in \Cref{tab:centrality}.
We note that the PageRank score is influenced by the number of incoming links to a node. While \algoA maintains the frequency of a node appearing in heads or tails, it does not preserve the frequency of appearance in heads/tails of specific sizes. Consequently, it does not fully preserve node degrees in the graph projection used for PageRank. 
Nevertheless, we observe high nDCG scores, indicating that relative node importance is preserved. 
\algoC, which maintains the frequency of a node of a certain degree appearing in heads/tails of specific sizes, achieves slightly higher scores than \algoA, due to its preservation of node degrees in the graph projection.
In the case of \base and \based, which preserve head and tail sizes only on average, there is greater variability in the PageRank scores, resulting in lower nDCG scores. This is particularly noticeable for \base, which preserves node degrees on average, whereas \based preserves group degrees.

For Authority and Hub scores, we generally observe moderate to high nDCG scores. The exceptions occur in datasets where there are either no authorities or no hubs (i.e., scores are nearly 0 for all nodes). In such cases, clear node rankings are absent, rendering the nDCG less meaningful. For instance, in \textsc{iaf1260b}, \textsc{ijo1366}, and \textsc{ecoli}, more than half of the nodes have an authority score of 0, while other nodes have scores close to 0. In \textsc{math}, hub scores are almost always 0.

The fact that we do not observe perfect nDCG scores for rankings based on Hub and Authority scores suggests that having a high number of outgoing links does not necessarily make a node a hub, nor does having a large number of incoming links make a node an authority. Instead, it underscores the significance of the actual connections between nodes in determining hub and authority status.

\subsection{Hyper-coreness Distribution}

The $k$-core decomposition is a powerful method for analyzing the structure of (hyper)graphs by dividing them into increasingly connected components.
This hierarchical decomposition provides a fingerprint of the (hyper)graph's organization, which is essential for various applications such as modeling spreading processes and quantifying influence~\cite{garcia2017understanding}.
In fact, the \emph{coreness} of a node, defined as the highest $k$ for which the node belong to the $k$-core, can measure its centrality and determine the potential impact of initiating a spreading process from that node. 
In the case of undirected hypergraphs, the core decomposition problem aims to identify a hierarchy of $(k,m)$-hyper-cores~\cite{mancastroppa2023hyper}. These hyper-cores are maximal sub-hypergraphs where each node is part of at least $k$ hyperedges with size at least $m$ within the hyper-core itself.
Then, the \emph{m-shell index} $c_m(v)$ of a node $v$ is the value $k$ at which $v$ belongs to the $(k,m)$-hyper-core but not to the $(k+1,m)$-hyper-core.
Finally, the concept of \emph{hyper-coreness}, denoted as $\mathsf{HC}$, provides a comprehensive measure of a node's centrality by considering multiple $m$-shell indices:
\[
\mathsf{HC}(v) = \sum\limits_{m=2}^M \omega(m) c_m(v) / k_m^{\mathrm{max}}\,,
\]
where $k_m^{\mathrm{max}}$ is the maximum $k$ such that there exists at least one node with $m$-shell index $k$, and $\omega(m)$ is a function to account for the significance of the different sizes of higher-order interactions~\footnote{In our experiments we set $\omega(m) = 1$ and $k_m^{\mathrm{max}} = 1$ for each $m$, to be able to compare the node rankings based on $\mathsf{HC}$ values computed in samples that may have different $\mathsf{HC}$ distributions.}.

We extend this decomposition to directed hypergraph, by defining the $(k,m)$-\textsc{h}-hyper-core $C^\textsc{h}_{k,m}$ (resp. the $(k,m)$-\textsc{t}-hyper-core $C^\textsc{t}_{k,m}$) as the maximal sub-hypergraph where each node belongs to at least $k$ heads (resp. $k$ tails) of hyperedges of size $m$, within the \textsc{h}-hyper-core (resp. \textsc{t}-hyper-core) itself.
Then, we define the $m$-\textsc{h}-shell index $c_m^\textsc{h}(v)$ and the $m$-\textsc{t}-shell index $c_m^\textsc{t}(v)$ of a node $v$ as~\cite{giatsidis2013d} 
\begin{align}
c_m^\textsc{h}(v) &= \max_{k}\left\{k \suchthat v \in C^{\textsc{h}}_{k,m}\right\}\\
c_m^\textsc{t}(v) &= \max_{k}\left\{k \suchthat v \in C^{\textsc{t}}_{k,m}\right\}
\end{align}
and the \textsc{h}-hyper-coreness and \textsc{t}-hyper-coreness of $v$ as the weighted sum of the shell indices:
\begin{align}\label{eq:hyper-coreness}
\mathsf{HC}^\textsc{h}(v) &= \sum\limits_{m=2}^M \omega(m) c_m^\textsc{h}(v) / k_m^{\textsc{h}}\\
\mathsf{HC}^\textsc{t}(v) &= \sum\limits_{m=2}^M \omega(m) c_m^\textsc{t}(v) / k_m^{\textsc{t}}
\end{align}
where $k_m^{\textsc{h}}$ (resp. $k_m^{\textsc{t}}$) is the max value $k$ for which $C^{\textsc{h}}_{k,m} \neq \varnothing$ (resp. $C^{\textsc{t}}_{k,m} \neq \varnothing$). 

\begin{figure}[!ht]
    \centering
    \begin{subfigure}{\linewidth}
      \centering
      \includegraphics[width=\linewidth]{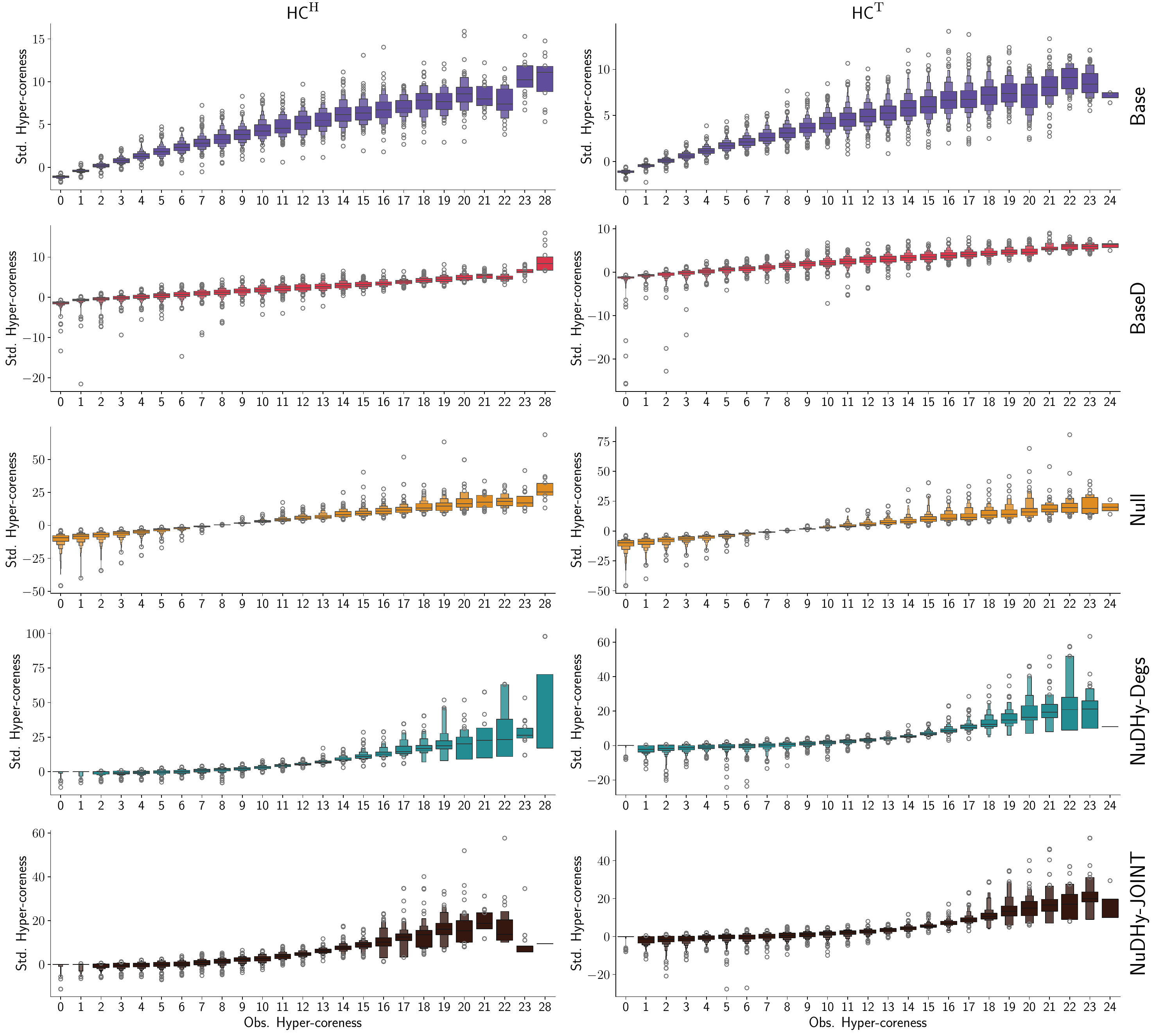}
    \end{subfigure}
    \caption{\textsc{dblp-9}: Distribution of the standardized $\mathsf{HC}^\textsc{h}$ (left) and $\mathsf{HC}^\textsc{t}$ scores, for each value measured in the observed hypergraph, for each sampler.}
    \label{fig:hyper_dblp9}
\end{figure}

\begin{figure}[!ht]
    \centering
    \begin{subfigure}{\linewidth}
      \centering
      \includegraphics[width=\linewidth]{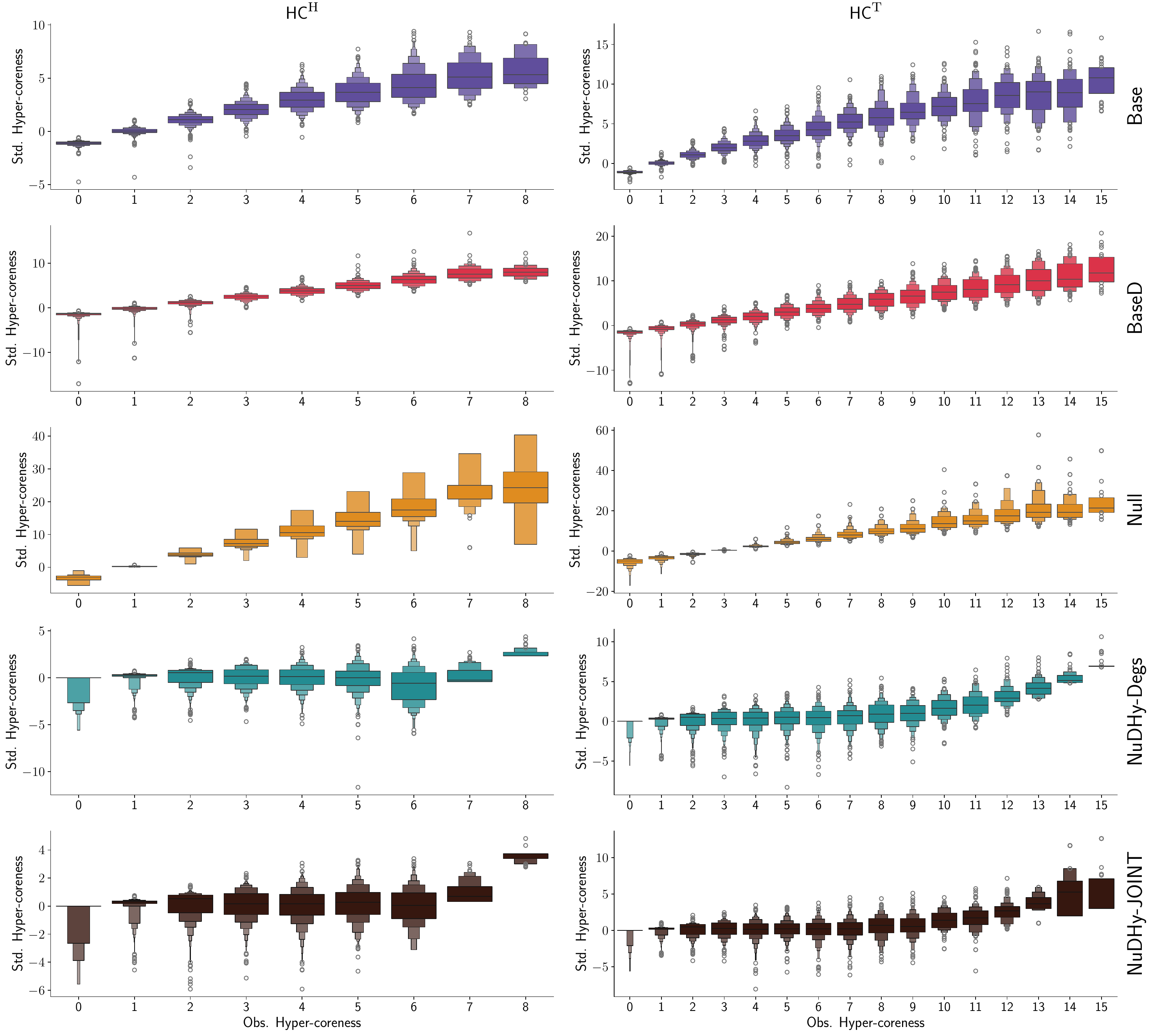}
    \end{subfigure}
    \caption{\textsc{math}: Distribution of the standardized $\mathsf{HC}^\textsc{h}$ (left) and $\mathsf{HC}^\textsc{t}$ scores, for each value measured in the observed hypergraph, for each sampler.}
    \label{fig:hyper_math}
\end{figure}

\Cref{fig:hyper_dblp9} and \Cref{fig:hyper_math} illustrate the distribution of standardized $\mathsf{HC}^\textsc{h}$ (left) and $\mathsf{HC}^\textsc{t}$ (right) scores, for each sampler. To obtain these scores, we took the observed score of a node, subtract its mean score in the samples, and divide by its standard deviation. 
In general, we find that scores tend to deviate further from $0$ for nodes with stronger connections. This outcome aligns with our expectations because nodes with low hyper-coreness are typically associated with low degrees, and by shuffling the connections of low-degree nodes, it is not possible to significantly boost their hyper-coreness. Conversely, connecting nodes with high degrees to nodes with low degrees can diminish the hyper-coreness of the former. 
Given that \based preserves, on average, the appearances of node groups, it also tends to maintain the neighborhood of each node. Consequently, a well-connected node in the observed hypergraph typically corresponds to a densely-connected node in a sample. This results in similar average hyper-coreness values and standardized scores that are close to $0$.
On the other hand, \nullm does not preserve the degrees of the nodes, leading to scenarios in which nodes that are sparsely-connected in the observed hypergraph become densely-connected in the samples. This can result in negative standardized scores.
Finally, by considering the samples generated by \algoA, we can identify relevant nodes with hyper-coreness higher than random. In \textsc{dblp-9}, these nodes can correspond to authors who cite/are cited by many authors that, in turn, cite/are cited by numerous other authors. This suggests a tendency for papers to reference similar groups of co-cited authors. This pattern is less pronounced when \algoC is employed, as it maintains the degrees of the neighbors of each node.

\begin{figure*}[!ht]
    \centering
    \begin{subfigure}{\linewidth}
      \centering
      \caption{\textsc{dblp-9}}
      \includegraphics[width=.9\linewidth]{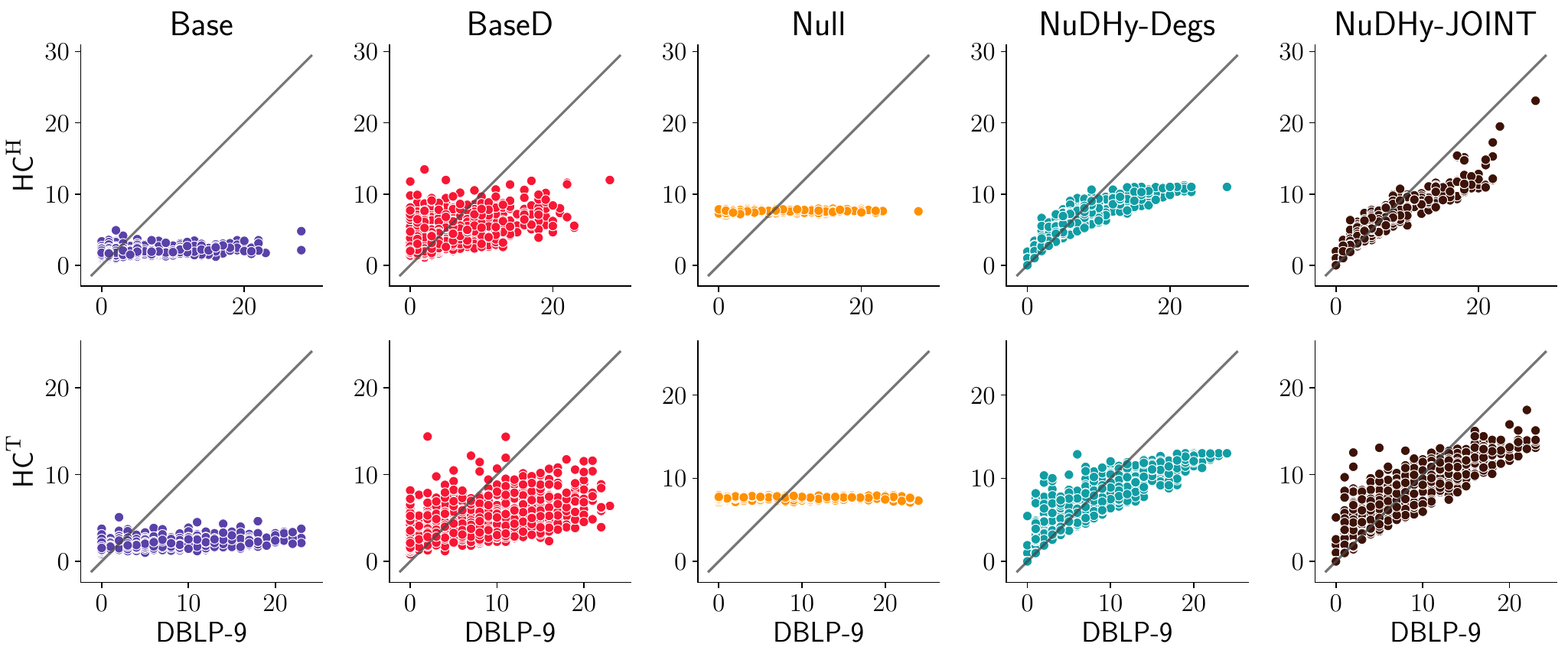}
    \end{subfigure}
    \begin{subfigure}{\linewidth}
      \centering
      \caption{\textsc{math}}
      \includegraphics[width=.9\linewidth]{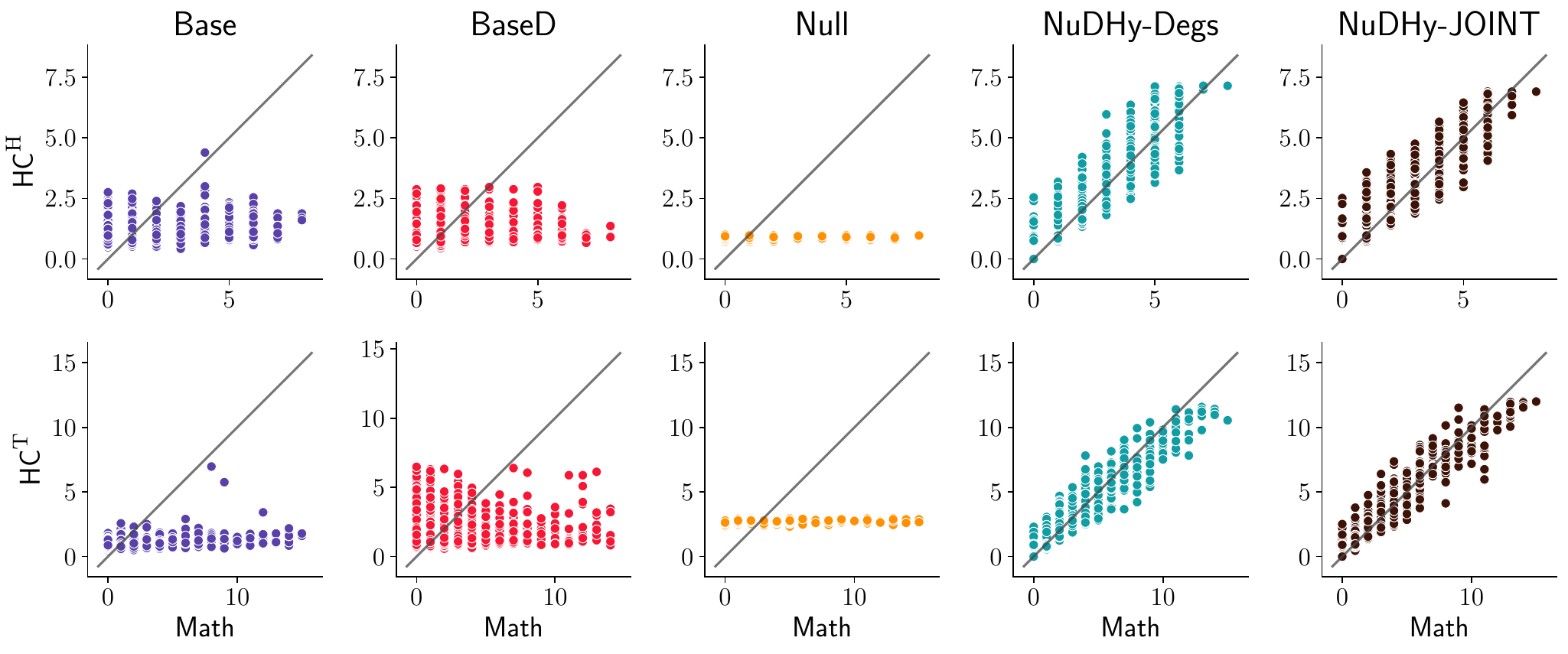}
    \end{subfigure}
    \caption{Scatterplot of $\mathsf{HC}^\textsc{h}$ (top) and $\mathsf{HC}^\textsc{t}$ (bottom) scores measured in the observed hypergraph (x-axis) and averages in $33$ samples (y-axis), for each sampler.}
    \label{fig:scatter_hyper}
\end{figure*}

\Cref{fig:scatter_hyper} compares the observed $\mathsf{HC}^\textsc{h}$ (top) and $\mathsf{HC}^\textsc{t}$ (bottom) scores with the mean scores in the samples, for each sampler. We observe a distinct trend for \algoA and \algoC.
The distribution of mean scores for \algoC more closely follows the distribution of observed scores, while the mean scores for \algoA tend to be higher for nodes with low observed scores and lower for nodes with high observed scores. 
This suggests that \algoA disrupts higher-order relations more than \algoC does. 
Since \nullm populates the hyperedges with nodes chosen uniformly at random, it leads to similar hyper-coreness values for all nodes. 
Finally, \base and \based may generate a different number of hyperedges and, on average, maintain the degree distributions but not the degree sequences. Consequently, nodes with high degrees in the observed hypergraph may exhibit low degrees in the samples, and vice versa. This explains the consistent wide range of mean hyper-coreness values for each observed hyper-coreness value.

\begin{table}[!th]  
    \centering
    \caption{Mean and standard deviation (in parentheses) of the Normalized Discounted Cumulative Gain (nDCG) of the rankings of the vertices obtained according to the $\textsc{h}$-hyper-coreness and the $\textsc{t}$-hyper-coreness in $33$ samples, for each dataset and each sampler. To ensure a fair comparison, we did not normalize the values by $k_m^\textsc{h}$ and $k_m^\textsc{t}$ as in \Cref{eq:hyper-coreness}.}
    \label{tab:hyper}
  \resizebox{\linewidth}{!}{
      \begin{tabular}{lcccccccccc}
      \toprule
    & \multicolumn{2}{c}{\base} & \multicolumn{2}{c}{\based} & \multicolumn{2}{c}{\nullm} & \multicolumn{2}{c}{\algoA} & \multicolumn{2}{c}{\algoC}\\
    \midrule
    & $\mathsf{HC}^\textsc{h}$ & $\mathsf{HC}^\textsc{t}$ & $\mathsf{HC}^\textsc{h}$ & $\mathsf{HC}^\textsc{t}$ & $\mathsf{HC}^\textsc{h}$ & $\mathsf{HC}^\textsc{t}$ & $\mathsf{HC}^\textsc{h}$ & $\mathsf{HC}^\textsc{t}$ & $\mathsf{HC}^\textsc{h}$ & $\mathsf{HC}^\textsc{t}$\\
    \midrule
    \multirow{2}{*}{\textsc{ecoli}} & 0.844 & 0.786 & 0.873 & 0.815 & 0.809 & 0.743 & 0.969 & 0.944 & 0.970 & 0.949\\
    & (0.014) & (0.023) & (0.012) & (0.018) & (0.010) & (0.008) & (0.006) & (0.013) & (0.005) & (0.018)\\
    \cline{2-11}
    \multirow{2}{*}{\textsc{iaf1260b}} & 0.821 & 0.867 & 0.842 & 0.877 & 0.829 & 0.870 & 0.974 & 0.978 & 0.973 & 0.980 \\
    & (0.006) & (0.004) & (0.005) & (0.005) & (0.005) & (0.003) & (0.003) & (0.003) & (0.002) & (0.002)\\
    \cline{2-11}
    \multirow{2}{*}{\textsc{ijo1366}} &  0.837 & 0.872 & 0.851 & 0.869 & 0.833 & 0.872 & 0.976 & 0.980 & 0.976 & 0.982\\
    & (0.005) & (0.004) & (0.002) & (0.003) & (0.005) & (0.003) & (0.002) & (0.002) & (0.003) & (0.002)\\
    \cline{2-11}
    \multirow{2}{*}{\textsc{cit-sw}} & 0.857 & 0.887 & 0.892 & 0.907 & 0.841 & 0.875 & 0.984 & 0.984 & 0.987 & 0.986\\
    & (0.003) & (0.002) & (0.002) & (0.001) & (0.002) & (0.001) & (0.001) & (0.001) & (0.001) & (0.001)\\
    \cline{2-11}
    \multirow{2}{*}{\textsc{math}} & 0.872 & 0.798 & 0.863 & 0.764 & 0.862 & 0.763 & 0.986 & 0.987 & 0.986 & 0.988\\
    & (0.001) & (0.002) & (0.001) & (0.001) & (0.001) & (0.002) & (0.001) & (0.001) & (0.001) & (0.001)\\
    \cline{2-11}
    \multirow{2}{*}{\textsc{dblp-9}}  & 0.857 & 0.848 & 0.898 & 0.892 & 0.838 & 0.832 & 0.984 & 0.984 & 0.993 & 0.986\\
    & (0.002) & (0.002) & (0.002) & (0.001) & (0.001) & (0.002) & (0.001) & (0.001) & (0.000) & (0.000)\\
    \cline{2-11}
    \multirow{2}{*}{\textsc{enron}} & 0.773 & 0.822 & 0.789 & 0.899 & 0.732 & 0.796 & 0.973 & 0.987 & 0.971 & 0.988\\
    & (0.004) & (0.002) & (0.005) & (0.001) & (0.003) & (0.002) & (0.001) & (0.001) & (0.001) & (0.001)\\
    \cline{2-11}
    \multirow{2}{*}{\textsc{ord}} & $-$ & $-$ & 0.915 & 0.952 & 0.886 & 0.951 & 0.992 & 0.986 &  0.994 & 0.988\\
    & $-$ & $-$ & (0.000) & (0.000) & (0.001) & (0.000) & (0.000) & (0.000) & (0.000) & (0.000)\\
    \bottomrule
    \end{tabular}
  }
\end{table}

While variations exist in the individual hyper-coreness values, \Cref{tab:hyper} reveals that the Normalized Discounted Cumulative Gain (nDCG) for ranking the vertices based on mean scores compared to the rankings derived from the observed scores is consistently high across all samplers. 
In particular, both \algoA and \algoC exhibit nearly perfect scores, signifying their ability to uphold the positions of the top-ranked nodes. 
In contrast, \nullm generally yields the lowest gains in nDCG, primarily due to the preservation of a smaller number of properties from the observed hypergraph.

\subsection{Structural Entropy}

We capture the probabilistic structure of a hypergraph $H$ at different granularity levels, by measuring the probability that a set of nodes of a certain size that is in the head (resp. tail) of some hyperedge in $H$ is also found in the head (resp. tail) of a hyperedge in a randomly sampled hypergraph from our null models $\nameA$ and $\nameC$.
The following approach is adapted from~\cite{young2021hypergraph}.
For each $j \in \{1,2\}$, this probability can be expressed as follows:
\[
\mathbb{P}(S|\nullset_j) = \frac{1}{n}\sum_{i=1}^n X_S(H_i)\,,
\]
where each $H_i$ is a hypergraph sampled from $\nullset_j$, and $X_S(H_i)$ takes value $1$ if the set of nodes $S$ appears in the head (resp. tail) of a hyperedge in $H_i$, and $0$ otherwise.
To quantify the entropy of these probabilities, denoted as $\hat{p}$, we use the following expression:
\begin{equation}\label{eq:entropy}
S(\hat{p}) = -\hat{p}\log(\hat{p}) - (1-\hat{p})\log(1-\hat{p})\,.
\end{equation}
This measure quantifies the variability of the hyperedges across the sampled hypergraphs, as it grows as $\hat{p}$ moves away from the extreme values $0$ (the node group appears in no sample) and $1$ (the node group appears in each sample).

\begin{figure}[!ht]
    \centering
    \begin{subfigure}{\linewidth}
      \centering
      \includegraphics[width=\linewidth]{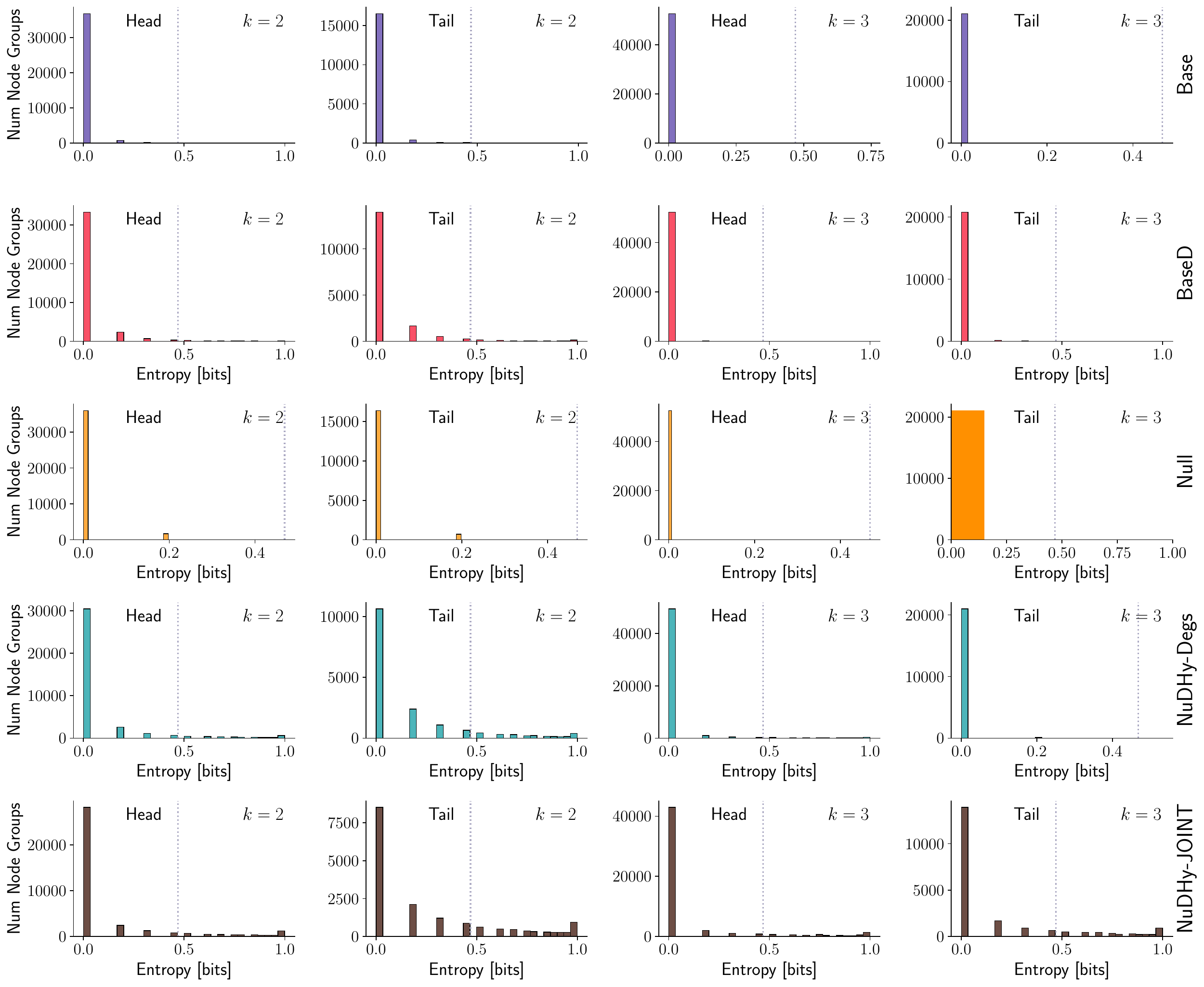}
    \end{subfigure}
    \caption{\textsc{dblp-9}: entropy of the head and tail node group probabilities, for groups of size $k=2$ and $k=3$, for each sampler.}
    \label{fig:entropy_dblp}
\end{figure}

To assess the degree of randomness in real-world hypergraph interaction processes, we investigate the likelihood that groups of nodes, engaging in interactions within the observed hypergraph (either in heads or tails), also participate in similar interactions within samples generated by the various samplers.
\Cref{fig:entropy_dblp} shows the distribution of entropy values for these probabilities (\Cref{eq:entropy}), for groups of sizes $k=2$ and $k=3$ in the heads (only for \textsc{dblp-9}) and in the tails. 
We define a \emph{certainty} threshold at $0.1$ (i.e, 10\% of the samples) and classify as uncertain all the groups with entropy greater than $S(0.1) \approxeq 0.469$ (green dotted line in the plots). 
As we can see, the majority of groups exhibit a high degree of certainty (i.e., they are either almost always present or absent in the samples), especially in the tails. 
This lower uncertainty may be attributed to the fact that, in most samples and for various group sizes, the probability of observing a specific group in the tail is close to zero, resulting in low entropy values. This trend is particularly prominent for the samples generated by \base and \based, because they struggle to accurately preserve the appearance of larger-sized groups.

\Cref{fig:probs_ord} further highlights that associations involving larger group sizes are disrupted by the samplers.
This figure shows the empirical cumulative distribution (ECDF) of the head node group probabilities, for groups of size $k=2$ (left charts) and $k=3$ (right charts), for \textsc{ecoli} and \textsc{ord}.
\nullm does not preserve any group of size $3$ and preserves the lowest amount of groups of size $2$.
\base follows a similar behavior, achieving slightly higher node group probabilities in the heads.
Among the competitors, \based is the most performing, as it preserves, on average, the number of appearances of each node group in the samples.
Finally, \algoC is the only sampler able to generate random hypergraphs preserving the appearance of groups of size $3$ with high probability.
The fact that \algoA struggles to preserve such appearances indicates the presence of a preference mechanism underlying the network, rather than these associations being mere consequences of node degrees.

\begin{figure}[!th]
    \centering
    \begin{subfigure}{.51\linewidth}
      \centering
      \includegraphics[width=\linewidth]{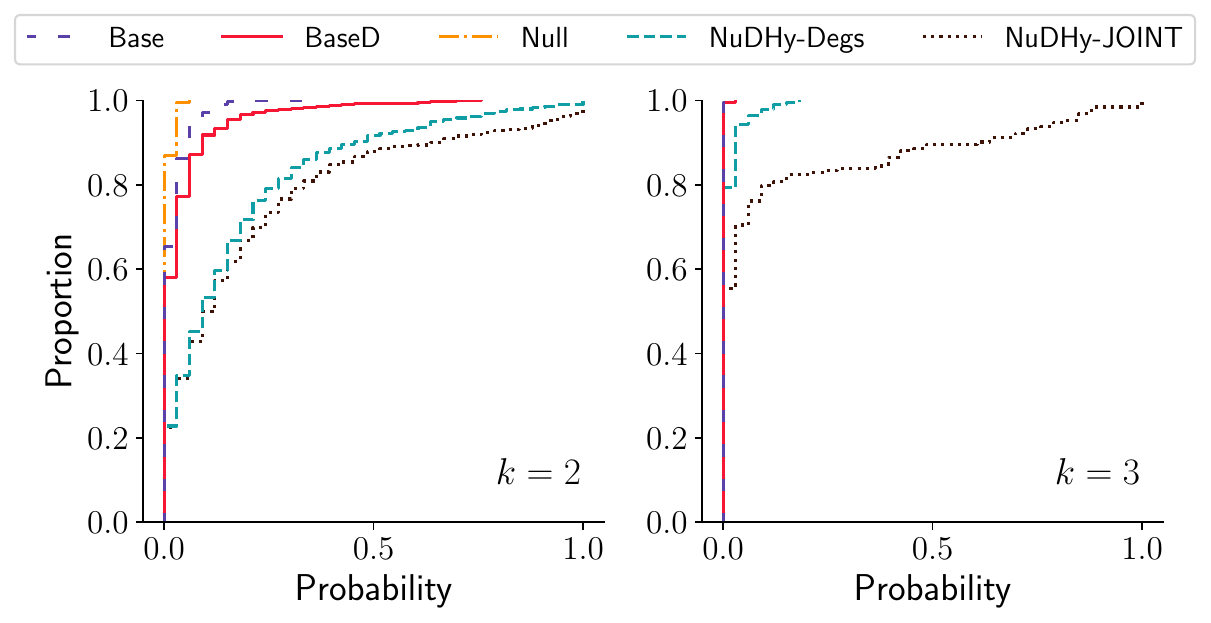}
    \end{subfigure}
    \hfill
    \begin{subfigure}{.48\linewidth}
      \centering
      \includegraphics[width=\linewidth]{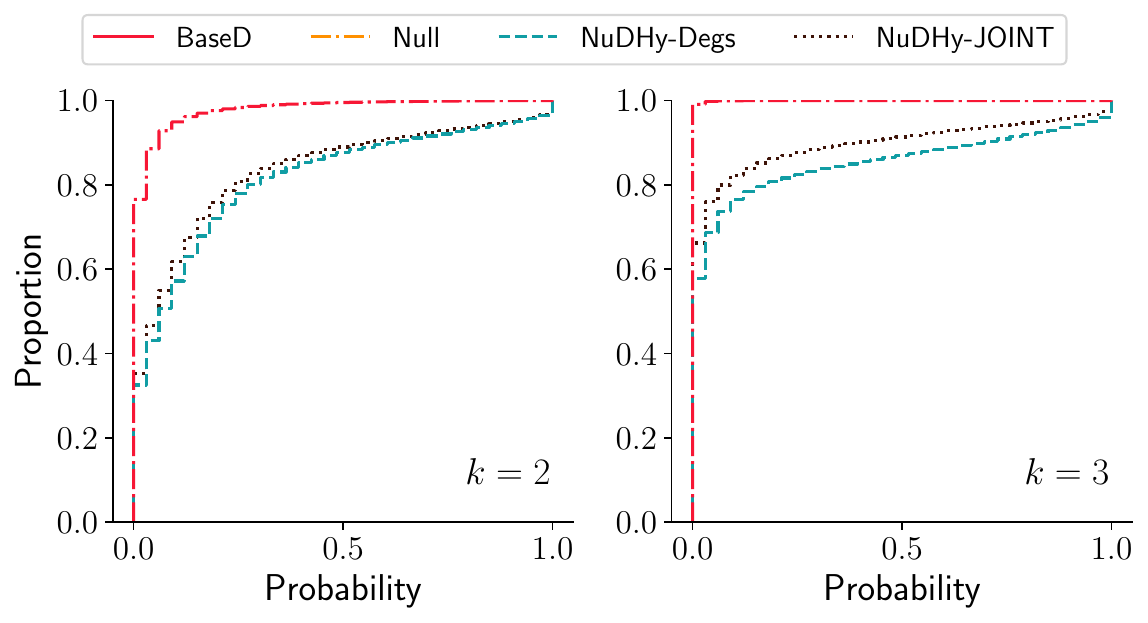}
    \end{subfigure}
    \caption{ECDF of the heads node group probabilities, for groups of size $k=2$ (left chart) and $k=3$ (right chart), for each sampler, for \textsc{ecoli} (left figure) and \textsc{ord} (right figure).}
    \label{fig:probs_ord}
\end{figure}

\subsection{Multi-order Spectrum}

Let $\hat{H} = (V, U)$ be the undirected hypergraph obtained by merging the head and tail of each directed hyperedge in the directed hypergraph $H \doteq (V, E)$. 
The \emph{generalized Laplacian of order $d$} of $\hat{H}$ is $L^{(d)} \doteq d K^{(d)} - A^{(d)}$, where $K^{(d)}$ is the degree matrix of order $d$, i.e., a diagonal matrix where each entry $K^{(d)}[v,v]$ indicates the number of hyperedges in $U$ of dimension $d$ including $v$; and $A^{(d)}$ is a symmetric matrix indicating, for each $(u,v)$, the number of hyperedges in $U$ of size $d$ containing both $u$ and $v$.
The \emph{multi-order Laplacian} of $\hat{H}$~\cite{lucas2020multiorder} is defined as the weighted sum of the Laplacian matrices of order $d$, i.e.,
\[
L^{(\mathrm{mul})} \doteq \sum_{d=2}^D \frac{\omega(d)}{\langle K^{(d)} \rangle}L^{(d)}\,,
\]
where $\omega(d)$ is a function to account for the significance of the different orders, $\langle K^{(d)} \rangle$ indicates the average degree of order $d$, and $D \in [1,|V|]$.
The \emph{spectrum} of $L^{(\mathrm{mul})}$ is the list of its eigenvalues $\Lambda = [\Lambda_1, \ldots, \Lambda_n]$, sorted in descending order.
We note that in this analysis, we do not consider the Laplacian matrix of the directed hypergraph $H$ as, in general, it is not symmetric, and thus, the eigenvalues may not be all real numbers.

The spectrum plays an important role in (hyper)network analysis, as it reveals information about the connectivity and robustness of the corresponding hyper(graph).
Similarity between two spectra implies similarity in structural properties and connectivity of the corresponding hypergraphs.

Given two hypergraphs $H_1$ and $H_2$, we extracted their top-$k$ smallest eigenvalues $\Lambda^{1,k}$ and $\Lambda^{2,k}$ and measured the \emph{spectral distance} between $H_1$ and $H_2$ as follows~\cite{jovanovic2012spectral,gu2015spectral}:
\[
\mathsf{sd}(H_1, H_2) \doteq \frac{1}{k}||\Lambda^{1,k} - \Lambda^{2,k}||_2\,.
\]

Low distance values indicate that the two hypergraphs are structurally similar, i.e., they are characterized by similar structural patters, local connectivity, and community structure.

We limited our analysis to the top-$6$ eigenvalues due to the computational complexity of the eigen-decomposition process, and we set $D$ equal to the minimum between the $8$ and the maximum dimension of a hyperedge in the original hypergraph.

\begin{figure}[!th]
    \centering
    \begin{subfigure}{\linewidth}
      \centering
      \includegraphics[width=\linewidth]{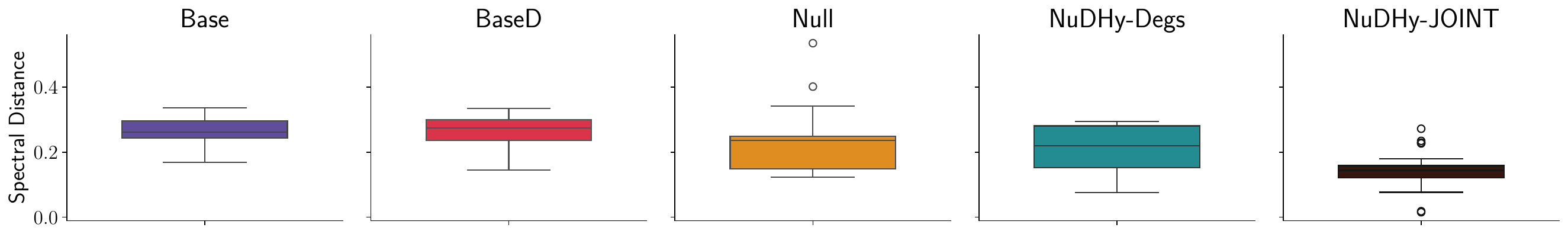}
      \caption{\textsc{ecoli}}
    \end{subfigure}
    \begin{subfigure}{\linewidth}
      \centering
      \includegraphics[width=\linewidth]{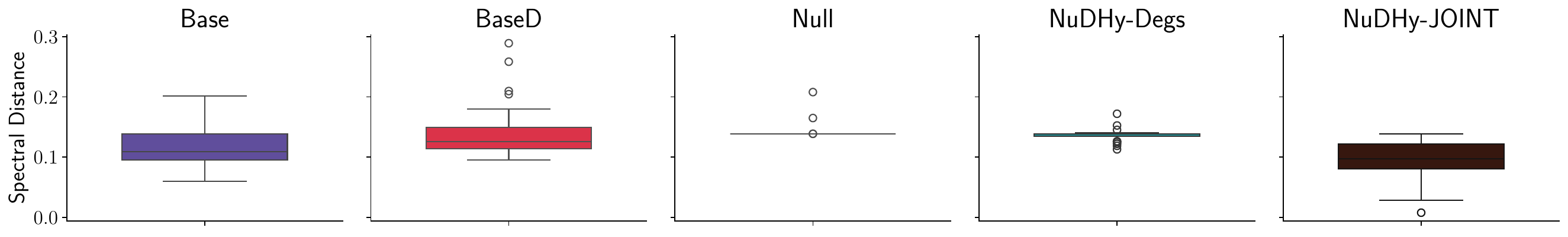}
      \caption{\textsc{ijo1366}}
    \end{subfigure}
    \begin{subfigure}{\linewidth}
      \centering
      \includegraphics[width=\linewidth]{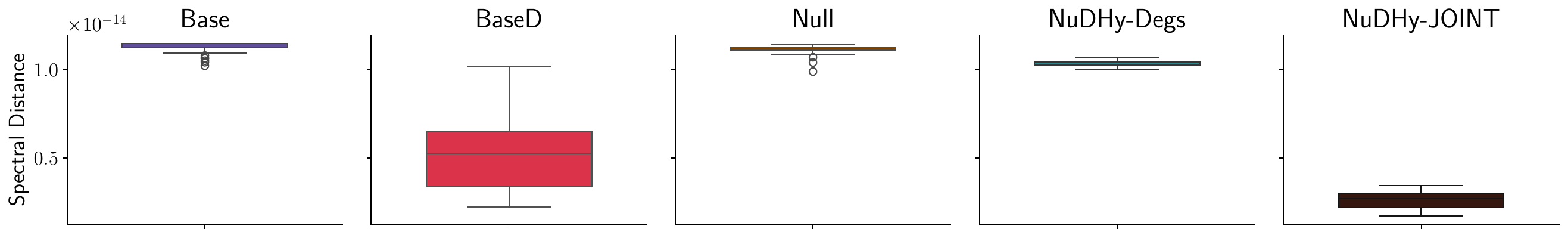}
      \caption{\textsc{enron}}
    \end{subfigure}
    \begin{subfigure}{\linewidth}
      \centering
      \includegraphics[width=\linewidth]{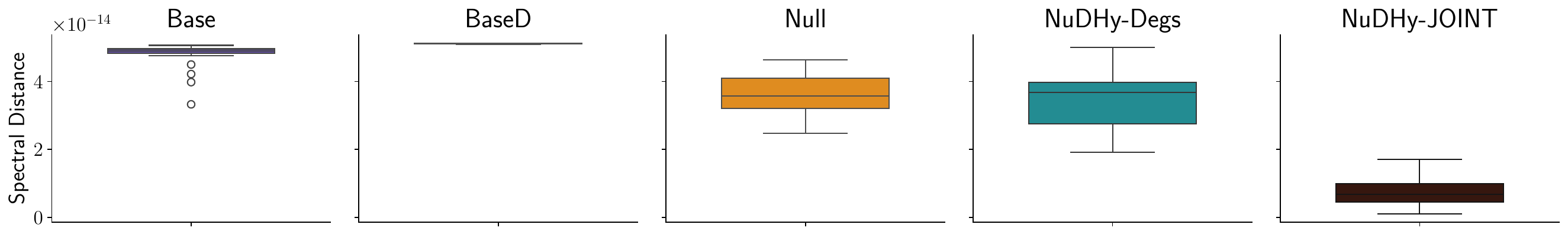}
      \caption{\textsc{math}}
    \end{subfigure}
    \caption{Distribution of the spectral distances between the spectrum of the observed hypergraph and that of $33$ samples, for each sampler, for \textsc{ecoli}, \textsc{ijo1366}, \textsc{enron}, and \textsc{math}.}
    \label{fig:lapl_spectral}
\end{figure}

\Cref{fig:lapl_spectral} shows the distribution of spectral distances across $33$ samples for each sampler, for four datasets.
We observe that \algoC consistently exhibit stable performance across all datasets, resulting in smaller distances between the eigenvalue vectors of the multi-order Laplacian, compared to those observed for \base, \based, and \nullm.
On the other hand, \algoA often achieves distances comparable to those measured for the samples generated by \nullm, indicating that the degree distribution alone does not preserve the community structure of the original hypergraph.

Interestingly, \based, which maintains node degrees on average, produces samples exhibiting spectral distance from the original hypergraph more akin to that observed in samples generated by \base (which preserves group degrees), rather than those generated by \algoA, despite \algoA exactly preserving the same properties. This difference may be attributed to the constraint imposed by the design of \base and \based, limiting the size of each hyperedge to $10$.

\subsection{Party Homophily in the US Congress}\label{ax:homophily}

\Cref{fig:homophily_s_h} shows the extent to which Republicans' and Democrats' observed homophily exceeds what would be expected if co-sponsoring relations formed randomly among legislators.
Consistently with prior research~\cite{grossmann2015ideological}, we observe that both parties exhibit an inclination toward associating with similar party members in co-sponsorship relations.
Then, we note that the inverse relationship between the curves of Republicans and Democrats, akin to the affinity plots, remains discernible, a pattern that remains unnoticed when solely examining the values of $\mathsf{m}$ measured in the observed hypergraphs.

\begin{figure*}[h!]
    \begin{subfigure}{.487\linewidth}
      \centering
      \includegraphics[width=\linewidth]{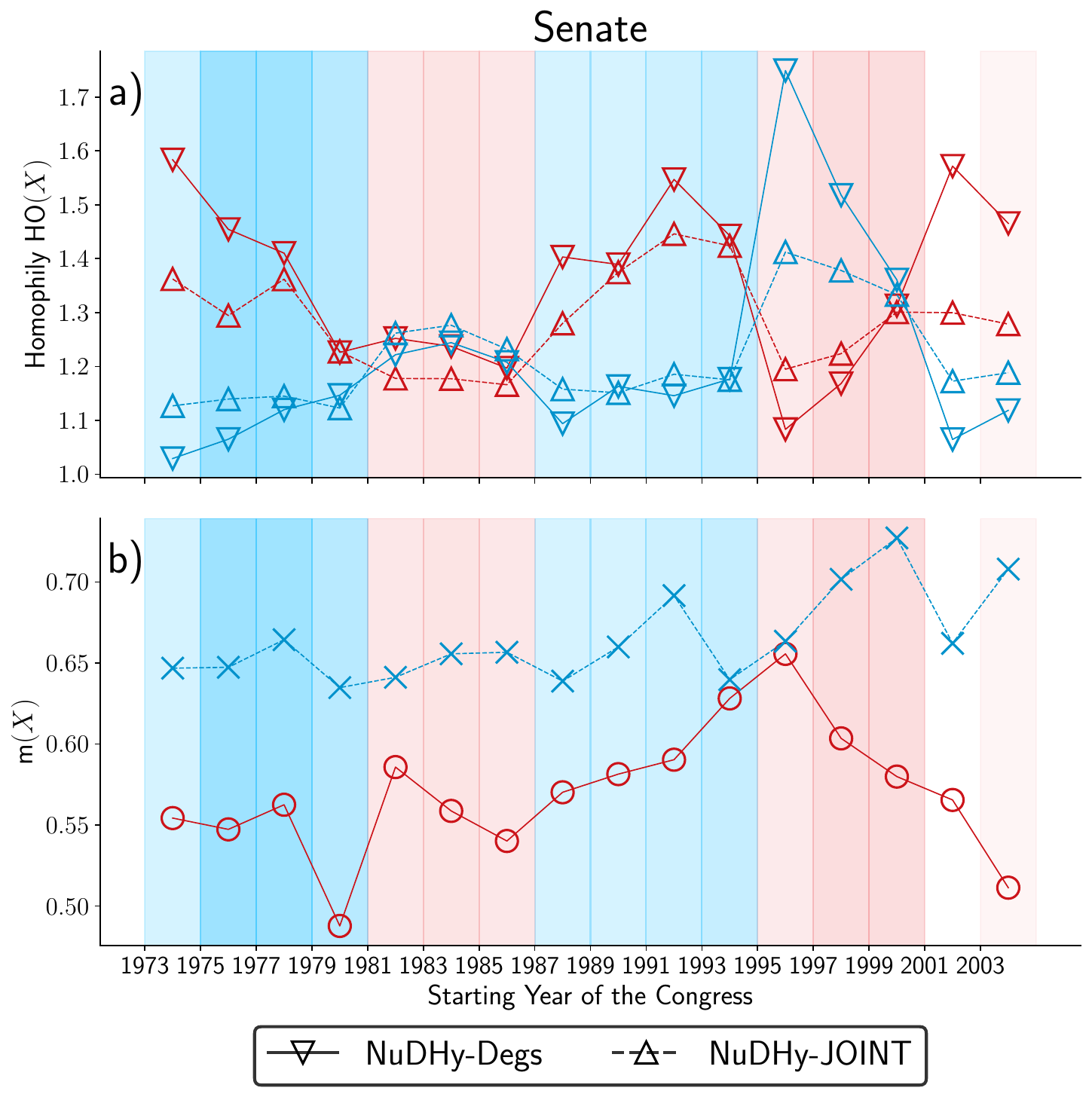}
    \end{subfigure}
    \begin{subfigure}{.495\linewidth}
      \centering
      \includegraphics[width=\linewidth]{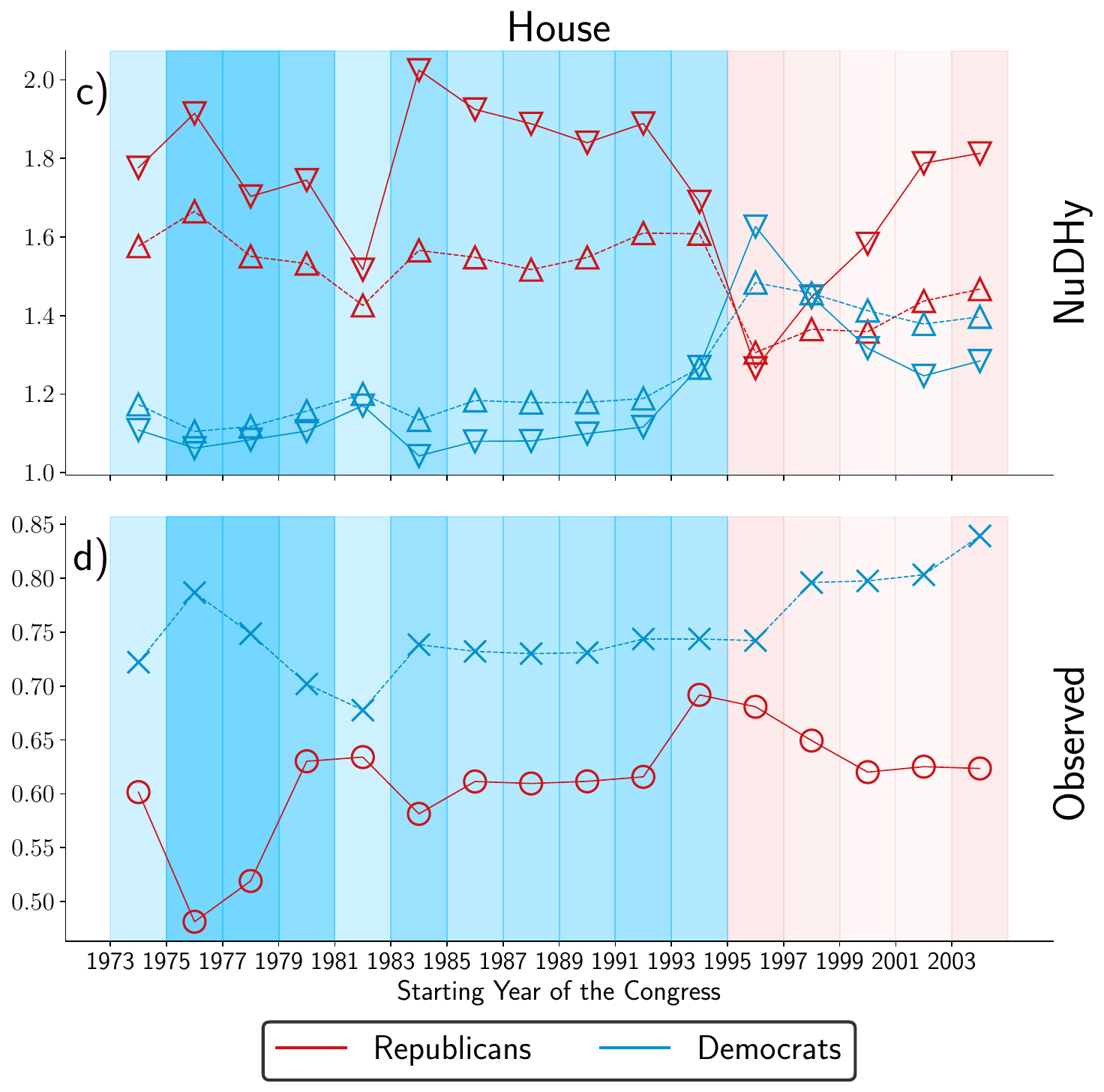}
    \end{subfigure}
    \caption{\textbf{Homophily in the US Congress co-sponsored bills.}
    We show results for \Cref{eq:homo} using $33$ samples generated by \algoA and \algoC for the US Senate (\textsc{S-bills)}, panel (a)) and House (\textsc{H-bills)}, panel (c)). 
    For comparison, we show the values of $\mathsf{m}(X)$ again for the US Senate (panel (b)) and House (panel (d)). 
    The colors indicate \emph{Democrats} (blue) and \emph{Republicans} (red).}
    \label{fig:homophily_s_h}
\end{figure*}

\subsection{Econometric Competitiveness: Analysis for Other Trade Networks}\label{ax:eco_other}

\Cref{tbl:ranks_genepy_eci_2} reports the mean Spearman's correlation and Kendall's Tau of the rankings of the countries obtained according to the GENEPY/ECI scores in $33$ samples generated by each sampler with respect to the rankings obtained from the observed GENEPY/ECI scores, for \textsc{hs1995}, \textsc{hs2009},and \textsc{hs2020}. 
The performances of the samplers are consistent through the various years, with \algoA and \algoC achieving the most correlated rankings with respect to the GENEPY scores. Regarding ECI, only \algoC leads to rankings that resemble that obtained from the observed scores.

\begin{table}[!ht]
  \caption{Mean Spearman's correlation and Kendall's Tau (KT) of the rankings of the countries obtained according to the ECI/Fitness/GENEPY scores in $33$ samples with respect to the rankings obtained from the observed scores. Standard deviations are reported in parentheses.}
  \label{tbl:ranks_genepy_eci_2}
  \resizebox{\columnwidth}{!}{
  \begin{tabular}{llcccccc}
  \textbf{Score} & \textbf{Metric} & \textbf{Year} & \base & \based & \nullm & \algoA & \algoC \\
  \midrule
  \multirow{6}{*}{ECI} & \multirow{3}{*}{Spearman} 
    & 1995 & -0.183 (0.143)  & -0.062 (0.118) & 0.011 (0.071) & 0.031 (0.112) & 0.947 (0.004)\\
  & & 2009 & -0.251 (0.143) & -0.110 (0.096) & -0.008 (0.079) & 0.009 (0.094) & 0.953 (0.002)\\
  & & 2020 & -0.223 (0.161) & -0.070 (0.135) & 0.004 (0.102) & 0.009 (0.121) & 0.978 (0.001)\\
  \cmidrule{3-8}
  & \multirow{3}{*}{KT}
    & 1995 & -0.119 (0.094) & -0.042 (0.079) & 0.008 (0.047) & 0.022 (0.076) & 0.817 (0.005)\\
  & & 2009 & -0.167 (0.097) & -0.074 (0.064) & -0.005 (0.054) & -0.007 (0.065) & 0.817 (0.006)\\
  & & 2020 & -0.152 (0.109) & -0.047 (0.090) & 0.003 (0.068) & 0.005 (0.082) & 0.878 (0.004)\\
  \midrule
  \multirow{6}{*}{Fitness} & \multirow{3}{*}{Spearman} 
    & 1995 & 0.162 (0.093) & 0.314 (0.057) & -0.001 (0.082) & 0.981 (0.001) & 0.997 (0.000)\\
  & & 2009 & 0.117 (0.086) & 0.272 (0.051) & 0.036 (0.076) & 0.979 (0.001) & 0.999 (0.000)\\
  & & 2020 & 0.137 (0.082) & 0.360 (0.054) & 0.020 (0.097) & 0.977 (0.001) & 0.998 (0.000)\\
  \cmidrule{3-8}
  & \multirow{3}{*}{KT}
    & 1995 & 0.110 (0.062) & 0.215 (0.042) & 0.000 (0.054) & 0.882 (0.004) & 0.962 (0.002)\\
  & & 2009 & 0.078 (0.059) & 0.182 (0.036) & 0.025 (0.050) & 0.884 (0.003) & 0.963 (0.001)\\
  & & 2020 & 0.091 (0.056) & 0.245 (0.037) & 0.014 (0.066) & 0.881 (0.003) & 0.964 (0.001) \\
  \midrule
  \multirow{6}{*}{GENEPY} & \multirow{3}{*}{Spearman} 
    & 1995 & 0.110 (0.101) & 0.312 (0.057) & 0.002 (0.100) & 0.953 (0.005) & 0.995 (0.000)\\
  & & 2009 & 0.083 (0.083) & 0.264 (0.050) & 0.034 (0.088) & 0.931 (0.004) & 0.994 (0.000)\\
  & & 2020 & 0.099 (0.078) & 0.348 (0.054) & 0.007 (0.097) & 0.941 (0.004) & 0.994 (0.000)\\
  \cmidrule{3-8}
  & \multirow{3}{*}{KT} 
    & 1995 & 0.074 (0.069) & 0.215 (0.041) & 0.002 (0.067) & 0.819 (0.009) & 0.944 (0.002)\\
  & & 2009 & 0.057 (0.056) & 0.178 (0.036) & 0.023 (0.060) & 0.790 (0.007) & 0.941 (0.002)\\
  & & 2020 & 0.065 (0.053) & 0.237 (0.038) & 0.004 (0.064) & 0.808 (0.007) & 0.942 (0.002)\\
  \bottomrule
  \end{tabular}
  }
\end{table}

\Cref{fig:avg_1995,fig:avg_2009,fig:avg_2020} illustrate the distribution of the rankings based on ECI/Fitness/GENEPY values across $33$ samples for \algoA (left) and \algoC (right) in comparison to the observed rankings for \textsc{hs1995}, \textsc{hs2009}, and \textsc{hs2020}.
In all cases, for \algoC, the distribution of rankings based on GENEPY closely mirrors the observed ranking, with noticeable deviations for the countries positioned in the middle of the observed ranking.
Conversely, for \algoA, deviations are more pronounced in the first positions of the rankings, but the overall trend aligns with the observed ranking.

Examining the rankings according to ECI reveals a different pattern. \algoC consistently produces distributions that closely align with the observed rankings. In contrast, \algoA generates diverse ranks for countries, resulting in mean rankings significantly different from the observed ones.

These findings suggest that the joint degree distribution significantly affects the relative GENEPY scores, whereas other underlying patterns in the data contribute to the relative ECI scores.
The degree distributions, which are weaker constraints than the joint distribution, fail to capture either the relative GENEPY or the relative ECI score distributions.

\begin{figure*}[!th]
    \centering
    \includegraphics[width=\linewidth]{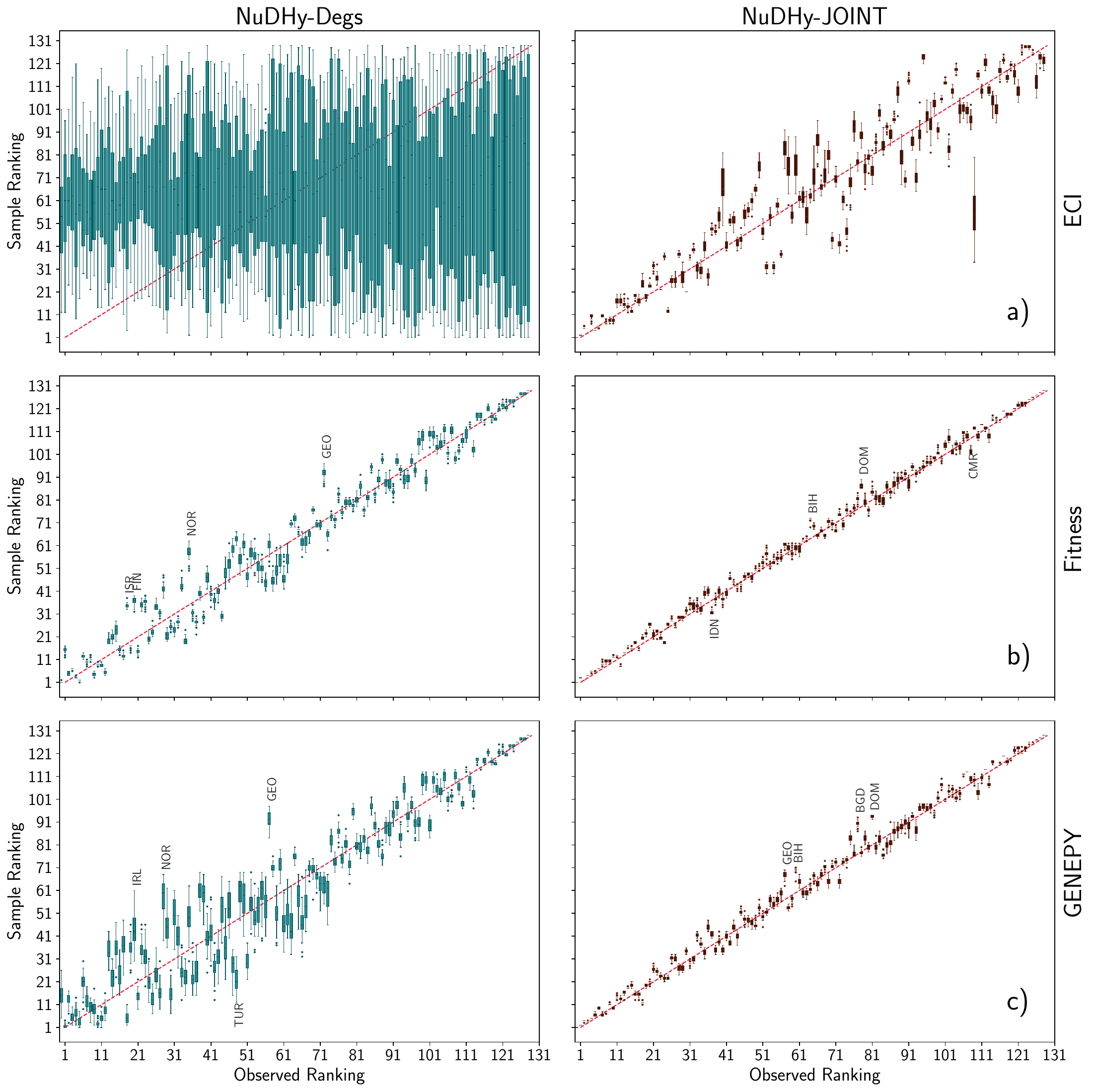}
    \caption{\textsc{hs1995}: Distributions of rankings obtained according to the ECI/Fitness/GENEPY values in $33$ samples for \algoA (left) and \algoC (right) compared to the observed rankings.
    The top 4 countries whose observed rank diverges the most from the sample mean ranking are annotated in a subset of the plots.}
    \label{fig:avg_1995}
\end{figure*}

\begin{figure*}[!th]
    \centering
    \includegraphics[width=\linewidth]{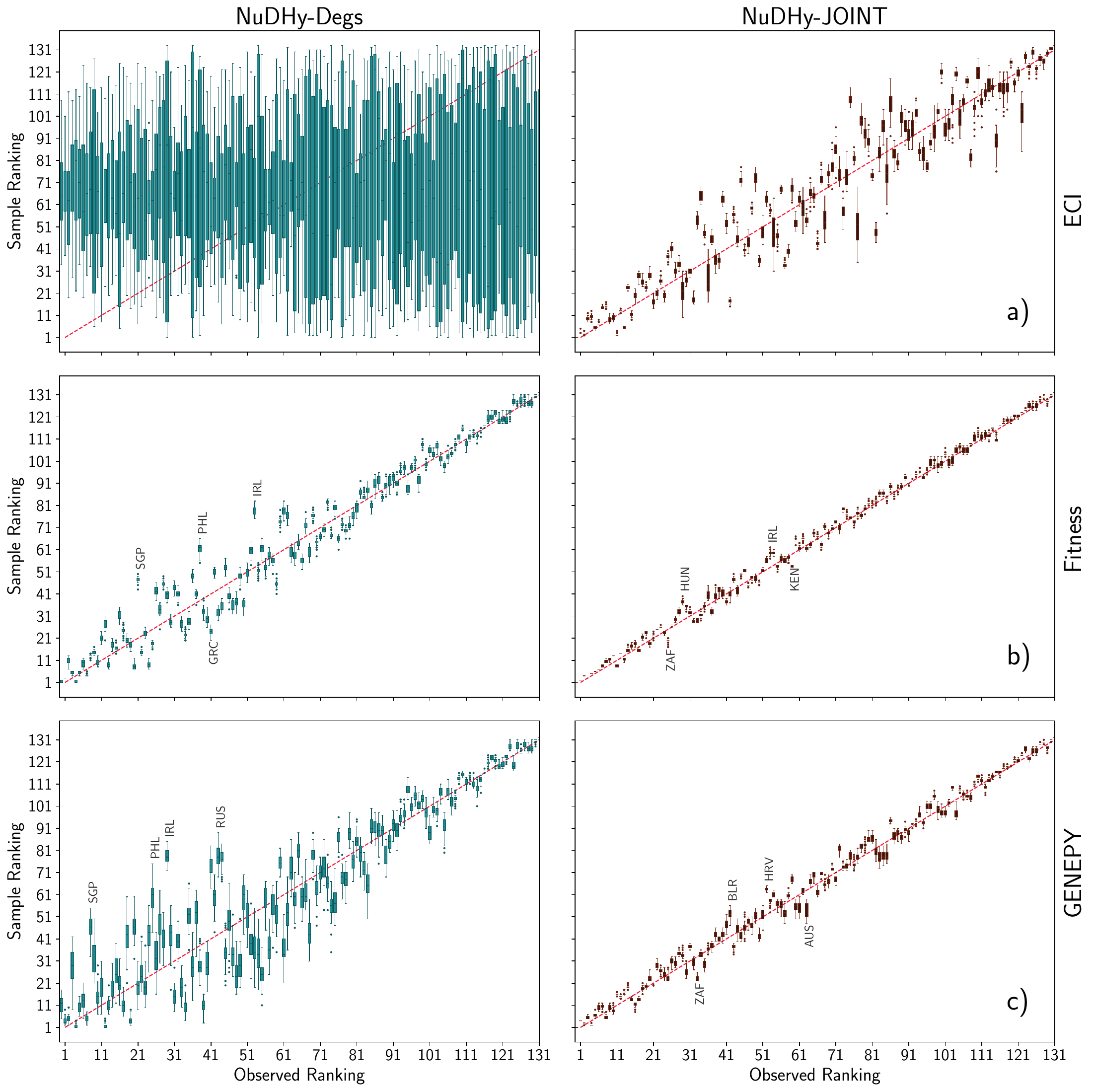}
    \caption{\textsc{hs2009}: Distributions of rankings obtained according to the ECI/Fitness/GENEPY values in $33$ samples for \algoA (left) and \algoC (right) compared to the observed rankings.
    The top 4 countries whose observed rank diverges the most from the sample mean ranking are annotated in a subset of the plots.}
    \label{fig:avg_2009}
\end{figure*}

\begin{figure*}[!th]
    \centering
    \includegraphics[width=\linewidth]{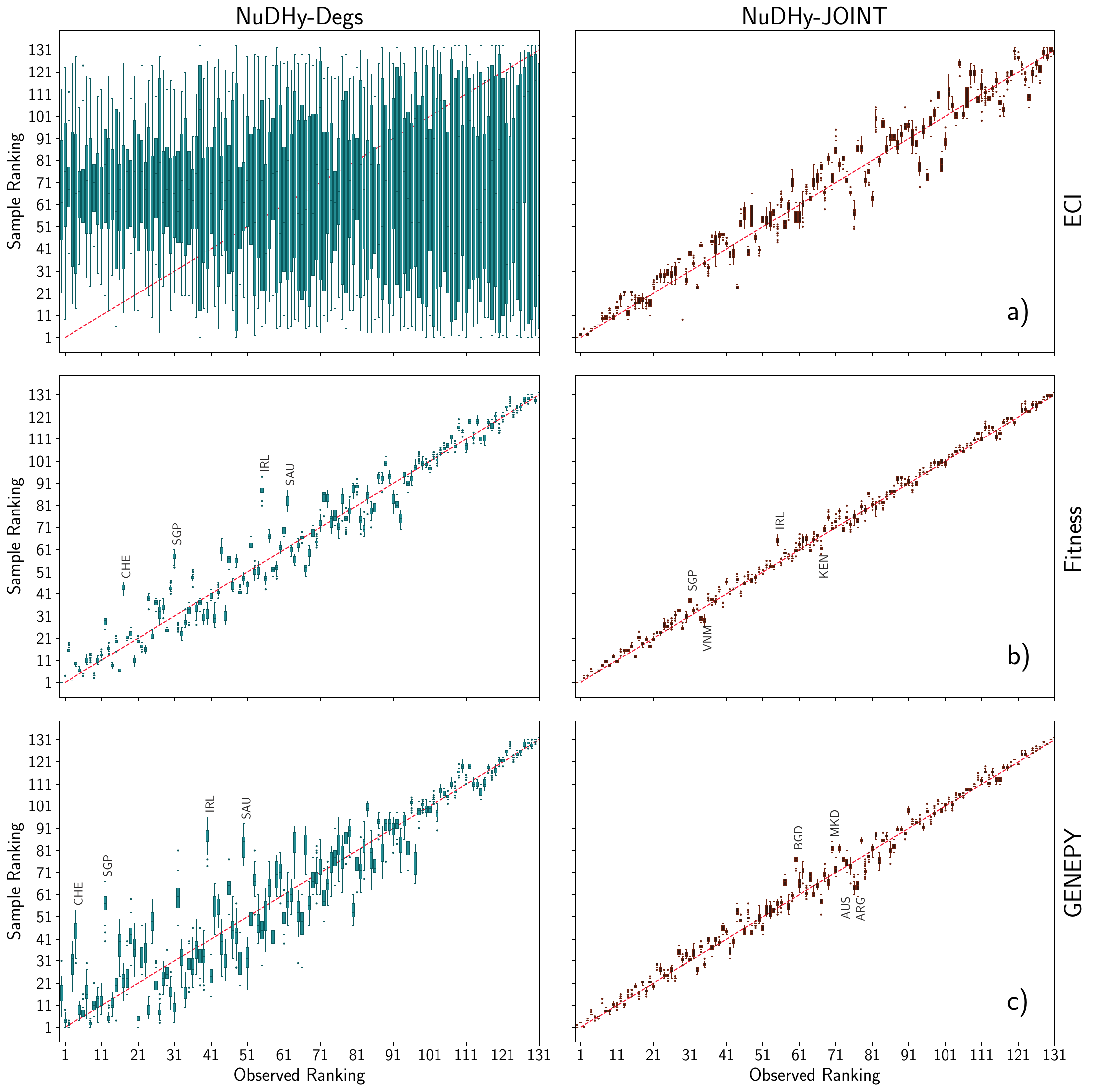}
    \caption{\textsc{hs2020}: Distributions of rankings obtained according to the ECI/Fitness/GENEPY values in $33$ samples for \algoA (left) and \algoC (right) compared to the observed rankings.
    The top 4 countries whose observed rank diverges the most from the sample mean ranking are annotated in a subset of the plots.}
    \label{fig:avg_2020}
\end{figure*}

%% file: main.bbl
\begin{thebibliography}{10}

\bibitem{battiston2020networks}
F.~Battiston, {\it et~al.\/}, Networks beyond pairwise interactions: Structure
  and dynamics, {\it Phys. Rep.\/} {\bf 874}, 1--92 (2020).

\bibitem{battiston2021physics}
F.~Battiston, {\it et~al.\/}, The physics of higher-order interactions in
  complex systems, {\it Nat. Phys.\/} {\bf 17}, 1093--1098 (2021).

\bibitem{bick2023higher}
C.~Bick, E.~Gross, H.~A. Harrington, M.~T. Schaub, What are higher-order
  networks?, {\it SIAM Review\/} {\bf 65}, 686--731 (2023).

\bibitem{ritz2015pathway}
A.~Ritz, B.~Avent, T.~Murali, Pathway analysis with signaling hypergraphs, {\it
  TCBB\/} {\bf 14}, 1042--1055 (2015).

\bibitem{feng2021hypergraph}
S.~Feng, {\it et~al.\/}, Hypergraph models of biological networks to identify
  genes critical to pathogenic viral response, {\it BMC bioinformatics\/} {\bf
  22}, 1--21 (2021).

\bibitem{schneidman2003network}
E.~Schneidman, S.~Still, M.~J. Berry, W.~Bialek, {\it et~al.\/}, Network
  information and connected correlations, {\it Phys. Rev. Lett.\/} {\bf 91},
  238701 (2003).

\bibitem{schneidman2006weak}
E.~Schneidman, M.~J. Berry~II, R.~Segev, W.~Bialek, Weak pairwise correlations
  imply strongly correlated network states in a neural population, {\it
  Nature\/} {\bf 440}, 1007 (2006).

\bibitem{giusti2015clique}
C.~Giusti, E.~Pastalkova, C.~Curto, V.~Itskov, Clique topology reveals
  intrinsic geometric structure in neural correlations, {\it Proc. Natl. Acad.
  Sci. U.S.A.\/} {\bf 112}, 13455--13460 (2015).

\bibitem{petri2014homological}
G.~Petri, {\it et~al.\/}, Homological scaffolds of brain functional networks,
  {\it J. R. Soc. Interface\/} {\bf 11}, 20140873 (2014).

\bibitem{patania2017shape}
A.~Patania, G.~Petri, F.~Vaccarino, The shape of collaborations, {\it EPJ Data
  Sci.\/} {\bf 6}, 18 (2017).

\bibitem{luo2022toward}
Q.~Luo, {\it et~al.\/}, Toward maintenance of hypercores in large-scale dynamic
  hypergraphs, {\it VLDBJ\/} pp. 1--18 (2022).

\bibitem{billings2019simplex2vec}
J.~C.~W. Billings, {\it et~al.\/}, Simplex2vec embeddings for community
  detection in simplicial complexes, {\it arXiv preprint arXiv:1906.09068\/}
  (2019).

\bibitem{luo2022directed}
X.~Luo, J.~Peng, J.~Liang, Directed hypergraph attention network for traffic
  forecasting, {\it IET Intelligent Transport Systems\/} {\bf 16}, 85--98
  (2022).

\bibitem{ranshous2017exchange}
S.~Ranshous, {\it et~al.\/}, {\it Financial Cryptography and Data Security\/}
  (2017), pp. 248--263.

\bibitem{berlt2010modeling}
K.~Berlt, {\it et~al.\/}, Modeling the web as a hypergraph to compute page
  reputation, {\it Information Systems\/} {\bf 35}, 530--543 (2010).

\bibitem{fisher1936design}
R.~A. Fisher, Design of experiments, {\it British Medical Journal\/} {\bf 1},
  554 (1936).

\bibitem{manly1995note}
B.~F. Manly, A note on the analysis of species co-occurrences, {\it Ecology\/}
  {\bf 76}, 1109--1115 (1995).

\bibitem{squartini2015unbiased}
T.~Squartini, R.~Mastrandrea, D.~Garlaschelli, Unbiased sampling of network
  ensembles, {\it New Journal of Physics\/} {\bf 17}, 023052 (2015).

\bibitem{schlauch2015influence}
W.~E. Schlauch, K.~A. Zweig, {\it Proceedings of the 2015 IEEE/ACM
  International Conference on Advances in Social Networks Analysis and Mining
  2015\/} (2015), pp. 514--519.

\bibitem{fischer2015sampling}
R.~Fischer, J.~C. Leitao, T.~P. Peixoto, E.~G. Altmann, Sampling
  motif-constrained ensembles of networks, {\it Physical review letters\/} {\bf
  115}, 188701 (2015).

\bibitem{verhelst2008efficient}
N.~D. Verhelst, An efficient mcmc algorithm to sample binary matrices with
  fixed marginals, {\it Psychometrika\/} {\bf 73}, 705--728 (2008).

\bibitem{strona2014fast}
G.~Strona, D.~Nappo, F.~Boccacci, S.~Fattorini, J.~San-Miguel-Ayanz, A fast and
  unbiased procedure to randomize ecological binary matrices with fixed row and
  column totals, {\it Nature communications\/} {\bf 5}, 4114 (2014).

\bibitem{stanton2012constructing}
I.~Stanton, A.~Pinar, Constructing and sampling graphs with a prescribed joint
  degree distribution, {\it Journal of Experimental Algorithmics\/} {\bf 17},
  3--1 (2012).

\bibitem{gjoka2015construction}
M.~Gjoka, B.~Tillman, A.~Markopoulou, {\it 2015 IEEE conference on computer
  communications\/} (2015), pp. 1553--1561.

\bibitem{cimini2019statistical}
G.~Cimini, {\it et~al.\/}, The statistical physics of real-world networks, {\it
  Nature Reviews Physics\/} {\bf 1}, 58--71 (2019).

\bibitem{kannan1999simple}
R.~Kannan, P.~Tetali, S.~Vempala, Simple markov-chain algorithms for generating
  bipartite graphs and tournaments, {\it Random Structures \& Algorithms\/}
  {\bf 14}, 293--308 (1999).

\bibitem{tabourier2011generating}
L.~Tabourier, C.~Roth, J.-P. Cointet, Generating constrained random graphs
  using multiple edge switches, {\it Journal of Experimental Algorithmics\/}
  {\bf 16}, 1--1 (2011).

\bibitem{saracco2015randomizing}
F.~Saracco, R.~Di~Clemente, A.~Gabrielli, T.~Squartini, Randomizing bipartite
  networks: the case of the world trade web, {\it Scientific Reports\/} {\bf
  5}, 10595 (2015).

\bibitem{boroojeni2017generating}
A.~A. Boroojeni, J.~Dewar, T.~Wu, J.~M. Hyman, Generating bipartite networks
  with a prescribed joint degree distribution, {\it Journal of complex
  networks\/} {\bf 5}, 839--857 (2017).

\bibitem{aksoy2017measuring}
S.~G. Aksoy, T.~G. Kolda, A.~Pinar, Measuring and modeling bipartite graphs
  with community structure, {\it Journal of Complex Networks\/} {\bf 5},
  581--603 (2017).

\bibitem{del2010efficient}
C.~I. Del~Genio, H.~Kim, Z.~Toroczkai, K.~E. Bassler, Efficient and exact
  sampling of simple graphs with given arbitrary degree sequence, {\it PloS
  one\/} {\bf 5}, e10012 (2010).

\bibitem{saracco2022entropy}
F.~Saracco, G.~Petri, R.~Lambiotte, T.~Squartini, Entropy-based random models
  for hypergraphs, {\it arXiv preprint arXiv:2207.12123\/}  (2022).

\bibitem{do2020structural}
M.~T. Do, S.-e. Yoon, B.~Hooi, K.~Shin, {\it SIGKDD\/} (2020), pp. 176--186.

\bibitem{barthelemy2022class}
M.~Barthelemy, Class of models for random hypergraphs, {\it Physical Review
  E\/} {\bf 106}, 064310 (2022).

\bibitem{guo2016non}
J.-L. Guo, X.-Y. Zhu, Q.~Suo, J.~Forrest, Non-uniform evolving hypergraphs and
  weighted evolving hypergraphs, {\it Scientific Reports\/} {\bf 6}, 36648
  (2016).

\bibitem{wang2010evolving}
J.-W. Wang, L.-L. Rong, Q.-H. Deng, J.-Y. Zhang, Evolving hypernetwork model,
  {\it The European Physical Journal B\/} {\bf 77}, 493--498 (2010).

\bibitem{chodrow2020configuration}
P.~S. Chodrow, Configuration models of random hypergraphs, {\it Journal of
  Complex Networks\/} {\bf 8}, cnaa018 (2020).

\bibitem{zeng2023hyper}
Y.~Zeng, B.~Liu, F.~Zhou, L.~L{\"u}, Hyper-null models and their applications,
  {\it Entropy\/} {\bf 25}, 1390 (2023).

\bibitem{sun2021higher}
H.~Sun, G.~Bianconi, Higher-order percolation processes on multiplex
  hypergraphs, {\it Physical Review E\/} {\bf 104}, 034306 (2021).

\bibitem{nakajima2021randomizing}
K.~Nakajima, K.~Shudo, N.~Masuda, Randomizing hypergraphs preserving degree
  correlation and local clustering, {\it IEEE Transactions on Network Science
  and Engineering\/} {\bf 9}, 1139--1153 (2021).

\bibitem{miyashita2023randomizing}
R.~Miyashita, K.~Nakajima, M.~Fukuda, K.~Shudo, {\it IEEE International
  Conference on Big Data and Smart Computing\/} (2023), pp. 316--317.

\bibitem{kim2022reciprocity}
S.~Kim, M.~Choe, J.~Yoo, K.~Shin, {\it ICDM\/} (2022).

\bibitem{hidalgo2007product}
C.~A. Hidalgo, B.~Klinger, A.-L. Barab{\'a}si, R.~Hausmann, The product space
  conditions the development of nations, {\it Science\/} {\bf 317}, 482--487
  (2007).

\bibitem{tacchella2012metrics}
A.~Tacchella, M.~Cristelli, G.~Caldarelli, A.~Gabrielli, L.~Pietronero, A new
  metrics for countries' fitness and products' complexity, {\it Scientific
  Reports\/} {\bf 2} (2012).

\bibitem{sciarra2020reconciling}
C.~Sciarra, G.~Chiarotti, L.~Ridolfi, F.~Laio, Reconciling contrasting views on
  economic complexity, {\it Nature communications\/} {\bf 11}, 3352 (2020).

\bibitem{chodrow2020annotated}
P.~Chodrow, A.~Mellor, Annotated hypergraphs: models and applications, {\it
  Applied network science\/} {\bf 5}, 1--25 (2020).

\bibitem{moody2001race}
J.~Moody, Race, school integration, and friendship segregation in america, {\it
  American journal of Sociology\/} {\bf 107}, 679--716 (2001).

\bibitem{loomis1946political}
C.~P. Loomis, Political and occupational cleavages in a hanoverian village,
  germany: A sociometric study, {\it Sociometry\/} {\bf 9}, 316--333 (1946).

\bibitem{lazarsfeld1954friendship}
P.~F. Lazarsfeld, R.~K. Merton, {\it et~al.\/}, Friendship as a social process:
  A substantive and methodological analysis, {\it Freedom and control in modern
  society\/} {\bf 18}, 18--66 (1954).

\bibitem{mcpherson2001birds}
M.~McPherson, L.~Smith-Lovin, J.~M. Cook, Birds of a feather: Homophily in
  social networks, {\it Annual review of sociology\/} {\bf 27}, 415--444
  (2001).

\bibitem{veldt2023combinatorial}
N.~Veldt, A.~R. Benson, J.~Kleinberg, Combinatorial characterizations and
  impossibilities for higher-order homophily, {\it Science Advances\/} {\bf 9},
  eabq3200 (2023).

\bibitem{fowler2006legislative}
J.~H. Fowler, Legislative cosponsorship networks in the us house and senate,
  {\it Social networks\/} {\bf 28}, 454--465 (2006).

\bibitem{park2007distribution}
J.~Park, A.-L. Barab{\'a}si, Distribution of node characteristics in complex
  networks, {\it Proceedings of the National Academy of Sciences\/} {\bf 104},
  17916--17920 (2007).

\bibitem{grossmann2015ideological}
M.~Grossmann, D.~A. Hopkins, Ideological republicans and group interest
  democrats: The asymmetry of american party politics, {\it Perspectives on
  Politics\/} {\bf 13}, 119--139 (2015).

\bibitem{neal2022homophily}
Z.~P. Neal, R.~Domagalski, X.~Yan, Homophily in collaborations among us house
  representatives, 1981--2018, {\it Social Networks\/} {\bf 68}, 97--106
  (2022).

\bibitem{peel2022statistical}
L.~Peel, T.~P. Peixoto, M.~De~Domenico, Statistical inference links data and
  theory in network science, {\it Nature Communications\/} {\bf 13}, 6794
  (2022). Number: 1 Publisher: Nature Publishing Group.

\bibitem{rottjers2021null}
L.~R{\"o}ttjers, D.~Vandeputte, J.~Raes, K.~Faust, Null-model-based network
  comparison reveals core associations, {\it ISME Communications\/} {\bf 1}, 36
  (2021).

\bibitem{monsted2017evidence}
B.~M{\o}nsted, P.~Sapie{\.z}y{\'n}ski, E.~Ferrara, S.~Lehmann, Evidence of
  complex contagion of information in social media: An experiment using twitter
  bots, {\it PloS one\/} {\bf 12}, e0184148 (2017).

\bibitem{karsai2014complex}
M.~Karsai, G.~Iniguez, K.~Kaski, J.~Kert{\'e}sz, Complex contagion process in
  spreading of online innovation, {\it Journal of The Royal Society
  Interface\/} {\bf 11}, 20140694 (2014).

\bibitem{iacopini2019simplicial}
I.~Iacopini, G.~Petri, A.~Barrat, V.~Latora, Simplicial models of social
  contagion, {\it Nature communications\/} {\bf 10}, 2485 (2019).

\bibitem{de2020social}
G.~F. de~Arruda, G.~Petri, Y.~Moreno, Social contagion models on hypergraphs,
  {\it Physical Review Research\/} {\bf 2}, 023032 (2020).

\bibitem{st2022influential}
G.~St-Onge, {\it et~al.\/}, Influential groups for seeding and sustaining
  nonlinear contagion in heterogeneous hypergraphs, {\it Communications
  Physics\/} {\bf 5}, 25 (2022).

\bibitem{cui2023general}
S.~Cui, F.~Liu, H.~Jard{\'o}n-Kojakhmetov, M.~Cao, General sis diffusion
  process with indirect spreading pathways on a hypergraph, {\it arXiv preprint
  arXiv:2306.00619\/}  (2023).

\bibitem{st2021universal}
G.~St-Onge, H.~Sun, A.~Allard, L.~H{\'e}bert-Dufresne, G.~Bianconi, Universal
  nonlinear infection kernel from heterogeneous exposure on higher-order
  networks, {\it Physical review letters\/} {\bf 127}, 158301 (2021).

\bibitem{gillespie2007stochastic}
D.~T. Gillespie, Stochastic {Simulation} of {Chemical} {Kinetics}, {\it Annual
  Review of Physical Chemistry\/} {\bf 58}, 35--55 (2007).

\bibitem{gemmetto2014mitigation}
V.~Gemmetto, A.~Barrat, C.~Cattuto, Mitigation of infectious disease at school:
  targeted class closure vs school closure, {\it BMC infectious diseases\/}
  {\bf 14}, 1--10 (2014).

\bibitem{mastrandrea2015contact}
R.~Mastrandrea, J.~Fournet, A.~Barrat, Contact patterns in a high school: a
  comparison between data collected using wearable sensors, contact diaries and
  friendship surveys, {\it PloS one\/} {\bf 10}, e0136497 (2015).

\bibitem{benson2018simplicial}
A.~R. Benson, R.~Abebe, M.~T. Schaub, A.~Jadbabaie, J.~Kleinberg, Simplicial
  closure and higher-order link prediction, {\it Proceedings of the National
  Academy of Sciences\/} {\bf 115}, E11221--E11230 (2018).

\bibitem{PhysRevE.71.016129}
M.~M. de~Oliveira, R.~Dickman, How to simulate the quasistationary state, {\it
  Phys. Rev. E\/} {\bf 71}, 016129 (2005).

\bibitem{tacchella2013economic}
A.~Tacchella, M.~Cristelli, G.~Caldarelli, A.~Gabrielli, L.~Pietronero,
  Economic complexity: Conceptual grounding of a new metrics for global
  competitiveness, {\it Journal of Economic Dynamics and Control\/} {\bf 37},
  1683--1691 (2013).

\bibitem{cristelli2013measuring}
M.~Cristelli, A.~Gabrielli, A.~Tacchella, G.~Caldarelli, L.~Pietronero,
  Measuring the intangibles: A metrics for the economic complexity of countries
  and products, {\it PloS one\/} {\bf 8}, e70726 (2013).

\bibitem{balassa1965trade}
B.~Balassa, Trade liberalisation and ``revealed'' comparative advantage 1, {\it
  The manchester school\/} {\bf 33}, 99--123 (1965).

\bibitem{harvard}
H.~G. Lab, Atlas complexity ranking. Https://atlas.cid.harvard.edu/rankings.

\bibitem{inoua2023simple}
S.~Inoua, A simple measure of economic complexity, {\it Research Policy\/} {\bf
  52}, 104793 (2023).

\bibitem{pugliese2016convergence}
E.~Pugliese, A.~Zaccaria, L.~Pietronero, On the convergence of the
  fitness-complexity algorithm, {\it The European Physical Journal Special
  Topics\/} {\bf 225}, 1893--1911 (2016).

\bibitem{hausmann2014atlas}
R.~Hausmann, C.~A. Hidalgo, S.~Bustos, M.~Coscia, A.~Simoes, {\it The atlas of
  economic complexity: Mapping paths to prosperity\/} (Mit Press, 2014).

\bibitem{straka2018ecology}
M.~J. Straka, G.~Caldarelli, T.~Squartini, F.~Saracco, From ecology to finance
  (and back?): A review on entropy-based null models for the analysis of
  bipartite networks, {\it Journal of Statistical Physics\/} {\bf 173},
  1252--1285 (2018).

\bibitem{artzy2005generating}
Y.~Artzy-Randrup, L.~Stone, Generating uniformly distributed random networks,
  {\it Physical Review E\/} {\bf 72}, 056708 (2005).

\bibitem{rao2020principles}
A.~S.~S. Rao, C.~R. Rao, {\it Principles and methods for data science\/}
  (Elsevier, 2020).

\bibitem{viger2016efficient}
F.~Viger, M.~Latapy, Efficient and simple generation of random simple connected
  graphs with prescribed degree sequence, {\it Journal of Complex Networks\/}
  {\bf 4}, 15--37 (2016).

\bibitem{czabarka2015realizations}
{\'E}.~Czabarka, A.~Dutle, P.~L. Erd{\H{o}}s, I.~Mikl{\'o}s, On realizations of
  a joint degree matrix, {\it Discrete Applied Mathematics\/} {\bf 181},
  283--288 (2015).

\bibitem{ryser1963combinatorial}
H.~J. Ryser, {\it Combinatorial mathematics\/}, vol.~14 (American Mathematical
  Soc., 1963).

\bibitem{gionis2007assessing}
A.~Gionis, H.~Mannila, T.~Mielik{\"a}inen, P.~Tsaparas, Assessing data mining
  results via swap randomization, {\it TKDD\/} {\bf 1}, 14--es (2007).

\bibitem{krugman2009international}
P.~R. Krugman, M.~Obstfeld, {\it International economics: Theory and policy\/}
  (Pearson Education, 2009).

\bibitem{shen2018genome}
T.~Shen, {\it et~al.\/}, A genome-scale metabolic network alignment method
  within a hypergraph-based framework using a rotational tensor-vector product,
  {\it Scientific Reports\/} {\bf 8}, 1--16 (2018).

\bibitem{tang2008arnetminer}
J.~Tang, {\it et~al.\/}, {\it ACM SIGKDD\/} (2008), pp. 990--998.

\bibitem{kearnes2021open}
S.~M. Kearnes, {\it et~al.\/}, The open reaction database, {\it Journal of the
  American Chemical Society\/} {\bf 143}, 18820--18826 (2021).

\bibitem{newman2002email}
M.~E. Newman, S.~Forrest, J.~Balthrop, Email networks and the spread of
  computer viruses, {\it Physical Review E\/} {\bf 66}, 035101 (2002).

\bibitem{tran2019pagerank}
L.~Tran, T.~Quan, A.~Mai, Pagerank algorithm for directed hypergraph, {\it
  CoRR, abs/1909.01132\/}  (2019).

\bibitem{kleinberg1999authoritative}
J.~M. Kleinberg, Authoritative sources in a hyperlinked environment, {\it
  Journal of the ACM\/} {\bf 46}, 604--632 (1999).

\bibitem{jarvelin2002cumulated}
K.~J{\"a}rvelin, J.~Kek{\"a}l{\"a}inen, Cumulated gain-based evaluation of ir
  techniques, {\it ACM Transactions on Information Systems\/} {\bf 20},
  422--446 (2002).

\bibitem{garcia2017understanding}
D.~Garcia, P.~Mavrodiev, D.~Casati, F.~Schweitzer, Understanding popularity,
  reputation, and social influence in the twitter society, {\it Policy \&
  Internet\/} {\bf 9}, 343--364 (2017).

\bibitem{mancastroppa2023hyper}
M.~Mancastroppa, I.~Iacopini, G.~Petri, A.~Barrat, Hyper-cores promote
  localization and efficient seeding in higher-order processes, {\it arXiv
  preprint arXiv:2301.04235\/}  (2023).

\bibitem{giatsidis2013d}
C.~Giatsidis, D.~M. Thilikos, M.~Vazirgiannis, D-cores: measuring collaboration
  of directed graphs based on degeneracy, {\it Knowledge and information
  systems\/} {\bf 35}, 311--343 (2013).

\bibitem{young2021hypergraph}
J.-G. Young, G.~Petri, T.~P. Peixoto, Hypergraph reconstruction from network
  data, {\it Communications Physics\/} {\bf 4}, 135 (2021).

\bibitem{lucas2020multiorder}
M.~Lucas, G.~Cencetti, F.~Battiston, Multiorder laplacian for synchronization
  in higher-order networks, {\it Physical Review Research\/} {\bf 2}, 033410
  (2020).

\bibitem{jovanovic2012spectral}
I.~Jovanovi{\'c}, Z.~Stani{\'c}, Spectral distances of graphs, {\it Linear
  algebra and its applications\/} {\bf 436}, 1425--1435 (2012).

\bibitem{gu2015spectral}
J.~Gu, B.~Hua, S.~Liu, Spectral distances on graphs, {\it Discrete Applied
  Mathematics\/} {\bf 190}, 56--74 (2015).

\end{thebibliography}
